\pdfoutput=1
\documentclass[11pt]{article}
\usepackage[margin=1in]{geometry}

\usepackage{mathtools}
\usepackage{amssymb,amsthm}
\usepackage{thmtools}
\usepackage{thm-restate}		
\usepackage[varbb]{newpxmath}
\usepackage{nicefrac}
\usepackage[most]{tcolorbox}

\usepackage{graphicx}
\usepackage{enumerate}
\usepackage{enumitem}
\usepackage{bbm} 
\usepackage{appendix}
\usepackage{caption}
\RequirePackage{subfigure}

\usepackage{algorithm}
\usepackage{multicol}
\usepackage[framemethod=TikZ]{mdframed}
\mdfsetup{%
backgroundcolor=gray!10, , 
roundcorner=4pt,
linewidth=1pt}
\usepackage[noend]{algpseudocode}

\usepackage[dvipsnames,svgnames]{xcolor}
\usepackage{color}
\definecolor{niceRed}{RGB}{190,38,38}
\definecolor{niceYellow}{HTML}{f5b400}
\definecolor{blueGrotto}{HTML}{059DC0}
\definecolor{royalBlue}{HTML}{057DCD}
\definecolor{navyBlue}{HTML}{0B579C}
\definecolor{limeGreen}{HTML}{81B622}
\definecolor{nicePurple}{HTML}{9c27b0}
\definecolor{lightRoyalBlue}{HTML}{def2ff}  
\definecolor{gold}{HTML}{ffa300}
\colorlet{MyRed}{FireBrick!50!Crimson}
\colorlet{MyBlue}{DodgerBlue!75!black}
\colorlet{MyGreen}{DarkGreen!85!black}
\colorlet{MyViolet}{DarkMagenta}

\colorlet{MyLightBlue}{DodgerBlue!20}
\colorlet{MyLightGreen}{MyGreen!20}

\colorlet{PrimalColor}{MyBlue}
\colorlet{PrimalFill}{MyLightBlue}
\colorlet{DualColor}{MyRed}

\colorlet{RevColor}{MyRed}
\colorlet{LinkColor}{MediumBlue}

\usepackage{hyperref} 
\usepackage[capitalize,nameinlink]{cleveref}

\hypersetup{
  colorlinks = true, 
  urlcolor = {blueGrotto},
  linkcolor = {royalBlue},
  citecolor = {navyBlue}
}
\usepackage[
  backref=true,
  backend=biber,
  natbib=true,
  style=alphabetic,
  sorting=alphabeticlabel,
  sortcites=true,
  minbibnames=3,
  maxbibnames=999,
  mincitenames=4,
  maxcitenames=4,
  minalphanames=4,
  maxalphanames=4
]{biblatex}

\usepackage{enumitem}

\usepackage{tikz}
\usepackage{pgfplots}
\pgfplotsset{compat=1.18}
\usepgfplotslibrary{fillbetween}

\usepackage{wrapfig}

\DeclareSortingTemplate{alphabeticlabel}{
  \sort[final]{%
    \field{labelalpha} 
  }
  \sort{%
    \field{year}   
  }
  \sort{%
    \field{title}  
  }
}

\AtBeginRefsection{\GenRefcontextData{sorting=ynt}}
\AtEveryCite{\localrefcontext[sorting=ynt]}
\usepackage{algpseudocode}

\usepackage[suppress]{color-edits}

\addauthor[Shuchen]{sl}{blue}
\addauthor[Alkis]{ak}{niceRed}
\addauthor[Manolis]{mz}{limeGreen}
\addauthor[Pravesh]{pk}{nicePurple}

\usepackage{acronym}		

\usepackage[normalem]{ulem}		

\newcommand{\debug}[1]{{#1}}		

\newcommand{\newmacro}[2]{\newcommand{#1}{\debug{#2}}}		
\newcommand{\newop}[2]{\DeclareMathOperator{#1}{\debug{#2}}}		

\DeclarePairedDelimiter{\braces}{\{}{\}}		
\DeclarePairedDelimiter{\bracks}{[}{]}		
\DeclarePairedDelimiter{\parens}{(}{)}		

\DeclarePairedDelimiter{\abs}{\lvert}{\rvert}		

\DeclarePairedDelimiter{\inner}{\langle}{\rangle}

\DeclarePairedDelimiter{\setof}{\{}{\}}		
\DeclarePairedDelimiterX{\setdef}[2]{\{}{\}}{#1:#2}		
\DeclarePairedDelimiterXPP{\exclude}[1]{\mathopen{}\setminus}{\{}{\}}{}{#1}		

\newcommand{\N}{\mathbb{N}}		
\newcommand{\R}{\mathbb{R}}		

\DeclareMathOperator*{\argmax}{arg\,max}		
\DeclareMathOperator*{\argmin}{arg\,min}		
\DeclareMathOperator{\poly}{poly}		
\DeclareMathOperator*{\trunc}{Trunc}		

\DeclareMathOperator{\bigoh}{\mathcal{O}}		
\DeclareMathOperator{\conv}{conv}		
\DeclareMathOperator{\dom}{dom}		
\DeclareMathOperator{\grad}{\nabla}		
\DeclareMathOperator{\one}{\mathds{1}}		
\DeclareMathOperator{\supp}{supp}		
\DeclareMathOperator{\unif}{unif}		

\newmacro{\dd}{\:d}		
\newcommand{\eps}{\varepsilon}		

\newmacro{\const}{c}		
\newmacro{\Const}{\rho}		
\newmacro{\coefalt}{\mu}		
\usepackage{xparse}
\newmacro{\param}{\theta}		
\newmacro{\params}{\Theta}		

\newmacro{\pexp}{p}		
\newmacro{\qexp}{q}		
\newmacro{\rexp}{r}		

\newmacro{\radius}{r}

\newmacro{\tstart}{0}			
\renewcommand{\time}{\debug{t}}	
\newmacro{\timealt}{\tau}		
\newmacro{\timealtalt}{s}		
\newmacro{\tend}{T}				

\newop{\brep}{br}		

\DeclarePairedDelimiter{\norm}{\lVert}{\rVert}		
\DeclarePairedDelimiterXPP{\dnorm}[1]{}{\lVert}{\rVert}{_{\ast}}{#1}		

\DeclarePairedDelimiterXPP{\onenorm}[1]{}{\lVert}{\rVert}{_{1}}{#1}		
\DeclarePairedDelimiterXPP{\twonorm}[1]{}{\lVert}{\rVert}{_{2}}{#1}		
\DeclarePairedDelimiterXPP{\supnorm}[1]{}{\lVert}{\rVert}{_{\infty}}{#1}		
\DeclarePairedDelimiterXPP{\opnorm}[1]{}{\lVert}{\rVert}{_{\operatorname{op}}}{#1}        

\makeatletter
\newcommand*{\tran}{{\mathpalette\@tran{}}}
\newcommand*{\@tran}[2]{\raisebox{\depth}{$\m@th#1\intercal$}}
\makeatother

\newcommand{\defeq}{\coloneqq}		
\newcommand{\eqdef}{\eqqcolon}		

\newcommand{\from}{\colon}		

\newop{\ex}{\mathbb{E}}		
\newop{\prob}{\mathbb{P}}		
\newop{\Var}{Var}		

\DeclarePairedDelimiterXPP{\exof}[1]{\ex}{[}{]}{}{
 #1}

\DeclarePairedDelimiterXPP{\probof}[1]{\prob}{(}{)}{}{		
 #1}

\DeclarePairedDelimiterXPP{\oneof}[1]{\one}{\{}{\}}{}{
 #1}

\newmacro{\sample}{\omega}		
\newmacro{\samples}{\Omega}		

\newmacro{\seed}{\theta}		
\newmacro{\seeds}{\Theta}		

\usepackage[kerning=true]{microtype}		

\usepackage{tabto}		
\usepackage{xspace}		

\newacro{FTRL}{Follow-The-Regularized-Leader}

\newmacro{\xnum}{n}		
\newmacro{\ynum}{m}		

\newmacro{\xsimplex}{\simplex^{\xnum}}
\newmacro{\ysimplex}{\simplex^{\ynum}}

\newmacro{\learn}{\eta}		

\newmacro{\x}{x}	
\newmacro{\y}{y}	

\newmacro{\xt}{x(t)}
\newmacro{\yt}{y(t)}

\newcommand{\xstrat}[1]{\x \parens*{#1}}	
\newcommand{\ystrat}[1]{\y \parens*{#1}}	
\newmacro{\xc}{\x_c}
\newmacro{\xd}{\x_d}

\newmacro{\mat}{A}		
\newmacro{\matalt}{B}	

\newmacro{\A}{A}
\newmacro{\B}{B}

\newop{\Valop}{Val}
\newmacro{\val}{\Valop(\mat)}	

\newmacro{\payoff}{u_{o}}	
\newmacro{\valfix}{u_{o}^*}	

\newmacro{\xsp}{\xsimplex}		
\newmacro{\ysp}{\ysimplex}		

\newop{\hr}{S}	
\newop{\HR}{S}	

\newop{\reg}{h} 
\newmacro{\regmax}{\reg_{\max}} 
\newmacro{\regmin}{\reg_{\min}} 
\newmacro{\choicemap}{Q} 

\newop{\regconj}{\reg^*} 
\newmacro{\dualy}{z} 
\newmacro{\dualysp}{\ysp_{\ast}} 

\newmacro{\fgen}{f}
\newmacro{\ugen}{u} 
\newmacro{\vgen}{v} 

\newop{\breg}{D} 

\newop{\simplex}{\Delta}		

\newmacro{\ball}{\mathbb{B}} 

\newop{\contreward}{R_{\mathrm{cont}}} 
\newcommand{\rcont}[3]{\contreward \parens*{#1,\, #2,\, #3}} 
\newcommand{\optrcont}[2]{\contreward^* \parens*{#1,\, #2}} 

\newop{\discreward}{R_{\mathrm{disc}}} 
\newcommand{\rdisc}[3]{\discreward \parens*{#1,\, #2,\, #3}} 
\newcommand{\optrdisc}[2]{\discreward^* \parens*{#1,\, #2}} 

\newmacro{\zeroreward}{0} 

\newmacro{\learnhr}{Y} 

\newmacro{\firstterm}{\Psi} 

\newmacro{\scparam}{\alpha} 

\newcommand{\sol}[1]{#1^{\ast}}		

\newcommand{\xopt}{\sol{\x}}		
\newmacro{\xfix}{\hat \x}
\newmacro{\xne}{\x^{\star}} 
\newmacro{\yopt}{\sol{\y}}		
\newmacro{\ybrep}{\sol{\hat \y}}

\renewcommand{\vv}{\debug{v}} 

\newmacro{\tpennies}{k} 

\newmacro{\smoothparam}{\beta}

\newop{\neven}{N_e}
\newop{\nodd}{N_o}

\newmacro{\gap}{\delta}
\newmacro{\gapmin}{\delta_{\min}}
\newmacro{\gapmax}{\delta_{\max}}
\newmacro{\optset}{\brep(\xfix)}
\newmacro{\optnum}{k}
\newop{\expl}{\operatorname{EXP}}
\newop{\expld}{\expl^{\mathrm{d}}}

\newmacro{\xcp}{\tilde{\x}}
\newmacro{\kkt}{\lambda}

\newop{\vagop}{VAG}
\newop{\lagop}{LAG}
\newmacro{\vag}{\vagop(\tend)}
\newmacro{\lag}{\lagop(\tend)}
\newmacro{\lagd}{\lagop^{\mathrm{d}}(\tend)}

\newmacro{\ns}{\text{ns}}

\newmacro{\Ch}{C_{\reg}(\tend)}
\newmacro{\Chl}{\underline{\Delta V}(\tend)}
\newmacro{\Chu}{\overline{\Delta V}(\tend)}
\newmacro{\Chsteep}{C_{\reg}^{\mathrm{steep}}}
\newmacro{\Chsteepl}{\underline{\Chsteep}}
\newmacro{\Chsteepu}{\overline{\Chsteep}}
\newmacro{\Chnsteep}{C_{\reg}^{\mathrm{nonsteep}}}
\newmacro{\Chnsteepl}{\underline{\Chnsteep}}
\newmacro{\Chnsteepu}{\overline{\Chnsteep}}

\newop{\acc}{\varepsilon}
\newop{\manip}{\mathcal{M}_{\reg}}
\newop{\trade}{\mathcal F}

\newop{\cost}{\mathcal{M}_{\reg}}

\newmacro{\cnorm}{c_{\norm{\cdot}}}

\DeclarePairedDelimiter{\posi}{[}{]_+}

\newmacro{\eventtwo}{\mathcal A_{\xnum \ynum}}
\newmacro{\eventgap}{\mathcal{E}_{\mathrm{gap}}}

\newcommand{\gvalue}{V^\star}            
\newcommand{\gregret}{\text{Regret}}   

\newmacro{\nsub}{N_{\mathrm{sub}}} 

\usepackage{tikz}
\usepackage{bm}

\theoremstyle{plain} 
\newtheorem{theorem}{Theorem}[section]
\newtheorem{corollary}[theorem]{Corollary}
\newtheorem{proposition}[theorem]{Proposition}
\newtheorem{lemma}[theorem]{Lemma}

\newtheorem{claim}[theorem]{Claim}

\newtheorem{assumption}{Assumption}

\newtheorem{definition}{Definition}

\newtheorem*{definition*}{Definition}

\newtheorem{observation}[theorem]{Observation}

\theoremstyle{definition} 

\theoremstyle{remark} 

\newtheorem{remark}{Remark}

\AfterEndEnvironment{definition}{\noindent\ignorespaces}
\AfterEndEnvironment{infdefinition}{\noindent\ignorespaces}
\AfterEndEnvironment{example}{\noindent\ignorespaces}
\AfterEndEnvironment{assumption}{\noindent\ignorespaces}
\AfterEndEnvironment{lemma}{\noindent\ignorespaces}
\AfterEndEnvironment{claim}{\noindent\ignorespaces}
\AfterEndEnvironment{theorem}{\noindent\ignorespaces}
\AfterEndEnvironment{proposition}{\noindent\ignorespaces}
\AfterEndEnvironment{fact}{\noindent\ignorespaces}
\AfterEndEnvironment{question}{\noindent\ignorespaces}
\AfterEndEnvironment{corollary}{\noindent\ignorespaces}
\AfterEndEnvironment{model}{\noindent\ignorespaces}
\AfterEndEnvironment{remark}{\noindent\ignorespaces}
\AfterEndEnvironment{proof}{\noindent\ignorespaces}
\AfterEndEnvironment{fact}{\noindent\ignorespaces}
\AfterEndEnvironment{minftheorem}{\noindent\ignorespaces}
\AfterEndEnvironment{inftheorem}{\noindent\ignorespaces}
\AfterEndEnvironment{maintheorem}{\noindent\ignorespaces}
\AfterEndEnvironment{restatable}{\noindent\ignorespaces}
\AfterEndEnvironment{observation}{\noindent\ignorespaces}

\crefname{section}{Section}{Sections}
\crefname{theorem}{Theorem}{Theorems}
\crefname{theorem*}{Theorem}{Theorems}
\crefname{inftheorem}{Informal Theorem}{Informal Theorems}
\crefname{assumption}{Assumption}{Assumptions}
\crefname{lemma}{Lemma}{Lemmas}
\crefname{definition}{Definition}{Definitions}
\crefname{infdefinition}{Informal Definition}{Informal Definitions}
\crefname{conjecture}{Conjecture}{Conjectures}
\crefname{corollary}{Corollary}{Corollaries}
\crefname{construction}{Construction}{Constructions}
\crefname{conjecture}{Conjecture}{Conjectures}
\crefname{claim}{Claim}{Claims}
\crefname{observation}{Observation}{Observations}
\crefname{proposition}{Proposition}{Propositions}
\crefname{fact}{Fact}{Facts}
\crefname{question}{Question}{Questions}
\crefname{problem}{Problem}{Problems}
\crefname{remark}{Remark}{Remarks}
\crefname{example}{Example}{Examples}
\crefname{equation}{Equation}{Equations}
\crefname{appendix}{Appendix}{Appendices}
\crefname{algorithm}{Algorithm}{Algorithms}
\crefname{model}{Model}{Models}
\crefname{figure}{Figure}{Figures}
\crefname{condition}{Condition}{Conditions}

\AfterEndEnvironment{definition}{\noindent\ignorespaces}
\AfterEndEnvironment{infdefinition}{\noindent\ignorespaces}
\AfterEndEnvironment{assumption}{\noindent\ignorespaces}
\AfterEndEnvironment{problem}{\noindent\ignorespaces}
\AfterEndEnvironment{lemma}{\noindent\ignorespaces}
\AfterEndEnvironment{theorem}{\noindent\ignorespaces}
\AfterEndEnvironment{proposition}{\noindent\ignorespaces}
\AfterEndEnvironment{fact}{\noindent\ignorespaces}
\AfterEndEnvironment{question}{\noindent\ignorespaces}
\AfterEndEnvironment{corollary}{\noindent\ignorespaces}
\AfterEndEnvironment{model}{\noindent\ignorespaces}
\AfterEndEnvironment{remark}{\noindent\ignorespaces}
\AfterEndEnvironment{proof}{\noindent\ignorespaces}
\AfterEndEnvironment{fact}{\noindent\ignorespaces}
\AfterEndEnvironment{minftheorem}{\noindent\ignorespaces}
\AfterEndEnvironment{inftheorem}{\noindent\ignorespaces}
\AfterEndEnvironment{maintheorem}{\noindent\ignorespaces}
\AfterEndEnvironment{restatable}{\noindent\ignorespaces}
\AfterEndEnvironment{infassumption}
{\noindent\ignorespaces}
\AfterEndEnvironment{infcorollary}{\noindent\ignorespaces}

\crefname{section}{Section}{Sections}
\crefname{theorem}{Theorem}{Theorems}
\crefname{lemma}{Lemma}{Lemmas}
\crefname{problem}{Problem}{Problems}
\crefname{program}{Program}{Progams}
\crefname{definition}{Definition}{Definitions}
\crefname{conjecture}{Conjecture}{Conjectures}
\crefname{corollary}{Corollary}{Corollaries}
\crefname{construction}{Construction}{Constructions}
\crefname{conjecture}{Conjecture}{Conjectures}
\crefname{claim}{Claim}{Claims}
\crefname{observation}{Observation}{Observations}
\crefname{proposition}{Proposition}{Propositions}
\crefname{fact}{Fact}{Facts}
\crefname{question}{Question}{Questions}
\crefname{problem}{Problem}{Problems}
\crefname{remark}{Remark}{Remarks}
\crefname{example}{Example}{Examples}
\crefname{equation}{Equation}{Equations}
\crefname{appendix}{Section}{Sections}
\crefname{algorithm}{Algorithm}{Algorithms}
\crefname{model}{Model}{Models}
\crefname{figure}{Figure}{Figures}
\crefname{infassumption}{Informal Assumption}{Informal Assumptions}
\crefname{inftheorem}{Informal Theorem}{Informal Theorems}
\crefname{infdefinition}{Informal Definition}{Informal Definitions}
\crefname{minftheorem}{Main Informal Theorem}{Main Informal Theorems}
\crefname{maintheorem}{Main Theorem}{Main Theorems}
\crefname{assumption}{Assumption}{Assumptions}
\crefname{step}{Step}{Steps}
\crefname{result}{Result}{Results}
\crefname{event}{Event}{Events}
\crefname{none}{}{}

\definecolor{myC}{rgb}{0, 255, 255}
\definecolor{myY}{rgb}{204, 204, 0}
\definecolor{myM}{rgb}{255, 0, 255}
\definecolor{secinhead}{RGB}{249,196,95}
\definecolor{lgray}{gray}{0.8}

\usepackage{appendix}
\crefname{appsec}{Appendix}{Appendices}
\title{No Coin Left Behind: Maximizing Strategic Surplus Against No-Regret Dynamics}

\date{}

\author{
Yiheng Su\\
University of Wisconsin--Madison\\
\texttt{ysu24@cs.wisc.edu}
\and
Emmanouil-Vasileios Vlatakis-Gkaragkounis\\
University of Wisconsin--Madison\\
\texttt{vlatakis@wisc.edu}
}

\begin{document}

\pagenumbering{roman}

\maketitle

\bigskip 
\bigskip 
\thispagestyle{empty}

\begin{abstract}
{
We investigate the \emph{strategic surplus} obtainable against a \emph{Follow-the-Regularized-Leader (FTRL)} learner with constant step size $\eta$ in  $n\times m$ two-player zero-sum games played over $T$ rounds against a clairvoyant optimizer. In contrast with prior analysis, we show that the extraction of such regret-scale surplus is an inherent feature of the FTRL family, rather than an artifact of specific instantiations. 
First, for a fixed max-min optimizer, we establish a sweeping law of order $\Omega(\nsub/\eta)$, proving that utility surplus scales with the number of the learner's suboptimal actions $N$ and vanishes in their absence.
Second, for an alternating optimizer, a surplus of $\Omega(\eta T/\poly(n,m))$ can be guaranteed regardless of the equilibrium structure, with high probability, in random games.
Our analysis uncovers a sharp geometric dichotomy: \emph{non-steep}
regularizers allow the optimizer to realize the maximal transient surplus via
finite-time elimination of suboptimal actions, whereas \emph{steep} regularizers
introduce a vanishing tail correction that can delay surplus saturation.
Finally, we discuss whether this leverage persists under bilateral payoff
uncertainty and propose a susceptibility measure quantifying which regularizers
are most vulnerable to learner-aware strategic steering.

}
\end{abstract}

\clearpage
\tableofcontents
\thispagestyle{empty}

\clearpage
\pagenumbering{arabic}

\section{Introduction}
\label{section:introduction}
The study of learning algorithms in multi-agent systems has traditionally been viewed through a \emph{defensive} lens.
For instance, in the archetypal setting of repeated zero-sum games, algorithms like \emph{Multiplicative Weight Updates (MWU)} \citep{arora2012multiplicative} are celebrated for their robustness: 
they guarantee that a learner secures an average payoff at least equal to the minimax value of the game, effectively neutralizing any adversary in the long run.
This “no-regret” property has established such dynamics as the bedrock of modern equilibrium computation from high-frequency trading \citep{yang2012behavior}  to solving imperfect-information large games like Poker \citep{zinkevich2007regret,brown2019superhuman}.

However, a parallel line of inquiry at the intersection of economics and learning theory 
suggests that this is only half the story. Pioneering work by \citet{braverman2018selling,cai2023selling} demonstrated that in auction settings, a strategic seller can {exploit} a no-regret buyer to {extract nearly the entire social surplus}. More recently, this perspective has evolved into an emerging theory of ``steering'' learning agents, with applications ranging from contract design \citep{kolumbus2024contracting} to Bayesian persuasion \citep{mansour2022strategizing}. The unifying insight across these works is that when the symmetry of information is broken—specifically, when a ``clairvoyant'' optimizer faces a learner with a known update rule—the interaction morphs from a standard game into a \emph{Principal-Agent} control problem \citep{lin2024generalized}. Hence, from this viewpoint, the learner's regret is no longer merely a vanishing error term; it becomes a structured and potentially harvestable reservoir of additional utility \cite{arunachaleswaran2025learning,brown2023learning}.

\paragraph{The \(\Omega(T)\) and \(o(T)\) exploitation regimes.}
Two distinct scales of exploitation have emerged in the literature.
Early results showed that against mean-based algorithms, a strategic optimizer can secure an \(\Omega(T)\) gain above the Stackelberg payoff \citep{braverman2018selling,cai2023selling,kolumbus2024contracting}.
More recent work, however, identifies settings in which the maximum attainable excess utility shrinks to \(o(T)\), together with intriguing asymmetries between the corresponding lower and upper bounds \citep{lin2024generalized,assos2024maximizing}.
The latter results reveal that even the sublinear regime is not merely a technical footnote: the order of a cumulative gain--\(\Theta(\sqrt{T})\) or \(\Theta(T^{2/3})\)-- represents different levels of persistent finite-horizon cost of learning.

Given their central role, zero-sum games have offered a natural testbed for these questions. 
Recent work \citep{assos2024maximizing} demonstrates that learner-aware optimization can exceed the von Neumann benchmark $T \cdot V^*$ even in this classical setting---yet results remain tied to specific algorithms such as Hedge, or to particular update geometries. Consequently, a unified account of exploitability across the full family of no-regret dynamics is still lacking, which brings us to the following central question:
\begin{center} \emph{Is the learner's vulnerability an artifact of the complex economic setting,\\ or is it a structural inevitability encoded in the algorithm itself?}\end{center}
\noindent In this work, we argue for the latter. To substantiate this thesis, we anchor our investigation in the canonical setting of repeated zero-sum games between an \emph{FTRL learner} and a \emph{clairvoyant optimizer}. We focus on the \emph{Follow-the-Regularized-Leader} (FTRL) framework—the standard paradigm for no-regret learning—as it unifies algorithms like Multiplicative Weights, 1/2-Tsallis or Online  Gradient Descent under a single geometric umbrella. 
For clarity, we work primarily in the \emph{full-feedback} model. Appendix~\ref{appendix:azuma} provides concentration bounds relating realized \emph{bandit-feedback} payoffs or noisy observations to the corresponding full-feedback quantities, while Appendix~\ref{appendix:experiment} reports exploratory bandit-feedback experiments with explicit exploration. 

\begin{wrapfigure}{r}{0.35\textwidth}
\footnotesize
\centering
\begin{tikzpicture}[scale=0.75, transform shape]
  \begin{axis}[
      width=\linewidth,
      height=4cm,
      scale only axis,
      axis lines=left,
      xmin=0, xmax=10,
      ymin=0, ymax=12,
      xlabel={Time-step $t$},
      ylabel={Cumulative Optimizer Gain ($R_T$)},
      xtick={0,10},
      xticklabels={0,$T$},
      ytick=\empty,
      clip=false,
      enlargelimits=false,
      axis line style={semithick},
      tick style={semithick},
      xlabel style={font=\scriptsize, yshift=2.2ex},
      ylabel style={font=\scriptsize, yshift=-0.2ex},
      title={\textbf{Cumulative Optimizer Gain ($R_T$)}\\[-0.2ex]\textbf{and Surplus Extraction}},
      title style={
        align=center,
        font=\footnotesize,
        yshift=7.6ex,%
        xshift=3.6ex,
        inner sep=0pt
      },
      every axis plot/.append style={semithick},
  ]

    \addplot[name path=baseline, domain=0:10, dashed] {0.51*x};

    \addplot[name path=clairvoyant, domain=0:10] {0.6*x + 1.6*sqrt(x)};

    \addplot[fill=blue!12, draw=none] fill between[of=baseline and clairvoyant];

    \node[font=\scriptsize, align=center, inner sep=1pt]
      at (axis cs:7.35,5.9) {Surplus Reservoir};


    \draw[<-]
      (axis cs:5.1,{0.6*5.1 + 1.6*sqrt(5.1)}) -- (axis cs:4.15,8.45)
      node[above, align=center, font=\scriptsize, inner sep=1pt]
      {Optimal Clairvoyant\\Optimizer};

    \draw[<-]
      (axis cs:6.45,{0.51*6.45}) -- (axis cs:7.05,2.1)
      node[below, align=center, font=\scriptsize, inner sep=1pt]
      {Naive Optimizer\\(baseline)};

    \draw[dotted, semithick]
      (axis cs:10,0) -- (axis cs:10,{0.6*10 + 1.6*sqrt(10)});

    \draw[<->]
      (axis cs:10.24,0) -- (axis cs:10.24,5.1)
      node[midway, right, font=\scriptsize, inner sep=1pt] {$T \cdot V^\star$};

    \draw[<->]
      (axis cs:10.24,5.2) -- (axis cs:10.24,{0.6*10 + 1.6*sqrt(10)})
      node[midway, right, font=\scriptsize, inner sep=1pt] {$+\Omega(\mathrm{Regret})$};

  \end{axis}
\end{tikzpicture}
\vspace{-1\baselineskip}
\label{img:regret-curve}
\vspace{-1\baselineskip}
\end{wrapfigure}
\paragraph{The Surplus Reservoir.}
To formalize the interaction, consider the optimizer’s cumulative reward $\mathrm{Rew}_T$ over $T$ rounds against the FTRL learner. Standard no-regret bounds confine this reward within a tight envelope: 
\begin{equation}\label{eq:reward sandwich}
{
\boxed{
T \cdot \gvalue
\;\le\;
\mathrm{Rew}_T
\;\le\;
T \cdot \gvalue + \gregret_{\mathcal{L}}^{\learn}(T)
}}
\tag{$\star$}
\end{equation}
where $T \cdot V^\star$ is the cumulative von Neumann value 
and $\mathrm{Regret}_{\mathcal{L}}^{\eta}(T)$ the learner's regret under its update rule $\mathcal{L}$ (e.g., FTRL) with step size $\learn$.
The lower bound is trivially achieved by any max-min strategy. 
Maximizing the upper bound, however, is not: obtaining both the learner's maximal regret and the game's cumulative value requires the optimizer to actively exploit the learner's specific update rule.
Recall that the folklore analysis shows that $\mathrm{Regret}_{\mathcal{L}}^{\eta}(T)\leq A/\eta+B\eta T$. This leads to the crux of our work, which we term \emph{maximal extractability}: 
\begin{center} \emph{Can a strategic optimizer actively \emph{steer} the learner’s dynamics so as to consistently saturate this reservoir of extractable surplus such that
\(
    \mathrm{Rew}_T \;\approx\; T \cdot V^\star + \Omega(\max\!\big\{1/\eta,\; \eta\,T\big\}).
\)
}\end{center}
For two-player zero-sum games, we answer affirmatively by constructing explicit optimizer strategies---equivalently, adversarial loss sequences---that force the learner to incur regret at least
\( \Omega(\max\!\big\{C_1/\eta,\; C_2\,\eta\,T\big\}),
\)
for game-dependent constants $C_1, C_2 > 0$, matching the classical upper bound up to constant factors. Hence, our results constitute \emph{instance-dependent adversarial regret lower bounds}: tight, game-interpretable certificates that the surplus reservoir is fully saturatable. Our analysis reveals that this saturability is governed by two structural features of the learner's algorithm: the \emph{suboptimality geometry} of the best-response polytope, and the \emph{curvature} of the regularizer---which together determine how much inertia the optimizer can exploit.

\subsection{Our Contributions}
\label{subsection:contributions}
Formally, we provide a comprehensive characterization for FTRL learners,
showing that surplus saturation is inextricably linked to the boundary
geometry of the regularizer. As a preliminary step, we develop a toolkit
that links discrete-time FTRL updates to their continuous-time
counterparts, yielding sharp best-response convergence rates for generic
regularizers. (Lemma~\ref{lem:bound-for-cont-ftrl} \&
Corollary~\ref{cor:discrete-exploit-finite})

Our contributions divide naturally into three results of increasing depth:

\noindent \textbf{1. Passive Utility-Extraction: The ``Inverse-Rate Law''} 
\textbf{Fixed Strategy.}
We begin with the warm-up case of a fixed optimizer strategy---the most natural choice being the max-min strategy. 
Although static play might seem
easily learnable and thus harmless, we show the opposite. Since the learner
does not know the game a priori, she must explore all actions before
identifying best responses. Every suboptimal action incurs a convergence
delay the optimizer passively collects. We establish the Inverse-Rate Law:
the surplus scales as 
{$\Theta(|\nsub|/\eta)$, where $|\nsub|$}
counts actions that are not best
responses to the optimizer's fixed strategy. At step size
$\eta = O(T^{-\beta})$ this yields surplus $\Theta(T^{\beta})$
(Fig.~\ref{fig:fixed-subopt}). The boundary case is instructive: when all
learner actions are best responses (e.g., Rock-Paper-Scissors against the
max-min strategy, Fig.~\ref{fig:rps-fixed}), fixed-strategy surplus vanishes
entirely, since the learner merely oscillates among equally optimal actions.
(Theorem~\ref{thm:exploit-finite})

\newpage
\begin{wrapfigure}[29]{r}{0.35\textwidth}
  \centering
  \includegraphics[width=\linewidth]{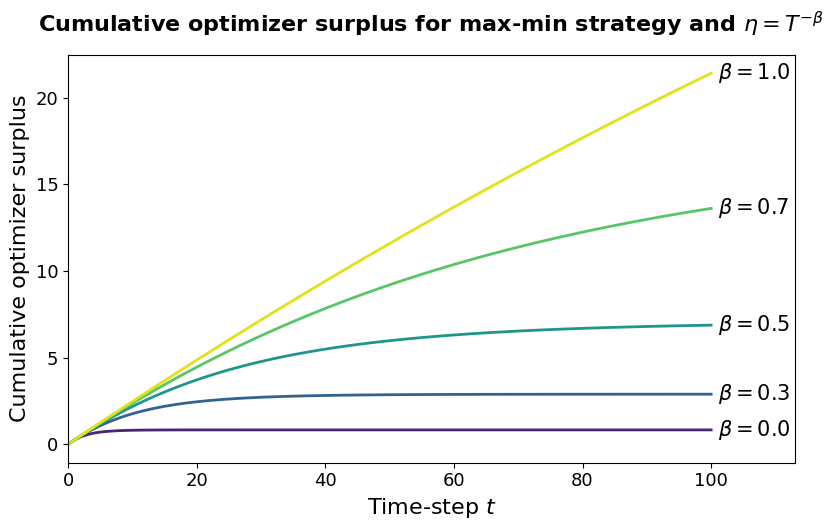}
  \captionsetup{font=scriptsize}
  \caption{\label{fig:fixed-subopt}Fixed max-min strategy, Game~1\textsuperscript{*}.  
    Surplus $\Theta(T^{\beta})$, confirming the \emph{Inverse-Rate Law}.}

  \includegraphics[width=\linewidth]{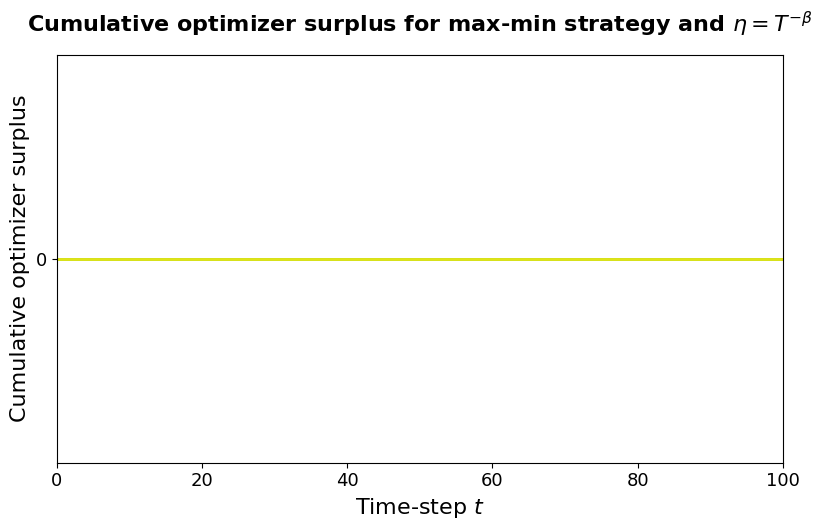}
  \caption{\label{fig:rps-fixed}Fixed max-min strategy, Game~2\textsuperscript{$\dagger$}.
    All curves coincide at $0$: surplus vanishes.}

  \includegraphics[width=\linewidth]{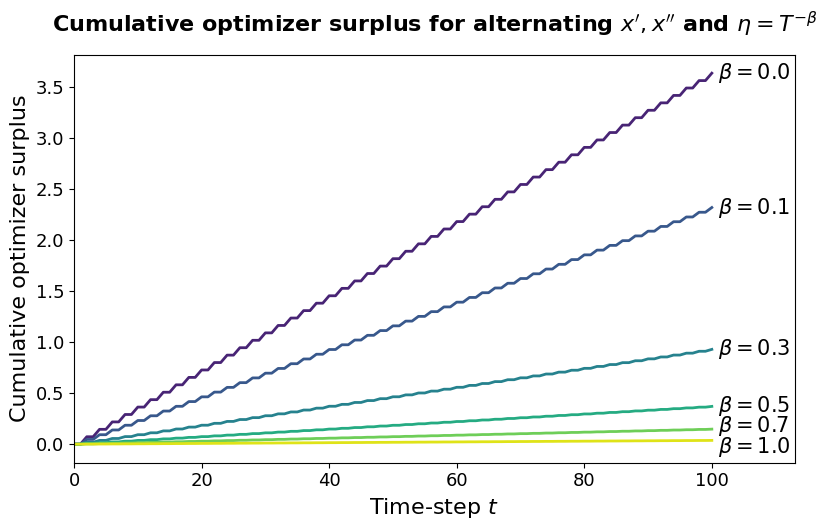}
  \caption{\label{fig:rps-alternating}Alternating trap, Game~2\textsuperscript{$\dagger$}.
    Surplus $\Omega(T^{1-\beta})$ for all $\beta<1$.}
\end{wrapfigure}
\noindent \textbf{2. Active Utility-Extraction: The ``Trap'' Dynamic Strategy.}
However, the vanishing of fixed-strategy surplus in fully mixed games is
not a dead end. We show the optimizer can recover nonzero surplus by
\emph{actively engineering} the learner's trajectory. The key structural
condition is that at least two best responses to the optimizer's max-min
support are not ordered by weak domination; this prevents the learner from
settling and keeps her oscillating perpetually. While this condition was
assumed in prior work \citep{assos2024maximizing}, its prevalence and
quantitative bite were left unexamined. Rather than taking it as given,
we turn to random games: we prove that the condition holds with high
probability. Our resulting \emph{alternating trap strategy}---reminiscent
of trap constructions in auctions \citep{braverman2018selling} and
free-fall phases in contracts \citep{kolumbus2024contracting}---guarantees
surplus $\Omega(\eta T)$ with high probability
(Theorem~\ref{thm:avg-main}). At step size $\eta = O(T^{-\beta})$,
this yields surplus $\Theta(T^{1-\beta})$ (Fig.~\ref{fig:rps-alternating}).

\par\smallskip
\noindent\textit{Results 1 and 2 are complementary: Result~1 captures surplus
$\Theta(|\nsub|/\eta)$ whenever suboptimal actions
exist; Result~2 guarantees nonzero surplus in fully mixed games where
Result~1 gives nothing. Together they cover all cases. Notably, the learner
minimizes total surplus by balancing the two terms, which is achieved at
$\eta = \Theta(T^{-1/2})$---recovering the classical $\Theta(\sqrt{T})$
benchmark.}
\par\smallskip

\noindent \textbf{3. The Price of Learning Best Responses.}
Due to space constraints, we defer the complementary learner-side view of saturability, based on the cost of identifying a best response in Appendix~\ref{appendix:app_pbr}.This yields a sharp dichotomy: non-steep regularizers (e.g., Euclidean) reach an $\varepsilon$-best response in $O(1/\varepsilon)$ time, whereas steep regularizers (e.g., entropic) require $\omega(1/\varepsilon)$. The reason is structural: their infinite boundary gradients---the very feature that grants them robustness near the boundary---also prevent rapid
convergence to pure strategies. (Theorem~\ref{thm:learning-cost})

\begingroup
\renewcommand{\thefootnote}{\fnsymbol{footnote}}

\footnotetext[1]{\noindent\textbf{Game~1}:
$A=\left(\begin{smallmatrix}
  1 & \!-1 & 1 & 2\\
  -1 & 1 & 1 & 0\\
  0 & 0 & \!-1 & 0
\end{smallmatrix}\right)$,\;
$x_{\max\min}^\star=(\tfrac{1}{3},\tfrac{1}{3},\tfrac{1}{3})$,\;
$|\nsub|=2$,\;
fixed-strategy surplus $\Theta(2/\eta)$.
}

\footnotetext[2]{\noindent\textbf{Game~2}:
$A_{\mathrm{RPS}}=\left(\begin{smallmatrix}
  0 & \!-1 & 1\\
  1 & 0 & \!-1\\
  \!-1 & 1 & 0
\end{smallmatrix}\right)$,\;
$x_{\max\min}^\star=(\tfrac{1}{3},\tfrac{1}{3},\tfrac{1}{3})$,\;
$|\nsub|=0$,\;
alternating-trap surplus $\Omega(\eta T)$.
}
\endgroup

\section{Preliminaries}
\label{section:preliminaries}
We consider a two-player zero-sum game played between an \emph{Optimizer} and a \emph{Learner} over a time horizon $T$, which may be discrete ($t \in \{0, \dots, T-1\}$) or continuous ($t \in [0, T]$).
The optimizer chooses strategies $\xstrat{t} \in \xsimplex$ and the learner chooses $\ystrat{t} \in \ysimplex$, where $\simplex^d$ denotes the probability simplex in $\R^d$.
The game is defined by a payoff matrix $\mat \in \R^{\xnum \times \ynum}$. At time $t$, the optimizer receives payoff $\payoff(x,y)=\xstrat{t}^\top \mat \ystrat{t}$ and the learner receives $u_{\ell}(x,y)=\xstrat{t}^\top \matalt \ystrat{t}$, where $\matalt = -\mat$.
Additionally, we denote by $\val \defeq \max_{\x \in \xsimplex} \min_{\y \in \ysimplex} \x^\top \mat \y$ the \emph{minimax value} of the game.

Next, we briefly review two essential concepts from convex analysis\footnote{
For an extensive convex-analysis recap, closed-form derivations, and the worked examples, see Appendix~\ref{appendix:app_preliminaries}.
} that underpin our results:
\vspace{-0.6em}
\begin{itemize}[topsep=2pt,itemsep=0pt, leftmargin=*]
\item\textbf{Fenchel (convex) conjugate.}
Let $\reg: \R^m \to \mathbb{R}\cup\{+\infty\}$ be a \emph{proper} convex function\footnote{By \emph{proper}, we mean that the effective domain $\dom(\reg) \defeq \{x \in \R^m \mid \reg(x) < +\infty\}$ is nonempty and $\reg(x) > -\infty$ everywhere. In our context, functions are usually defined on the simplex $\Delta^m$ and are set to $+\infty$ outside it.}. 
The convex conjugate $\reg^{\ast}:\mathbb{R}^m\to\mathbb{R}\cup\{+\infty\}$ is defined by
\(    \reg^{\ast}(y)\;\coloneqq\;\sup_{x\in \dom(\reg)}\;\big\{\langle x,y\rangle - \reg(x)\big\},
 y\in\mathbb{R}^m.
\)\item\textbf{Bregman divergence.}
Let $\reg:\Delta^m\to\mathbb{R}\cup\{+\infty\}$ be a strictly convex function that is differentiable on the \emph{relative interior} of the simplex\footnote{Recall that the relative interior $\mathrm{ri}(C)$ is the interior of a set $C$ relative to its affine hull. For the probability simplex, this corresponds to the set of strictly positive probability distributions, $\mathrm{ri}(\Delta^m) = \{x \in \Delta^m \mid x_i > 0, \forall i\}$.}, denoted by $\mathrm{ri}(\Delta^m)$.
The Bregman divergence induced by $\reg$ is defined as
\(
\breg_{\reg}(x, x')\;\coloneqq\;\reg(x)-\reg(x')-\langle \nabla \reg(x'),\,x-x'\rangle,
 x\in \dom(\reg), \; x'\in \mathrm{ri}(\Delta^m).
\)
\end{itemize}
\paragraph{$\blacktriangleright$The Learner. (Alice) } \vspace{-0.75em}
The learner employs a fixed update \emph{Follow-the-Regularized-Leader (FTRL)} rule with a constant step size $\learn > 0$.\footnote{ Although we use a constant step size, we do not treat it as independent of the horizon. Instead, we work in an asymptotic regime where $\eta=\eta(T)\to 0$ as $T\to\infty$ while $\eta T\to\infty$; for instance, $\eta=\Theta(1/\sqrt{T})$ satisfies these conditions.}
She maintains a cumulative payoff vector, or \emph{score}, denoted by $\HR(t) \in \R^{\ynum}$.
Initially $\HR(0) = 0$. The score accumulates her observed utility gradients:
\begin{align}
    \HR(t) =
    \begin{cases}
        \sum_{\tau=0}^{t-1} \matalt^\top \xstrat{\tau} & \text{(Discrete time),} \\
        \int_{0}^{t} \matalt^\top \xstrat{\tau} \dd\tau & \text{(Continuous time).}
    \end{cases}\tag{Cumulative Reward}
    \label{eq:score_def}
\end{align}
The learner's strategy $\ystrat{t}$ is derived from this score via a regularized projection.
Let $\reg: \ysimplex \to \R$ be a strictly convex, continuous \emph{regularizer}.
We define the \emph{choice map} $\choicemap_{\reg}: \R^{\ynum} \to \ysimplex$ as
\begin{align} \label{eq:choice-map}
       \text{Learner plays strategy} \quad
 \ystrat{t} = \choicemap_{\reg}\big(\learn \HR(t)\big)
    \defeq \argmax_{\y \in \ysimplex} \big\{ \inner{\learn \HR(t), \y} - \reg(\y) \big\}. \tag{Choice Map}
\end{align}
Intuitively, FTRL is a ``smoothed" best-response: as $\learn \to \infty$ (or $\reg \to 0$), the learner approaches the exact best-response dynamics. In this work, we focus on the standard \emph{separable} regularizers of the form $\reg(\y) = \sum_{i=1}^{\ynum} \theta(y_i)$, where the kernel $\theta: [0,1] \to \R$ is strictly convex and differentiable on $(0,1]$.  In this setting, the choice map decouples, as follows:
\begin{align}\label{eq:kkt_structure}
    \ystrat{t}_i = \phi\big(\lambda(t) + \learn \HR(t)_i\big),
    \quad \text{subject to } \sum_i \ystrat{t}_i = 1,
\end{align}
where $\phi=\trunc_{[0,1]}\{(\theta')^{-1}\}$ is the inverse of kernel's derivative and $\lambda(t)$ is a normalization scalar.
\subsubsection*{Examples of FTRL methods (see Appendix~\ref{appendix:kernel-examples} for details)}
\begin{description}[leftmargin=0pt, font=\bfseries, itemsep=0.5pt]
    \item[$\star$Example 1. \emph{Entropic}:] The negative entropy $\reg(y)=\sum_{i} y_i \ln y_i$ leads to the \emph{logit} choice map {\small$\ystrat{t}_i \propto \exp(\eta \HR_i(t))$}, yielding the classic Exponential Weights, cf.~
    \citep{auer1995gambling}.
    \item[$\star$Example 2. \emph{Euclidean}:] The quadratic penalty $\reg(y)=\frac{1}{2}\|y\|_2^2$ induces the standard \emph{projection} map\\{\small$\ystrat{t}= \argmin_{y \in \ysimplex} \|y - \eta \HR(t)\|_2$}, corresponding to Online Gradient Descent, cf.~\citep{zinkevich2003online}.
    \item[$\star$Example 3. \emph{Tsallis}:]\vspace{-0.25em}
    For $q\in(0,1)$, the negative Tsallis regularizer is
    $\reg(y)=\frac{1}{q-1}\sum_i (y_i^q-y_i),$\vspace{-0.25em}\\ 
    cf.~\citep{abernethy2015fighting}. The associated choice map is characterized by $q\,\ystrat{t}_i^{\,q-1}
    =
    1+(q-1)\eta\bigl(\HR_i(t)-\tau\bigr),$
    where $\tau\in\mathbb{R}$ is the normalization scalar chosen so that $\sum_i \ystrat{t}_i=1$. As $q\to 1$, this regularizer converges to the negative entropy.

\end{description}
\vspace{-0.2em}
\noindent The geometry of $\theta$ near the boundary induces a fundamental dichotomy in the learner's behavior:
\begin{definition}[Steep vs. Non-Steep Regularizers]\it
    A regularizer is \textbf{steep} if $|\theta'(y)| \to \infty$ as $y \to 0^+$ (e.g., {Negative Entropy/Negative Tsallis}). It is \textbf{non-steep} if $\theta'(0^+)$ is finite (e.g., {Euclidean}).
\end{definition}

\begin{remark}\it This distinction is critical: steep regularizers force $\ystrat{t}$ to strictly stay in the interior of the simplex (since $\phi(-\infty)=0$), whereas non-steep regularizers allow the learner to assign \emph{exact zero probability} to actions in finite time, enabling the ``finite-time elimination" of suboptimal strategies.
\end{remark}

\begin{remark}\it
The map $\phi$, introduced in \cite{giannou2021rate}, is an analytic workhorse that lifts our
analysis beyond the Hedge-specific calculations of
\cite{assos2024maximizing}. For any regularizer $\theta$, $\phi$
encodes its local curvature in a single object: the weights
$1/\theta''(y_i(r))$ emerge naturally, driving the sharp
coordinate-wise bounds of Lemma~\ref{lem:ftrl-fixed-x} in fixed strategy and the
curvature terms in the alternating-strategy analysis in Lemma~\ref{lem:variance-identity}. This is
precisely what enables matching upper and lower bounds across the
full FTRL family.
\end{remark}
\paragraph{$\blacktriangleright$ The Optimizer. (Bob)}
The optimizer is a strategic agent with \emph{strategic foresight}. Unlike the learner, who reacts to past payoffs, the optimizer knows the game matrix $\mat$ and the learner's update rule. The optimizer's goal is to select a trajectory $\{\xstrat{t}\}_{t}$ to maximize his total payoff $\int \xstrat{t}^\top \mat \ystrat{t} \dd t$.
We summarize the interaction protocol below:
\begin{center}
\fbox{%
\begin{minipage}{0.95\textwidth}
\textbf{Protocol: Repeated Game with FTRL Learner}
\vspace{0.3em}
\hrule
\vspace{0.3em}
For each round $t = 1, \dots, T$:
\begin{enumerate}[leftmargin=1.5em,itemsep=-1pt,topsep=0pt]
    \item \textbf{Learner updates:} The learner computes $\ystrat{t} = \choicemap_{\reg}(\learn \HR(t))$.
    \item \textbf{Optimizer plays:} The optimizer selects $\xstrat{t} \in \xsimplex$ (potentially using knowledge of $\ystrat{t}$).
    \item \textbf{Payoffs \& Feedback:} Players receive payoffs $u_{o/\ell}( \xstrat{t},\ystrat{t})$. The learner observes the feedback vector $\matalt^\top \xstrat{t}$.
    \item \textbf{State update:} The learner updates the score $\HR(t) = \HR(t-1) + \matalt^\top \xstrat{t}$.
\end{enumerate}
\end{minipage}}
\end{center}
\paragraph{Feedback Models.}
Our primary analysis assumes \emph{full-feedback}, where the learner observes the entire vector $\matalt^\top \xstrat{t}$. In Appendix~\ref{appendix:azuma}, we briefly extend our results to the \emph{bandit-feedback} setting, where players only observe the realized scalar payoff $\pm\A_{i_t,j_t}$ from sampled actions $i_t\sim \x(t)$ and $j_t\sim \y(t)$. 
Unless stated otherwise, all results refer to the full-feedback setting.

\section{Our Toolbox}
\label{section: toolbox}

\subsection{A Discrete-Continuous Bridge}
\label{section:reward_decomposition}
To understand the long-run behavior in our setting, we establish a fundamental link between the optimizer's discrete and continuous objectives. We show that the discrete reward tracks a continuous-time benchmark, deviating only by a \emph{discretization gap} governed by the Bregman divergence of the learner's dual updates. We begin by defining the optimizer’s total reward in both the continuous-time and discrete-time settings, each of which depends on the learner’s strategy as determined by FTRL. For any optimizer’s strategy $\x(t) \from \bracks{\tstart, \tend} \to \xsp$:
\begin{align}
    \rcont{\x(t)}{\HR(\tstart)}{\tend}
    &= \int_{\tstart}^{\tend}
    \inner*{
        \x(t),
        \A \cdot
        \choicemap_{\reg}\parens*{
            \learn \bracks*{
                \HR(\tstart) +
                \int_{\tstart}^{\time}
                \B^\top \x(\timealt)\, \dd \timealt
            }
        }
    } \dd \time \label{eqn: cont reward} \\[-0.5ex]
    \rdisc{\x(t)}{\HR(\tstart)}{\tend}
    &= \sum_{\time=\tstart}^{\tend-1}
    \inner*{
        \x(t),
        \A \cdot
        \choicemap_{\reg}\parens*{
            \learn \bracks*{
                \HR(\tstart) +
                \sum_{\timealtalt=\tstart}^{\time-1}
                \B^\top \x(\timealtalt)
            }
        }
    } \label{eqn: disc reward}
\end{align}
We define the optimizer's optimal reward in both continuous and discrete time settings as follows:
\[
    \optrcont{\HR(\tstart)}{\tend} \defeq \max_{\x(t)} \rcont{\x(t)}{\HR(\tstart)}{\tend}\text{ and }
    \optrdisc{\HR(\tstart)}{\tend} \defeq \max_{\x(t)} \rdisc{\x(t)}{\HR(\tstart)}{\tend}.
\]
In continuous time, $\x(t)$ is defined for $\time\in[\tstart,\tend]$. In discrete time, $\x(t)$ denotes the sequence $\{\x(0),\x(1),\ldots,\x(\tend-1)\}$. To analyze these reward functions, we introduce the learner's \emph{potential function}:
$\Phi_\learn(Z) \defeq  \regconj(\learn Z)/\learn$ for any $Z \in \R^{\ynum}.$ We defer the proofs of this part in Appendix~\ref{appendix:reward-decomposition}.
\vspace{0.5em}
\\\textbf{$\blacktriangleright$The Continuous Benchmark}
\\ 
 Our first observation reveals that in continuous time, the optimizer's reward is fully determined by the initial and final states of the learner, rendering the intermediate trajectory irrelevant: 
\begin{restatable}{lemma}{rewardpotentialcont}
\label{lem:reward-potential-cont}
    For any continuous optimizer strategy $\xc: [\tstart, \tend] \to \xsp$, the total reward is given exactly by the drop in the learner's potential:
    $
        \rcont{\xc(t)}{\HR(\tstart)}{\tend}
        \;=\;
        \Phi_\learn\bigl(\HR(\tstart)\bigr) - \Phi_\learn\bigl(\HR(\tend)\bigr).
$\footnote{The proof follows since $\nabla \Phi_\learn(Z)=\choicemap_{\reg}(\learn Z)$, so by the chain rule 
$\frac{d}{dt}\Phi_\learn(\HR(t))
=\big\langle \nabla \Phi_\learn(\HR(t)),\,\dot{\HR}(t)\big\rangle
=\big\langle \ystrat{t},\,\matalt^\top \xc(t)\big\rangle .
$
Integrating on $[\tstart,\tend]$ gives the result. We defer the full proof for the continuous-time case to Theorem~\ref{lem:reward-potential-both} in Appendix~\ref{appendix:reward-decomposition}. The range for the total reward is characterized in Proposition \ref{prop:optimal-continuous-reward-range}, which recovers the reward sandwich \eqref{eq:reward sandwich} in the continuous-time setting.
}
\end{restatable}
\begin{remark}[Path Independence] \it\label{rem:path_independence}
    Crucially, Lemma~\ref{lem:reward-potential-cont} implies that the continuous reward depends on the strategy $\xc(t)$ only through its time-average $\bar{\x} \defeq \frac{1}{\tend} \int_{\tstart}^{\tend} \xc(t) \dd t$.
    Since $\HR(\tend) = \HR(\tstart) + \tend \matalt^\top \bar{\x}$, any two strategies with the same mean yield identical rewards.
    This collapses the infinite-dimensional variational problem of finding the optimal trajectory $\xc(t)$ into a finite-dimensional convex minimization over constant strategies $\x \in \xsp$:
\begin{align}
    \label{eq:cont_opt_problem}
        \optrcont{\HR(\tstart)}{\tend} \;=\; 
        \Phi_\learn\bigl(\HR(\tstart)\bigr)
        - \min_{\x \in \xsp} \Phi_\learn\bigl(\HR(\tstart) + \tend \matalt^\top \x\bigr).
\end{align}
\end{remark}
\paragraph{Efficient Computation \& Correction to Prior Work.}
The reduction of Remark~\ref{rem:path_independence} allows for
efficient optimization via Frank-Wolfe. Setting
$G(\x) \defeq \Phi_\learn(\HR(\tstart) + \tend \matalt^\top \x)$,
one verifies that $G$ is convex and $\beta$-smooth with
$\beta = {(\learn \tend \opnorm{\mat})^2}/{\scparam}$, yielding the following complexity.\textsuperscript{\hyperref[fn:one]{6},\hyperref[fn:two]{7}}
\stepcounter{footnote}\footnotetext[6]{\label{fn:one}
From Proposition \ref{prop: alpha strongly 1/alpha smooth},
if $\reg$ is $\scparam$-strongly convex then $\regconj$
is $(1/\scparam)$-smooth. The extra $(\learn\tend\opnorm{\mat})^2$
factor arises from the chain rule through the affine argument.
See Appendix~\ref{appendix:claim:smooth-of-composition}.}
\stepcounter{footnote}\footnotetext[7]{\label{fn:two}\it Correction to \cite{assos2024maximizing}:
\citet[Proposition~5]{assos2024maximizing}
analyze the MWU setting assuming smoothness constant $\beta=1$.
Our analysis shows $\beta = \Theta((\learn\tend\opnorm{\mat})^2)$,
which is not $O(1)$ for large horizons. Correcting for this, the
Frank-Wolfe complexity is
$\bigoh(nm\learn\tend^2\opnorm{\mat}^2/(\scparam\varepsilon))$
rather than $\bigoh(nm/\varepsilon)$. Recall that even if $\eta=T^{-\beta}$ for some $\beta\in[0,1]$, the runtime will be heavily horizon-dependent.}
\begin{restatable}[Frank-Wolfe Complexity
{\cite[Theorem~3.8]{bubeck2015convexoptimizationalgorithmscomplexity}}]{corollary}{fwcomplexity}
\label{cor:fw_complexity}
An $\epsilon$-approximate strategy $\hat{\x}$ with
$\optrcont{\HR(\tstart)}{\tend} - \rcont{\hat{\x}}{\HR(\tstart)}{\tend}
\le \epsilon$
can be computed in time
$\bigoh\!\left(nm\,{R^{2}\learn\tend^2\opnorm{\mat}^{2}}
/({\scparam\epsilon})\right),$
where $R$ is the diameter of $\xsp$.
\end{restatable}
\textbf{$\blacktriangleright$The Discretization Gap}\\
In discrete time, the optimizer cannot perfectly integrate the learner's potential.\footnote{\it
The challenges of optimal control differ substantially in continuous and discrete time. The continuous optimum $\optrcont{x_c(t)}{T}$ is obtained from a single static convex program in Equation~\eqref{eq:cont_opt_problem}. In contrast, the discrete optimum $\optrdisc{x_d(t)}{T}$ maximizes a cumulative Bregman sum, and thus becomes a sequential decision problem: the optimizer must balance minimizing the terminal potential with maximizing the path-dependent divergence term, i.e., volatility harvesting. Since the continuous optimum provides a lower bound, and since we later show that asymptotically simple fixed or alternating strategies already achieve the optimal upper bounds, we do not study the exact computation of global or local optima here. For this perspective, see \cite{ananthakrishnan2025learning,brown2023learning}.
} The error induced by the step size $\learn$, as we will see, manifests as a strictly non-negative ``bonus" for the optimizer, captured by the Bregman divergence.
\begin{restatable}[Discrete-Continuous Decomposition]{theorem}{mainreduction}
\label{thm:main-reduction}
    Let $\xd(\time)$ be a discrete strategy and $\xcp(t) = \xd(\lfloor t \rfloor)$ its piecewise-constant extension. The discrete reward satisfies:
    {\setlength{\abovedisplayskip}{1pt}
\setlength{\belowdisplayskip}{-3pt}
    \begin{equation} \label{eq:main-reduction-decomp}
        \rdisc{\xd(t)}{\hr(\tstart)}{\tend}
        \;=\;
        \underbrace{\rcont{\xcp(t)}{\HR(\tstart)}{\tend}\vphantom{\frac{1}{\learn}\sum_{t=\tstart}^{\tend-1}\breg_{\regconj}\bigl(\learn \hr(t+1),\, \learn \hr(t)\bigr)}}_{\text{Continuous Benchmark}}
        \;+\;
        \underbrace{\frac{1}{\learn}\sum_{t=\tstart}^{\tend-1} \breg_{\regconj}\bigl(\learn \HR(t+1),\, \learn \HR(t)\bigr)}_{\text{Discretization Gap }:=\mathrm{DG}(\xd,T) \ge 0}. 
    \end{equation}}
\end{restatable}
\vspace{0.25em}
\begin{remark}\label{remark:bound} \it
    The non-negativity of $ \breg_{\regconj}$ in Theorem~\ref{thm:main-reduction} implies that the discrete optimizer always outperforms the continuous benchmark: $\optrdisc{\zeroreward}{\tend} \ge \optrcont{\zeroreward}{\tend}$.
    Conversely, maximizing the decomposition \eqref{eq:main-reduction-decomp} yields an upper bound controlled by the maximum cumulative divergence:
    {\small\begin{equation}\label{eqn:upper-bound-optimal-discrete-reward-by-optimal-continuous-reward-plus-breg}
        \optrdisc{\zeroreward}{\tend}
        \;\le\;
        \optrcont{\zeroreward}{\tend}
        \;+\;
        \sup_{\xd(t)}\mathrm{DG}(\xd(t),T)
      \; \text{ and } \;  
     \mathrm{DG}(\xd(t),T) \;\le\;\frac{\learn}{2\scparam}\sum_{t=0}^{\tend-1}
\dnorm*{\A^\top \xd(t)}^2,
    \end{equation}}
where $\reg$ is $\scparam$-strongly convex with respect to $\norm{\cdot}$.\vspace{-1em} 
\end{remark}
\subsection[Surplus Toolbox: Benchmarks, Decomposition, and Decay]{{Surplus} Toolbox: Benchmarks, Decomposition, and Decay}
\label{section:exploitation}
To set the stage for our formal notion of {surplus}, we introduce a few definitions and remarks.
\vspace{-0.5em}
\paragraph{$\ast$ \emph{Learner's Viewpoint}. (See Appendix~{\ref{appendix:learner-viewpoint}} for details and proof)}
By path independence (Remark~\ref{rem:path_independence}) in continuous time, it suffices to consider a fixed optimizer strategy $\xfix\in\xsp$. Let $\vv \defeq \A^\top \xfix$ be the vector of expected payoffs. We denote the \emph{best-response value} as $\valfix \defeq \min_i \vv_i$, the \emph{suboptimality gaps} as $\gap_i \defeq \vv_i - \valfix \ge 0$, and the \emph{best-response set} as $\optset \defeq \{i \mid \gap_i = 0\}$ and $\optnum := |\optset|$.
The following lemma characterizes the learner's trajectory.
\begin{restatable}[Learner's limit]{lemma}{lemmainftrlfixedx}
    \label{lemmain:ftrl-fixed-x} $\y(t)$ converges to $\yopt\in\brep(\xfix)$ as $t\to\infty$. Moreover, strict convexity of $\reg$ implies that $\yopt$ is unique. When $\xfix=\xne$ is max-min, we have that $\yopt\in\argmin_{\y\in\conv(\brep(\xne))}\reg(\y)$, so the payoff converges to the value $\val$ although $\yopt$ need not be a min-max strategy.
\end{restatable}
\paragraph{$\ast$ \emph{Optimizer's Viewpoint}. (See Appendix~{\ref{appendix:optimizer-viewpoint}} for details and proof)}\vspace{-0.5em}
Returning to the optimizer's payoff, it is driven by the learner's transient mass on suboptimal actions.
\begin{restatable}[Optimizer's payoff]{lemma}{surplusgapdecomp}
    \label{lem:surplus-gap-decomp}
    $\payoff(\xfix,\y(t)) = \valfix + \sum_{i\notin\optset} \gap_i\, \y_i(t)$, for all $t\ge 0$.
\end{restatable}
\paragraph{$\ast$ \emph{The Cumulative {Surplus}}.}\vspace{-0.5em}
Then, we define the \emph{cumulative {surplus}} as the optimizer's excess payoff above the game value.
\begin{definition}
    \it
For a horizon $T>0$,
$
\expl(T)\defeq\int_{0}^{T}\big(\payoff(\x(t),\y(t))-\val\big)\,dt .
$
When the optimizer plays a fixed strategy $\xfix$, Lemma~\ref{lem:surplus-gap-decomp} yields the decomposition
\vspace{-0.25em}
\begin{align}
    \label{eqn:surplus_decomposition}
\boxed{
    \expl(T) = 
    \underbrace{T\,(\valfix-\val)\vphantom{\int_0^T \sum_{i\notin\optset}\gap_i\,\y_i(t)\,dt}}_{\vag} 
    + 
    \underbrace{\int_0^T \sum_{i\notin\optset}\gap_i\,\y_i(t)\,dt}_{\lag}
    \quad \tag{Surplus Decomposition}
}
\end{align}
The \textbf{Value Gap} $\vag$ is a linear penalty (since $\valfix \le \val$) that vanishes iff $\xfix$ is max-min. The \textbf{Learner Approximation Gap} $\lag \ge 0$ captures the profit extracted from the learner's suboptimal actions before convergence.
\end{definition}
\begin{remark}\it
\ref{eqn:surplus_decomposition} separates a static \emph{cost of commitment}, $\vag$, from dynamic \emph{gains from learning}, $\lag$. The former vanishes iff $\xfix$ is max-min. The latter measures the transient profit extracted from the learner’s suboptimal actions. Hence, under perfect knowledge, controlling {surplus} amounts to bounding $\lag$, namely, how quickly the regularizer drives $y_i(t)$ to zero.
\end{remark}
\noindent The following structural lemma offers a unified analysis for generic separable regularizers, acting as our ``speedometer,'' linking the geometry of the regularizer (through the inverse map $\phi$) to the decay rate of non-best-response probabilities.\footnote{%
Specializing Lemma~\ref{lem:bound-for-cont-ftrl} to standard regularizers,
in the supplement we show that for suboptimal weights $i \notin \optset$:
\textbf{(Negative Entropy)} yields \emph{exponential decay}
$y_i(t) \sim e^{-\learn t \gap_i}$;
\textbf{(Euclidean)} quadratic penalty yields \emph{finite-time decay}
$y_i(t) \sim [\text{const} - \learn t \gap_i]_+$;
\textbf{(Negative Tsallis)} for $q\in(0,1)$ yields \emph{polynomial decay}
$y_i(t) \sim (\text{const} + \learn t \gap_i)^{-1/(1-q)}$,
while $q>1$ allows sparse solutions similar to the Euclidean case.%
}
The proof of it is in Appendix~\ref{appendix:lem:bound-for-cont-ftrl}.
\begin{restatable}[Pointwise FTRL Bounds]{lemma}{lemboundforcontftrl}
    \label{lem:bound-for-cont-ftrl}
    Fix $\xfix$, step size $\learn > 0$, and a separable regularizer $\reg$. The FTRL response satisfies $y_i(t)=\phi(\kkt(t)-\learn t\, \vv_i)$ for a normalization scalar $\kkt(t)$. Then for any $t\ge 0$:
    \begin{description}[topsep=-1pt,itemsep=-0.5pt, leftmargin=1em]
        \item[$\star$\textit{Optimal Actions:}]For all $i \in \optset$, $y_i(t)$ are identical $a_t$ with $\frac{1}{\ynum} \le a_t \le \frac{1}{\optnum}$.
        \item[$\star$\textit{Suboptimal Actions:}]For all $j \notin \optset$:
        {\small \(
            \phi\big(\theta'(1/\ynum)-\learn t\,\gap_j\big)
            \;\le\; y_j(t) \;\le\;
            \phi\big(\theta'(1/\optnum)-\learn t\,\gap_j\big).
        \)}
    \end{description}
    \noindent Additionally, \textit{steep} regularizers yield asymptotic decay ($y_j(t)\to 0$), whereas \textit{non-steep} ones eliminate suboptimal actions in finite time $t \ge t_{\theta}^*\coloneqq\frac{\theta'(1/\optnum)-\theta'(0^+)}{\learn\gapmin}$.\vspace{-1em}
\end{restatable}
\section{Main Results}
\label{section:main-results}
\vspace{-4pt}
%
\subsection{Passive Extraction: The Inverse-Rate Law}
\label{subsec:inverse-rate}
\vspace{-3pt}
For our first result, fix a constant optimizer strategy $\xfix$.
Let $\optset$ be the set of best responses,
$\optnum=|\optset|$, $\ynum-\optnum = N_{\mathrm{sub}}$
the number of suboptimal actions, and
$\gapmin\le\gap_i\le\gapmax$ the payoff gaps of suboptimal actions.
Define the \emph{dual potential}
$V(u)\coloneqq\theta^*(\theta'(u))=u\,\theta'(u)-\theta(u)$%
\footnote{For standard kernels: Euclidean $V(u)=\tfrac12u^2$,
Tsallis-$q$ $V(u)=u^q$, entropic $V(u)=u$ (up to constants).},
and the positive potential drops $\overline{\Delta V}(T),\underline{\Delta V}(T)$:
\begin{align}
\overline{\Delta V}(T)
&\coloneqq \tfrac{\gapmax}{\gapmin}
\Bigl[V(1/\optnum)-\theta^*\!\bigl(\max\{\theta'(1/\optnum)
-\learn\gapmin T,\,\theta'(0^+)\}\bigr)\Bigr],
\label{eq:dv_high}\\
\underline{\Delta V}(T)
&\coloneqq \tfrac{\gapmin}{\gapmax}
\Bigl[V(1/\ynum)-\theta^*\!\bigl(\max\{\theta'(1/\ynum)
-\learn\gapmax T,\,\theta'(0^+)\}\bigr)\Bigr].
\label{eq:dv_low}
\end{align}

\begin{restatable}[Inverse-Rate Law]{theorem}{thmexploitfinite}
\label{thm:exploit-finite}
The cumulative surplus $\lag$ against a fixed strategy satisfies
\[
\frac{N_{\mathrm{sub}}}{\learn}\cdot\underline{\Delta V}(T)
\;\le\;\lag\;\le\;
\frac{N_{\mathrm{sub}}}{\learn}\cdot\overline{\Delta V}(T).
\tag{Exploitation Bound}
\]
\end{restatable}
\noindent The discrete-time version holds up to an upper-bound additive
$O(N_{\mathrm{sub}}\gapmax/\optnum)$ term
(Corollary~\ref{cor:discrete-exploit-finite}). The asymptotic regime is dictated by
$\theta'(0^+)$, yielding a sharp dichotomy.
\begin{restatable}[Steep vs.\ Non-Steep]{theorem}{thmexploitinfinite}
\label{thm:exploit-infinite}
Let $V_{\mathrm{bdry}}\coloneqq\theta^*(\theta'(0^+))$.
As $\learn T\to\infty$:
\begin{enumerate}[label=\textbf{(\Alph*)},
  leftmargin=*,itemsep=-2pt,topsep=-1pt]
\item \textbf{Non-steep} ($\theta'(0^+)>-\infty$,
  $V_{\mathrm{bdry}}<\infty$): suboptimal actions are
  eliminated in finite time $T^*<\infty$ and surplus
  \emph{saturates} to a constant in
  $\frac{N_{\mathrm{sub}}}{\learn}
  \bigl[\frac{\gapmin}{\gapmax}(V(1/\ynum)-V_{\mathrm{bdry}}),\,
  \frac{\gapmax}{\gapmin}(V(1/\optnum)-V_{\mathrm{bdry}})\bigr]$.
\item \textbf{Steep} ($\theta'(0^+)=-\infty$,
  $V_{\mathrm{bdry}}=0$): the learner approaches the
  boundary only asymptotically and surplus has an
  \emph{infinite tail}, converging to
  $\frac{N_{\mathrm{sub}}}{\learn}
  \bigl[\frac{\gapmin}{\gapmax}V(1/\ynum),\,
  \frac{\gapmax}{\gapmin}V(1/\optnum)\bigr]\pm o(1)$.
\end{enumerate}
\end{restatable}
We provide an expanded discussion in Appendix \ref{app subsec:inverse-rate}. The formal statements and proofs of Theorems \ref{thm:exploit-finite} and \ref{thm:exploit-infinite} are presented as Theorems \ref{app thm:exploit-finite} and \ref{app thm:exploit-infinite}, respectively, with concrete examples available in Appendix \ref{appendix:app:examples-gaps}. Therefore, if the optimizer simply plays max-min strategy $\xfix=\xne$, he will receive as surplus exactly $\lag$.

\vspace{-4pt}
\subsection{Active Extraction: The Alternating Trap}
\label{subsec:alternating}
\vspace{-2pt}

As discussed in the introduction, a fixed strategy fails whenever $N_{\mathrm{sub}}=0$
(e.g., RPS against the max-min). Here, we show that the optimizer
can recover again optimal surplus by \emph{actively engineering}
the learner's trajectory.
\begin{restatable}{theorem}{thmavgmain}
\label{thm:avg-main}
Let $\A\in\R^{\xnum\times\ynum}$ have i.i.d.\
$\mathrm{Unif}[-1,1]$ entries. Fix any
$\scparam$-strongly convex separable regularizer $\reg$
with curvature bound
$M=\sup_{a\in[\delta,1-\delta]}\theta''(a)$, and
{step size
$\learn\le\frac{\scparam}{\cnorm\sup_{u\in[-1,1]^{\ynum}}\dnorm{u}}
(\frac{1}{\ynum}-\delta)$ where $\cnorm = \sup_{u\neq 0} (\norm{u}_\infty / \norm{u})$.}
Then with probability at least
$1-\frac{\xnum!\ynum!}{(\xnum+\ynum-1)!}-\frac{1}{\xnum\ynum}$,
there exists an \emph{alternating trap strategy}
$\xd(t)$ such that for every horizon $T$,
\[
\rdisc{\xd(t)}{0}{T}
\;\ge\;
T\,\val \;+\; \learn T\cdot C_{\A}^{\reg},
\qquad
C_{\A}^{\reg}
\;\defeq\;
\frac{1}{2M(\xnum\ynum)^6}.
\]
\end{restatable}
\begin{remark}
    \it At a high level, the following presentation has the same general blueprint as in \citet[Proposition 11]{assos2024maximizing}, but several technical ingredients must be refined substantially in order to extend the analysis to arbitrary regularizers. The first issue is to obtain a quantitative lower bound ensuring that two distinct learner best responses to the max-min strategy are not ordered by weak dominance. Rather than taking such a condition as an assumption, we show that, with high probability, a uniform polynomial lower bound holds for random zero-sum games. This is a critical ingredient: although a related property is implicitly used in \cite{assos2024maximizing}, the relevant game-dependent magnitude is never made explicit there.
A second key ingredient is a uniform control of the curvature of the optimization landscape under an arbitrary regularizer. This is precisely where our proof departs from prior entropy-specific arguments. In this way, the negative-entropy case emerges as a special instance of a broader proof template that applies uniformly to all regularizers under consideration.
\end{remark}
\vspace{-0.5em}
\paragraph{Our Proof Sketch.}
The formal statement and proof are provided in Appendix~\ref{appendix:app_alternating}. Below, we provide a proof sketch.
The \emph{key idea} is to construct two perturbed strategies
$x',x''$ with $(x'+x'')/2=\xne$ that flip the strict
preference between two best responses $i\neq j$, locking
the learner in perpetual oscillation:
\vspace{-1em}
\paragraph{Step 1: Two high-probability events on $\A$.}
With high probability over the random payoff matrix $\A$, there exist a max-min optimizer strategy $\xne$ such that two distinct learner best responses $i\neq j\in\brep(\xne)$, and an index $\ell\in\supp(\xne)$ such that
$e_\ell^\top \A e_i = \A_{\ell i} \neq \A_{\ell j} = e_\ell^\top \A e_j.$
We record this support separation property as the event $\eventtwo$ (Assumption~\ref{assump:non-identify-on-support}) and show that it holds with high probability in Lemma~\ref{lem:assump:non-identify-on-support-holds-w.h.p.}.
We also work on an anti-degeneracy event $\eventgap$, which guarantees a uniform polynomial lower bound on $\abs{A_{pq}-A_{pq'}}$ whenever $A_{pq}\neq A_{pq'}$.
Conditioning on $\eventtwo\cap \eventgap$, we fix such $\xne$, $i\neq j$, and $\ell$ for the remainder of the proof.
\vspace{-1em}
\paragraph{Step 2: An alternating perturbation construction.}
Using the fixed $(\xne,i,j,\ell)$, we construct two perturbed strategies $x',x''$ satisfying $(x' + x'') / 2 = \xne$,
such that the strict preference between $i$ and $j$ flips $(A^\top x')_i > (A^\top x')_j,
(A^\top x'')_j > (A^\top x'')_i .$
The perturbation is defined in \eqref{eq:push-to-ek} and its properties are proved in Lemma~\ref{lem:push-to-ek-properties}.
We then define an alternating optimizer strategy $\xd(t)$ that plays $x'$ on even rounds and $x''$ on odd rounds. This keeps the even-round time-average at $\xne$ while repeatedly switching the unique best response between $i$ and $j$.
Let $\vv \defeq \A^\top \xne$ and $\vv' \defeq \A^\top x'$.
Using $(x'+x'')/2=\xne$ and the max-min property of $\xne$, we lower bound the optimizer's payoff over each pair of rounds $(2s,2s+1)$ by
\begin{align}
2\,\val \;+\; \Bigl(\inner{\vv',y_{2s}}-\inner{\vv',y_{2s+1}}\Bigr).
\label{eq:two-round-decomposition}
\end{align}
Thus, it suffices to prove a uniform lower bound on the one-step difference
$\inner{\vv',y_{2s}}-\inner{\vv',y_{2s+1}}$.
\vspace{-1em}
\paragraph{Step 3. Interpolation and a variance identity.}
To control the difference term, we interpolate between the two consecutive FTRL iterates $y(r)=\choicemap_{\reg}\!\bigl(-2\learn s\, \vv - r\,\learn\, \vv'\bigr), r\in[0,1]$, so that $\inner{\vv', y_{2s}} - \inner{\vv', y_{2s+1}} = -\int_0^1 \frac{d}{dr}\inner{\vv',y(r)} \dd r$.
Differentiating the KKT conditions along this path yields the variance identity
\begin{align*}
\frac{d}{dr}\inner{\vv',y(r)}
\;=\;
-\learn\,W_r\,\Var_{\pi_r}(Z_r),
\end{align*}
where $w_r(\ell)\propto 1/\theta''(y_\ell(r))$, $W_r\defeq\sum_{\ell} w_r(\ell)$, $\pi_r(\ell)\defeq w_r(\ell)/W_r$, and $Z_r\defeq v'_I$ for $I\sim \pi_r$ (curvature-weighted sampling).
A two-point lower bound on variance then isolates the distinguished coordinates $i\neq j$:\vspace{-1.5em}
\begin{align}\label{eqn:roadmap-theta''-lower bound}
W_r\,\Var_{\pi_r}(Z_r)
\;\ge\;
\frac{(v'_i-v'_j)^2}{\theta''(y_i(r))+\theta''(y_j(r))}.
\end{align}
We formalize this step in Lemma~\ref{lem:variance-identity}.
\vspace{-1em}
\paragraph{Step 4. Uniform control of curvature.}
We next show that along the interpolation, the learner keeps nontrivial mass on the two distinguished actions $i$ and $j$: both $y_i(r)$ and $y_j(r)$ stay uniformly bounded away from 0 by the choice of our step size $\learn$ (Lemma~\ref{lem:path-uniform-lb-yi-yj}). This prevents the curvature $\theta''(y_i(r))$ and $\theta''(y_j(r))$ from blowing up along the path and stay uniformly upper bounded by some constant $M$.
Therefore, the denominator in \eqref{eqn:roadmap-theta''-lower bound} is at most $2M$, which yields a constant lower bound    \vspace{-0.5em}
\begin{align}
    \inner{\vv', y_{2s}} - \inner{\vv', y_{2s+1}}
    = \int_0^1 \learn\, W_r\, \Var(Z_r)\, dr
    \ge \learn \frac{(v'_i - v'_j)^2}{2M} > 0.
    \label{eq:roadmap-uniform-lower-bound}
\end{align}
\paragraph{Step 5. Sum the uniform one-step lower bound.}
We lower bound the payoff gap by combining the perturbation property from Lemma \ref{lem:push-to-ek-properties}, $v'_i-v'_j \;=\; x_{\ell}\bigl(\A_{\ell i}-\A_{\ell j}\bigr),$
with the anti-degeneracy event $\eventgap$ from Step~1, which yields $v'_i-v'_j \ge C_{\A}>0$.
Together with the uniform curvature bound $M$ from Step 4, this gives the one-step lower bound $\inner{\vv',y_{2s}}-\inner{\vv',y_{2s+1}}
\ge 
\learn C_{\A}^{\reg}$
as stated in \eqref{eq:roadmap-uniform-lower-bound}, where $C_{\A}^{\reg} = (C_{\A})^2 / (2M)$.
Substituting into \eqref{eq:two-round-decomposition} and summing over $T/2$ pairs yields
\begin{align*}
\sum_{t=0}^{T-1}\xd(t)^\top \A y_t
\;\ge\;
T\,\val + \learn T\,C_{\A}^{\reg}.
\end{align*}
The stated success probability follows by a union bound over $\eventtwo\cap\eventgap$. 
\vspace{-6pt}
\subsection{The Price of $\varepsilon$-Best-Response}
\label{subsec:pbr}
\vspace{-5pt}

Results~1 and~2 quantify how much surplus a clairvoyant
optimizer can extract. We close by inverting this
perspective: \emph{how much surplus must the learner
inevitably surrender simply to identify her approximate $\varepsilon$-best response?}

Under classical minimax analysis, the choice of regularizer
is largely immaterial---all strongly convex choices yield
$\Theta(\sqrt{T})$ regret, with the entropic regularizer
preferred on the simplex for its dimension dependence.
This apparent equivalence breaks once we view learning as
\emph{information acquisition} under adversarial asymmetry.
The optimizer's surplus is not merely extracted; it is a
\emph{tuition fee} the learner pays to filter noise and
converge to $\yopt$. We formalize the minimum such fee from learner's perspective as
the \emph{Price of Best-Response}:
\[
\mathcal{C}(\varepsilon)
\;\defeq\;
\inf_{\learn,\,t}\;\sup_{\xd(\cdot),\,\A}
\Bigl\{\,\expl_{\reg}^{(\learn,t)}
\;\Big|\;
\|y(t)-\yopt\|_1 \le \varepsilon \,\Bigr\}.
\tag{PBR}
\]
In this regime, we provide the following results:
\begin{restatable}{theorem}{learningCost}
\label{thm:learning-cost}
For any sufficiently small $\varepsilon > 0$,
\(
\mathcal{C}(\varepsilon)
\;=\;
\Omega(
  \theta'(1/\optnum)
  -\theta'(\tfrac{\varepsilon}{2(\ynum-\optnum)})).
\)
\end{restatable}

\begin{observation}[Euclidean vs.\ Entropic tuition]
\label{obs:tuition} \it
As $\varepsilon\to 0$: Euclidean regularization
($\theta'(u)=u$) yields
$\mathcal{C}_{\mathrm{Euc}}(\varepsilon)=\Theta(1)$---a
\emph{finite}, one-time fee, independent of precision.
Entropic regularization ($\theta'(u)\approx\log u$) yields
$\mathcal{C}_{\mathrm{Ent}}(\varepsilon)=\Theta(\log(1/\varepsilon))
\to+\infty$---an \emph{unbounded} rent that grows as the
learner seeks higher precision.
\end{observation}
\noindent This reveals a fundamental ``paradox'': the very
property that makes entropic regularization
robust---infinite boundary gradients that prevent
over-committing to pure strategies---is precisely what
makes it \emph{infinitely saturatable} in the limit.
\emph{Strict safety and efficient best-response
identification are structurally incompatible.}
The full analysis is in Appendix~\ref{appendix:app_pbr}, and Appendix~\ref{appendix:azuma} records the corresponding realized-payoff concentration bounds for the bandit-feedback setting.

\paragraph{Epilogue.}
Taken together, our three results describe a single phenomenon from complementary angles: \emph{how a clairvoyant optimizer extracts utility against an FTRL learner, and how the regularizer governs that extraction}. The experiments in Appendix~\ref{appendix:experiment} indicate that the predicted rates are tight up to constants. Natural next steps include adaptive step sizes and movement-constrained learners, partial-feedback models such as bandit observations against optimal learning rates, and extensions beyond full-information zero-sum games to richer strategic settings. Another intriguing direction is to understand how much surplus can still be recovered when the optimizer knows only the learner's algorithmic profile, but not the payoff matrix \(\A\). A broader discussion is deferred to Appendix~\ref{appendix:future-work}.



\printbibliography
\newpage
\appendix
\newpage
{ \section{Broader Impact and Limitations}
\label{appendix:broader-impact}
\vspace{-0.75em}
\paragraph{Broader Impact.}
This work contributes to the mathematical foundations of
multi-agent learning, with implications for the design,
evaluation, and governance of learning algorithms in strategic
environments.

On the positive side, our results provide \emph{explicit,
quantitative} warnings to practitioners who deploy FTRL-based
algorithms---including Multiplicative Weights, Online Gradient
Descent, and Tsallis-regularized methods---in settings where
an informed counterpart may anticipate their dynamics. The
Inverse-Rate Law (Section \ref{subsec:inverse-rate}) and the Price
of \(\varepsilon\)-Best-Response (Section \ref{subsec:pbr}) together furnish a
concrete diagnostic checklist: how many suboptimal actions
does the learner maintain, what step size is used, and how
steep is the regularizer? These are, in principle, measurable
quantities, and our bounds translate them into explicit
surplus estimates. In particular, the geometric dichotomy
between steep and non-steep regularizers offers actionable
guidance: when high-precision convergence is required in the
presence of a strategic counterpart, non-steep regularizers
(e.g., Euclidean projection) incur a \emph{finite} structural
cost, whereas steep regularizers (e.g., entropic/MWU) incur
a diverging one.

Such awareness is especially consequential in domains where
algorithmic decisions interact repeatedly with strategic human
or institutional actors, including automated market-making and
algorithmic trading \citep{yang2012behavior}, repeated
procurement and contract design \citep{kolumbus2024contracting},
and multi-agent reinforcement learning systems in which one
agent has partial knowledge of another's policy class. In such
settings, understanding the structural cost of learning---not
just worst-case regret or asymptotic convergence---is essential
for responsible system design.

On the negative side, we acknowledge that our results also
provide a \emph{constructive} toolkit: the alternating trap
strategy of Theorem~\ref{thm:avg-main} and the fixed-strategy
surplus bounds of Theorem~\ref{thm:exploit-finite} could, in
principle, be used by a strategic agent to improve its own
payoff at the expense of a learning counterpart. This ethical
concern is especially relevant in high-stakes settings, where
even modest but repeated utility extraction may lead to unfair
or systematically asymmetric outcomes for less informed
participants. We believe, however, that transparency about
such vulnerabilities---rather than obscuring them---is the
more responsible path, because it enables scrutiny,
robustification, and, where appropriate, institutional or
regulatory safeguards. The present work is purely theoretical
and does not introduce a deployable system or recommend
adversarial use in practice.\hrule
\vspace{-0.5em}
\paragraph{Limitations.}
We highlight four boundaries of the current analysis that
represent natural directions for future work.

\begin{enumerate}[leftmargin=*, itemsep=2pt, topsep=0pt]

\item \textbf{Zero-sum and separable structure.}
Our main results (Theorems~\ref{thm:exploit-finite}
and~\ref{thm:avg-main}) are established for two-player
zero-sum matrix games with separable regularizers. While
this is the canonical adversarial benchmark, extension to
\emph{general-sum} games or \emph{non-separable} regularizers
(e.g., matrix-entropy or spectral regularizers) requires new
tools beyond those developed here.
\item \textbf{Constant step size.}
Our results are established for a \emph{constant} step size \(\eta\) throughout the horizon \(T\). Although adaptive or decaying schedules, such as \(\eta_t=\Theta(1/\sqrt{t})\), are theoretically attractive due to their noise-damping and stabilizing effects, constant or non-vanishing step sizes remain central in practice because of their simplicity, scalability, and strong finite-time behavior. In particular, vanishing schedules may incur long transients, whereas constant step sizes often reach near-stationary regimes more quickly. Extending our framework to adaptive and decaying schedules remains an important direction for future work.

\item \textbf{Full-information feedback.}
The main body of the paper operates under the
\emph{full-information} (full-feedback) model, where the
learner observes the entire payoff vector at each round.
We extend the analysis to \emph{bandit-feedback} in
Appendix~\ref{appendix:azuma}, but the bounds there apply
under multiplicatively suboptimal constants; a tight
characterization of surplus under optimal bandit algorithms
remains open.

\item \textbf{Clairvoyance assumption.}
The optimizer is assumed to know the learner's update rule
$\mathcal{L}$, step size $\eta$, and regularizer $\theta$
exactly. This ``white-box'' model is natural for theoretical
lower bounds but may be unrealistic in practice. Robustness
of the surplus extraction results to partial or approximate
knowledge of the learner's algorithm is not addressed here.

\end{enumerate}

\noindent None of these limitations undermines the core
contribution: showing that surplus saturation is a
\emph{structural property} of the FTRL family in zero-sum
games, encoded in the geometry of the regularizer rather
than in any external economic mechanism.}
\section{Related Work}
\label{appendix:app_literature}
In this section, we provide a detailed overview of the literature surrounding the strategic anticipation of learning algorithms, placing our contributions within the broader landscape of mechanism design and game theory.

\paragraph{Strategic Anticipation in Mechanism Design and Principal--Agent Problems.}
The paradigm of {exploiting} a learning agent's predictability originated in the field of algorithmic mechanism design. The seminal work of \citet{braverman2018selling} and subsequently \citet{cai2023selling} established that in repeated auctions, a strategic seller can leverage the ``mean-based'' property of standard no-regret algorithms (such as EXP3 and MWU) to  {extract revenue approaching the entire social surplus}. This line of work highlights that without specific format constraints (e.g., no overbidding), the learner's regret guarantees are insufficient to protect their utility. \citet{collina2023efficientpriorfreemechanismsnoregret} further refined this by addressing adversarial states and policy-regret goals, proposing ``stable'' policies to mitigate alignment assumptions.

Beyond auctions, this perspective has been generalized to Principal--Agent interactions. \citet{lin2024generalized} provide a unified framework for general Stackelberg settings, reducing the repeated interaction to a one-shot approximate best response problem and bounding the principal's utility based on the agent's regret (or swap-regret) guarantees. In the specific context of dynamic contracting, \citet{kolumbus2024contracting} demonstrate how a principal can induce a ``free-fall'' in the agent's action space by switching from a linear to a zero-reward contract, effectively securing rewards at zero cost. 
\paragraph{Strategizing in Repeated Games.}
Parallel to mechanism design, recent research has examined rational optimizers in general repeated games. \citet{deng2025strategizingnoregretlearners} analyze the limits of what an optimizer can force against a no-regret learner, showing that while ``mean-based'' learners can be {exploited} to yield strictly more than the Stackelberg value, learners with no-swap-regret can effectively cap the optimizer's payoff. In the Bayesian setting, \citet{mansour2022strategizing} extend this inquiry to games with private types, introducing the concept of ``polytope swap regret'' to characterize the minimum payoff an optimizer can guarantee against a learning opponent. 

\paragraph{\emph{Comparison with our work. \hspace{-1.5em} }} The works above share a common thread: in rich economic environments---auctions, contracts, Bayesian games---a clairvoyant principal can steer a learning agent to yield substantially more than any static benchmark, often achieving linear surplus $\Omega(T)$ above the Stackelberg value. Our work departs from this literature in two respects:
\begin{itemize}
\item First, we operate in the simplest possible strategic environment: a \emph{zero-sum matrix game}, where the Stackelberg value coincides with the minimax von Neumann value $V^\star$ and no asymmetric economic structure is present. This is precisely the setting where one might expect learning to be safest---and yet we show that structural surplus persists.
\item Second, and more importantly, we focus on the \emph{sublinear regime}. The reward sandwich 
{$T \cdot V^\star \le \mathrm{Rew}_T \le T \cdot V^\star + \gregret_{\mathcal{L}}^{\learn}(T)$}
caps the optimizer's gain at $o(T)$ when the learner employs standard no-regret algorithms with $\eta = \Theta(T^{-1/2})$. Rather than asking whether linear surplus is achievable---it is not, in this setting---we ask whether a clairvoyant optimizer can \emph{saturate} the sublinear ceiling $\Theta(\sqrt{T})$. This is a strictly harder and more delicate question, requiring a fine-grained understanding of how the regularizer's geometry governs the learner's convergence rate. Our answer is affirmative, and the saturation conditions are encoded entirely in the structure of the FTRL algorithm itself---not in any external economic mechanism.
\end{itemize}

In other words, both the works on general-sum games that obtain surplus \(\Omega(T)\) and the works on zero-sum games that obtain surplus \(\Omega(T^{1/2})\)  revolve around the same central question:\begin{center} \emph{How can one optimally anticipate a no-regret learner so as to maximize one’s own utility?}
\end{center}

\paragraph{The $o(T)$ Regime.} It is important to note that $\Omega(T)$ {exploitation is not inevitable under all conditions}. \citet{zhao2026noregretstrategicallyrobustlearning} identify conditions in single-item auctions where monotone bidding strategies combined with gradient feedback render learners ``strategically robust,'' preventing the auctioneer from exceeding Myerson's optimal revenue. The authors provided results where saturated surplus is tightly $\Theta(\sqrt{T})$. {The most closely related thread to our work involves optimizers that explicitly anticipate the learner's update rule. \citet{assos2024maximizing} adopt this ``white-box'' view, designing optimal strategies against Replicator Dynamics (the continuous-time analogue of MWU) and providing discrete-time guarantees for MWU. Our work aligns with this ``control-theoretic'' approach but generalizes the scope significantly. Instead of focusing solely on MWU/Replicator dynamics, we develop a unified theory for the entire FTRL family. Crucially, we introduce the geometric dichotomy between Steep and Non-Steep regularizers, showing that maximal {extractability} is not merely a feature of multiplicative updates but a structural consequence of boundary curvature.
}
\vspace{-1em}
\paragraph{Technical comparison with \citet{assos2024maximizing}}
Assos et al. establish a strong zero-sum result for MWU/Replicator Dynamics, which corresponds to the negative-entropy regularizer in our FTRL framework. The argument is fundamentally tied to the entropy structure of MWU and does not provide a theory for general FTRL regularizers. In contrast, our analysis works directly with the FTRL choice map
$y(t)=\choicemap_{\reg}(\learn \HR(t)),$
and therefore covers entropy/MWU, Tsallis, Euclidean, and other separable regularizers within a single framework. This also lets us expose how the regularizer geometry changes the optimizer's surplus. In particular, while the analysis of Assos et al. does not compare different no-regret learners, our \(\phi=(\theta')^{-1}\)-based proof shows that the surplus is governed by the curvature and boundary behavior of \(\reg\): steep regularizers such as entropy and Tsallis keep all actions in the interior and approach the boundary only asymptotically, whereas non-steep regularizers such as Euclidean may eliminate suboptimal actions in finite time.

A second difference concerns how the fixed-strategy surplus is organized. 
Assos et al.'s proof of Proposition 7 contains an implicit separation between the contribution of best-response coordinates and the exponentially vanishing contribution of non-best-response coordinates in the MWU formula. However, this separation is used inside the entropy-specific closed-form calculation, rather than being formulated as a general structural decomposition for arbitrary constant optimizer strategies and general FTRL choice maps. Our analysis makes this separation explicit through the decomposition from Equation \eqref{eqn:surplus_decomposition}
\[
    \expl(T)=\vag+\lag.
\]
This decomposition holds for every constant optimizer strategy \(\hat x\). When \(\hat x=x^\star\) is max-min, the value-gap term vanishes, so the surplus is exactly the learner-approximation gap. By contrast, if \(\hat x\) is not max-min, then the value-gap term is negative and linear in \(T\). Since the learner-approximation gap is only \(O(1/\eta)\), this negative linear term dominates whenever \(1/\eta=o(T)\). Thus, among fixed strategies, the max-min strategy is the only one that can sustain positive surplus asymptotically, and the remaining surplus comes precisely from the learner's transient mass on suboptimal actions.

Finally, for alternating optimizer strategies, Assos et al. prove an \(\Omega(\eta T)\) lower bound under a structural condition, with a game-dependent constant that may depend on the payoff matrix \(A\). Our contribution is not only to extend the proof architecture beyond MWU, but also to quantify when this alternating construction is meaningful. In their result, the constant hidden in the \(\Omega(\eta T)\) term is game-dependent, and the argument does not verify that this constant is bounded away from zero for a typical game, nor does it provide an explicit rate. In contrast, we show that for random zero-sum games the required alternating-trap condition holds with high probability, and we obtain an explicit polynomial lower bound on the corresponding game-dependent constant. Thus the lower bound is not merely qualitative: the \(\Omega(\eta T)\) term has a nondegenerate coefficient under high-probability random-game events. The randomization in our theorem is used only to certify these structural events.

\paragraph{Dynamics-Specific Analysis.}
The technical foundation for our discrete-continuous bridge is rooted in the work of \citet{kwon2014continuoustimeapproachonlineoptimization}, who formalize the disparity between continuous and discrete time to provide a unified treatment of FTRL regret through potential-based decompositions. Regarding the stability of the learning process itself, \citet{giannou2021survivalstricteststableunstable} provide a comprehensive analysis of FTRL dynamics in multi-player games, establishing an equivalence between the stochastic stability of equilibria and the strictness of best responses. While their work focuses on convergence and stability rather than adversarial surplus extraction, it provides the dynamical foundations upon which our  surplus-extraction analysis is built.


\section{Details Omitted from Section \ref{section:preliminaries} (Preliminaries)}
\label{appendix:app_preliminaries}
\subsection{Definitions and Notation}
\label{appendix:definitions}
\begin{definition}[Simplex] \it
    For a positive integer $d\in\N$, we define the probability simplex in $\R^d$ as
    \begin{align*}
            \simplex^d\;:=\;\setdef{z\in\R^d}{\sum_{i=1}^d z_i=1,\ z_i\ge 0,\ \forall i\in[d]}.
    \end{align*}
\end{definition}

\begin{definition}[Best-response set]\label{def:best-response} \it
For $\x\in\xsp$, define the learner's best-response set
\begin{align*}
\brep(\x)
\;\defeq\;
\argmin_{i\in[\ynum]}\ \inner{\x,\,\A e_i}
\;\subseteq\;
[\ynum].
\end{align*}
\end{definition}
This definition is equivalent to the one used in Section~\ref{section:exploitation},
\begin{align*}
\brep(\x)\;\defeq\;\{i\in[\ynum]\mid \gap_i=0\},
\end{align*}
since $\gap_i=0$ holds exactly for learner actions $i$ attaining the minimum payoff against $\x$.

\begin{definition}[Support] \it
    For a mixed strategy $\x\in\xsp$, define its support as the index set
    \[
    \supp(\x)\;:=\;\setdef{k\in[\xnum]}{\x_k>0}\ \subseteq [\xnum].
    \]
\end{definition}

\subsection{Omitted Convex Analysis Preliminaries}
\label{appendix:omitted definitions and properties}
First, we introduce the definitions of strong convexity and smoothness.
\begin{definition}\label{def:strongly-convex-and-smooth} \it
    Let $\fgen \from \ysp \to \R$ be a differentiable function and let $\norm{\cdot}$ be a norm on the space.
    For $\scparam > 0$ and $\smoothparam > 0$, we define:
    \begin{enumerate}
        \item $\fgen$ is $\scparam$-strongly convex with respect to $\norm{\cdot}$ if, for all $\y_1, \y_2 \in \ysp$,
        \[
            \fgen(\y_1) \geq \fgen(\y_2) + \inner{\grad \fgen(\y_2),\, \y_1 - \y_2}
            + \frac{\scparam}{2} \norm{\y_1 - \y_2}^2.
        \]

        \item $\fgen$ is $\smoothparam$-smooth with respect to $\norm{\cdot}$ if, for all $\y_1, \y_2 \in \ysp$,
        \[
            \fgen(\y_1) \leq \fgen(\y_2) + \inner{\grad \fgen(\y_2),\, \y_1 - \y_2}
            + \frac{\smoothparam}{2} \norm{\y_1 - \y_2}^2.
        \]
    \end{enumerate}
\end{definition}

Strong convexity is equivalent to smoothness of the conjugate.
\begin{proposition}[{\citet[Proposition 12.60]{rockafellar1998variational}}] \label{prop: alpha strongly 1/alpha smooth}
    Let $\fgen \from \ysp \to \R \cup \setof{+\infty}$ be proper and lower semi-continuous. Then for $\scparam > 0$, the following are equivalent:
    \begin{itemize}
        \item $\fgen$ is $\scparam$-strongly convex with respect to $\norm{\cdot}$.
        \item $\fgen^*$ is differentiable and $\grad \fgen^*$ is $1 / \scparam$-Lipschitz.
        \item $\fgen^*$ is $\smoothparam = 1 / \scparam$-strongly smooth with respect to the dual norm $\dnorm{\cdot}$ on $\dualysp$.
    \end{itemize}
\end{proposition}
Therefore, since our regularizers are proper and lower semicontinuous by definition, they satisfy the assumptions of Proposition~\ref{prop: alpha strongly 1/alpha smooth}.
\begin{proposition}[{\citet[Section 2.7]{shalev2012online}}]
    Under the assumption of Proposition~\ref{prop: alpha strongly 1/alpha smooth}, the choice map $\choicemap_{\fgen}$ is given by the gradient of the conjugate:
    \begin{align*}
        \choicemap_{\fgen}(\y) = \grad \fgen^*(\y),\qquad \forall \y\in\dualysp.
    \end{align*}
\end{proposition}

We next introduce the Bregman divergence induced by a convex potential, which will serve as our basic geometric tool when we move from continuous to discrete time doing discretization.
\begin{definition}[Bregman Divergence] \label{def:bregman-divergence} \it
    Let $\fgen \from \ysp \to \R$ be a differentiable convex function. The Bregman divergence induced by $\fgen$ is defined as
    \begin{align*}
        \breg_\fgen(\y_1, \y_2) \defeq \fgen(\y_1) - \fgen(\y_2) - \inner*{\grad \fgen(\y_2),\, \y_1 - \y_2},
    \end{align*}
    for all $\y_1, \y_2 \in \ysp$.
\end{definition}
Much of our analysis will take place in the dual space $\dualysp$. In this context, it is natural to consider the Bregman divergence induced by the conjugate function $\fgen^*$, which measures deviation in $\dualysp$ relative to its linearization.
Building on the equivalence between strong convexity and smoothness, we now derive a useful upper bound on the Bregman divergence of the conjugate.
\begin{proposition} \label{prop: upper bound Bregman divergence}
    If $\fgen$ is $\scparam$-strongly convex with respect to $\norm{\cdot}$, then for all $z_1, z_2 \in \dualysp$,
    \begin{align*}
        \breg_{\fgen^*}(z_1, z_2) \le \frac{1}{2\scparam} \dnorm{z_1 - z_2}^2.
    \end{align*}
\end{proposition}

\begin{proof}
    Since $\fgen$ is $\scparam$-strongly convex, Proposition~\ref{prop: alpha strongly 1/alpha smooth} implies that $\fgen^*$ is $\smoothparam = 1 / \scparam$-strongly smooth. By the definition of Bregman divergence and smoothness (Definition~\ref{def:strongly-convex-and-smooth}),
    \begin{align*}
        \breg_{\fgen^*}(z_1, z_2)
        = \fgen^*(z_1) - \fgen^*(z_2) - \inner{\grad \fgen^*(z_2),\, z_1 - z_2}
        \le \frac{1}{2\scparam} \dnorm{z_1 - z_2}^2.
    \end{align*}
\end{proof}

\subsection{KKT Derivation of the Kernel Choice Map}
\label{appendix:kkt-derivation}
Fix $z\in\R^m$ and recall \eqref{eq:choice-map}
\begin{align}
    \choicemap_{\reg}(z)\;\defeq\;\argmax_{y\in\ysp}\,\{\inner{z,y}-\reg(y)\}.
\end{align}
Assume $\reg$ is a separable regularizer on $\ysp$ with kernel $\theta$.

\paragraph{KKT conditions.}
Introduce $\kkt\in\R$ for the equality constraint $\sum_i y_i=1$ and $\nu\in\R_+^m$ for the inequalities $y_i\ge 0$.
The Lagrangian is
\begin{align}\label{eq:appendix:lagrangian}
    \mathcal{L}(y,\kkt,\nu)
    \;\defeq\;
    \langle z,y\rangle-\reg(y)
    +\kkt\Bigl(1-\sum_{i=1}^m y_i\Bigr)
    +\sum_{i=1}^m \nu_i y_i .
\end{align}
At the maximizer $y^*=\choicemap_{\reg}(z)$, the KKT system is
\begin{align}
    &y^*\in\ysp,\quad \nu\in\R_+^m, \label{eq:appendix:kkt-feas}\\
    &0\in\partial_{y_i}\mathcal{L}(y^*,\kkt,\nu)
    \;\;\Longleftrightarrow\;\;
    \theta'(y_i^*)=z_i+\kkt+\nu_i,\qquad i\in[m], \label{eq:appendix:kkt-stationary}\\
    &\nu_i\,y_i^*=0,\qquad i\in[m]. \label{eq:appendix:kkt-comp}
\end{align}
Since $\reg$ is strictly convex on $\ysp$, the maximizer $y^*=\choicemap_{\reg}(z)$ is unique.
By \eqref{eq:appendix:kkt-comp}, if $y_i^*>0$ then $\nu_i=0$ and $\theta'(y_i^*)=\kkt+z_i$.
If $y_i^*=0$, then $\nu_i\ge 0$ and \eqref{eq:appendix:kkt-stationary} implies $z_i+\kkt\le \theta'(0^+)$.
If $y_i^*=1$, then \eqref{eq:appendix:kkt-stationary} implies $z_i+\kkt\ge \theta'(1)$.
Recalling the definition of $\phi=\trunc[(\theta')^{-1}]$, \eqref{eq:appendix:kkt-stationary}--\eqref{eq:appendix:kkt-comp} yield Equation \eqref{eq:kkt_structure} we need:
\begin{align}\label{eq:appendix:Qh-coord}
    (\choicemap_{\reg}(z))_i \;=\; \phi(\kkt(z)+z_i),\qquad i\in[m],
\end{align}
for some scalar $\kkt(z)\in\R$, determined by the simplex constraint
\begin{align}\label{eq:appendix:kkt-eq}
    \sum_{i=1}^m \phi(\kkt(z)+z_i)\;=\;1.
\end{align}

\paragraph{Existence and uniqueness of $\kkt(z)$.}
Since $\theta$ is strongly convex, $\phi$ is strictly increasing and continuous, the function
\begin{align*}
    g(\kkt)\;\defeq\;\sum_{i=1}^m \phi(\kkt+z_i)
\end{align*}
is strictly increasing and continuous, with $\lim_{\kkt\to-\infty} g(\kkt)=0$ and $\lim_{\kkt\to+\infty} g(\kkt)=m$.
Therefore, by the intermediate value theorem, \eqref{eq:appendix:kkt-eq} admits at least one solution.
Uniqueness follows from strict monotonicity of $g$ at the solution.

\subsection{Examples: Euclidean, Negative Entropy, and Negative Tsallis Kernels}
\label{appendix:kernel-examples}
We collect three standard separable regularizers $\reg(y)=\sum_{i=1}^{\ynum}\theta(y_i)$ on $\ysp$ and their associated one-dimensional \eqref{eq:kkt_structure}.
For each example we list $\theta$, $\theta'$, the boundary slopes $\theta'(0^+)$, the resulting $\theta^*$ and $\phi$.

\subsubsection{Euclidean Regularizer.}
\label{appendix:euclidean-regularizer}
We take the quadratic kernel $\theta(y_i)\defeq \tfrac12 y_i^2$ for $y_i\in[0,1]$, and the associated regularizer $\reg(y)\defeq \tfrac12\norm{y}_2^2$.
Its derivative satisfies $\theta'(y_i)=y_i$ and $\theta'(0^+)=0$, so $\reg$ is non-steep at the boundary $y_i=0$.
The conjugate is obtained by maximizing a concave quadratic over $y_i\in[0,1]$:
\begin{align}\label{eq:appendix:euclid-conj}
    \theta^*(z_i)
    \;=\;
    \sup_{y_i\in[0,1]}\Big\{z_i y_i-\tfrac12 y_i^2\Big\}
    \;=\;
    \begin{cases}
        0, & z_i < 0,\\
        \tfrac12 z_i^2, & 0\leq z_i \leq 1,\\
        z_i-\tfrac12, & z_i > 1,
    \end{cases}
\end{align}
and its derivative recovers the kernel choice map
\begin{align}\label{eq:appendix:euclid-phi}
    \phi(z_i)
    \;=\;
    (\theta^*)'(z_i)
    \;=\;
    \Pi(z_i)
    \;\defeq\;
    \argmin_{y_i\in[0,1]} \norm{y_i-z_i}^2
    \;=\;
    \begin{cases}
        0, & z_i < 0,\\
        z_i, & 0\leq z_i \leq 1,\\
        1, & z_i > 1.
    \end{cases}
\end{align}

\subsubsection{Negative Entropy Regularizer.}
\label{appendix:entropy-regularizer}
We take the entropy kernel $\theta(y_i)\defeq y_i\log y_i$ for $y_i\in(0,1]$ with $\theta(0)\defeq 0$, and the associated regularizer $\reg(y)\defeq \sum_{i=1}^{\ynum} y_i\log y_i$.
Its derivative satisfies $\theta'(y_i)=1+\log y_i$ and $\theta'(0^+)=-\infty$, so $\reg$ is steep at the boundary $y_i=0$.
The conjugate is obtained by maximizing $z_i y_i-y_i\log y_i$ over $y_i\in[0,1]$:
\begin{align}\label{eq:appendix:entropy-conj}
    \theta^*(z_i)
    \;=\;
    \sup_{y_i\in[0,1]}\Big\{z_i y_i-y_i\log y_i\Big\}
    \;=\;
    \begin{cases}
        \exp(z_i-1), & z_i\leq 1,\\
        z_i, & z_i > 1,
    \end{cases}
\end{align}
and its derivative recovers the kernel choice map
\begin{align}\label{eq:appendix:entropy-phi}
    \phi(z_i)
    \;=\;
    (\theta^*)'(z_i)
    \;=\;
    \begin{cases}
        \exp(z_i-1), & z_i\leq 1,\\
        1, & z_i > 1.
    \end{cases}
\end{align}

\subsubsection{Negative Tsallis Regularizer ($q\in(0,1)$).}
\label{appendix:tsallis-regularizer}
We take the negative Tsallis kernel $\theta(y_i)\defeq \frac{y_i^q-y_i}{q-1}$ for $y_i\in[0,1]$, and the associated regularizer $\reg(y)\defeq \sum_{i=1}^{\ynum}\frac{y_i^q-y_i}{q-1}$.
Its derivative satisfies $\theta'(y_i)=\frac{q\,y_i^{q-1}-1}{q-1}$ and $\theta'(0^+)=-\infty$, so $\reg$ is steep at the boundary $y_i=0$.
The conjugate is obtained by maximizing a concave function over $y_i\in[0,1]$
\begin{align}\label{eq:appendix:tsallis-conj}
    \theta^*(z_i)
    \;=\;
    \sup_{y_i\in[0,1]}\Big\{z_i y_i-\theta(y_i)\Big\}
    \;=\;
    \begin{cases}
        \Bigl(\frac{1+(q-1)z_i}{q}\Bigr)^{\frac{q}{q-1}}, & z_i\leq 1,\\
        z_i, & z_i > 1,
    \end{cases}
\end{align}
and its derivative recovers the kernel choice map:
\begin{align}\label{eq:appendix:tsallis-phi}
    \phi(z_i)
    \;=\;
    (\theta^*)'(z_i)
    \;=\;
    \begin{cases}
        \Bigl(\frac{1+(q-1)z_i}{q}\Bigr)^{\frac{1}{q-1}}, & z_i\leq 1,\\
        1, & z_i > 1.
    \end{cases}
\end{align}

\section{Details Omitted from Section \ref{section: toolbox} (Toolbox)}

\subsection{Completed Proof of Lemma \ref{lem:reward-potential-cont} and Theorem \ref{thm:main-reduction}}
\label{appendix:reward-decomposition}
We combine these statements into a single theorem, restate it here, and prove it.
\begin{restatable}[Lemma~\ref{lem:reward-potential-cont} \& Theorem~\ref{thm:main-reduction}]{theorem}{rewardpotentialboth}
    \label{lem:reward-potential-both}
Fix a constant step size $\learn>0$ and let $\reg$ be a regularizer on $\ysp$. 
\begin{itemize}[leftmargin=*]
\item\textbf{Continuous time.}
Let the continuous historical reward $\HR(\time)$ be defined as in \eqref{eq:score_def}.
Then, for any continuous strategy $\xc(\time):[\tstart,\tend]\to\xsp$,
\begin{align}\label{eq:cont-reward-potential}
    \rcont{\xc(t)}{\HR(\tstart)}{\tend}
    = \frac{1}{\learn}\,\regconj\!\bigl(\learn\,\HR(\tstart)\bigr) - \frac{1}{\learn}\,\regconj\!\bigl(\learn\,\HR(\tend)\bigr).
\end{align}
\item\textbf{Discrete time.}
Let the discrete historical reward $\HR(\time)$ be defined as in \eqref{eq:score_def}.
Then, for any discrete strategy $\xd(\time)\from\setof{0,1,\dots,\tend-1}\to\xsp$,
\begin{align}\label{eq:disc-reward-potential}
    \rdisc{\xd(t)}{\HR(0)}{\tend}
    = \frac{1}{\learn}\,\regconj\!\bigl(\learn\,\HR(\tstart)\bigr) - \frac{1}{\learn}\,\regconj\!\bigl(\learn\,\HR(\tend)\bigr)
    + \frac{1}{\learn}\sum_{t=0}^{\tend-1}
    \breg_{\regconj}\parens*{\learn \HR(t+1),\, \learn \HR(t)}.
\end{align}
\end{itemize}
\end{restatable}

\begin{proof}
Recall that the FTRL choice at dual state \(z\) is
$y(z)=\choicemap_{\reg}(z)=\nabla\regconj(z).$
To analyze these reward functions, define the learner's potential
\begin{align}
    \Phi_{\learn}(Z) \defeq \frac{1}{\learn}\regconj(\learn Z),
\qquad Z\in\R^{\ynum}.
\end{align}
Then, by the chain rule,
$\nabla \Phi_{\learn}(Z)
=
\nabla \regconj(\learn Z)
=
\choicemap_{\reg}(\learn Z).$
Thus, changes in the potential \(\Phi_{\learn}\) track the reward accumulated along the trajectory.
\paragraph{Continuous time.}
Suppose \(\xc(t)\from [0, \tend] \to \xsp\) is a strategy for the continuous case. The continuous reward is defined as the integral of the instantaneous payoff
\begin{align*}
    \rcont{\xc(t)}{\HR(0)}{\tend}
    = \int_{0}^{\tend} \inner{\xc(t), \mat \cdot \choicemap_{\reg}(\learn \HR(t))} \dd t.
\end{align*}
Since \(\dot{\HR}(t) = \matalt^\top \xc(t)\) and \(\mat = -\matalt\), we rewrite
\begin{align*}
    \inner{\xc(t), \mat \cdot \choicemap_{\reg}(\learn \HR(t))}
    = \inner{\mat^\top \xc(t), \choicemap_{\reg}(\learn \HR(t))}
    = \inner{-\dot{\HR}(t), \choicemap_{\reg}(\learn \HR(t))}.
\end{align*}
Using \(\choicemap_{\reg}(\learn \HR(t))=\nabla \Phi_{\learn}(\HR(t))\), we obtain
\begin{align*}
    \inner{-\dot{\HR}(t), \nabla \Phi_{\learn}(\HR(t))}
    = - \frac{d}{dt} \Phi_{\learn}(\HR(t)).
\end{align*}
Therefore, by the Fundamental Theorem of Calculus,
\begin{align}
    \rcont{\xc(t)}{\HR(0)}{\tend}
    = -\int_{0}^{\tend} \frac{d}{dt} \Phi_{\learn}(\HR(t)) \dd t 
    & = \Phi_{\learn}(\HR(0)) - \Phi_{\learn}(\HR(\tend))  \\
    & = \frac{1}{\learn} \regconj(\learn \HR(0)) - \frac{1}{\learn} \regconj(\learn \HR(\tend)).
\end{align}
\paragraph{Discrete time.}
Suppose \(\xd(t)\from \setof{0, 1, \dots, T-1} \to \xsp\) is a strategy for the discrete case. The reward at time step \(t\) is
\[
\inner{\xd(t), \mat \cdot \choicemap_{\reg}(\learn \HR(t))}.
\]
Since \(\HR(t+1) = \HR(t) + \matalt^\top \xd(t)\) and \(\mat = -\matalt\), we have
\[
\mat^\top \xd(t) = \HR(t) - \HR(t+1).
\]
Substituting this into the reward expression yields
\begin{align*}
    \rdisc{\xd(t)}{\HR(t)}{1}
    = \inner{\HR(t) - \HR(t+1), \choicemap_{\reg}(\learn \HR(t))}
    = \inner{\HR(t) - \HR(t+1), \nabla \Phi_{\learn}(\HR(t))}.
\end{align*}
Applying the standard identity for the Bregman divergence from Definition~\ref{def:bregman-divergence} gives
\begin{align*}
    \inner{\HR(t) - \HR(t+1), \nabla \Phi_{\learn}(\HR(t))}
    = \Phi_{\learn}(\HR(t)) - \Phi_{\learn}(\HR(t+1))
    + \breg_{\Phi_{\learn}}(\HR(t+1), \HR(t)).
\end{align*}
Substituting this back into the reward equation gives
\begin{align*}
    \rdisc{\xd(t)}{\HR(t)}{1}
    = \Phi_{\learn}(\HR(t)) - \Phi_{\learn}(\HR(t+1))
    + \breg_{\Phi_{\learn}}(\HR(t+1), \HR(t)).
\end{align*}
Summing from \(t=0\) to \(T-1\), the first two terms form a telescoping sum:
\begin{align*}
    \sum_{t=0}^{T-1} \rdisc{\xd(t)}{\HR(t)}{1}
    &= \Phi_{\learn}(\HR(0)) - \Phi_{\learn}(\HR(T))
    + \sum_{t=0}^{T-1} \breg_{\Phi_{\learn}}(\HR(t+1), \HR(t)) \\
    &= \frac{1}{\learn} \regconj(\learn \HR(0)) - \frac{1}{\learn} \regconj(\learn \HR(T))
    + \sum_{t=0}^{T-1} \frac{1}{\learn} \breg_{\regconj}(\learn \HR(t+1), \learn \HR(t)),
\end{align*}
because $\breg_{\Phi_{\learn}}(Z', Z) = \frac{1}{\learn} \breg_{\regconj}(\learn Z', \learn Z)$.
\end{proof}

\subsubsection{Completed Proof of Smoothness and Corollary \ref{cor:fw_complexity}}
\label{appendix:claim:smooth-of-composition}
\begin{restatable}
{claim}{smoothofcomposition} \label{claim:smooth-of-composition}
The function $G$ is $\beta$-smooth with respect to $\norm{\cdot}$, with
\[
\beta \;=\; \frac{(\learn\,\tend\,\opnorm{\B^\top})^{2}}{\scparam}, \qquad \opnorm{\B^\top}\defeq \sup_{\x\neq 0}\frac{\dnorm{\B^\top \x}}{\norm{\x}}.
\]
\end{restatable}
\begin{proof}
Define the affine map
\[
f(\x)\defeq \learn\big(\HR(\tstart)+\tend \B^\top \x\big),
\]
so that $G(\x)=\regconj(f(\x))$ and $\nabla f(\x)=\learn\,\tend\,\B^\top$.
By the chain rule,
\[
\nabla G(\x) = (\learn\,\tend\,\B^\top)^\top \nabla \regconj\!\big(f(\x)\big).
\]
Since $\reg$ is $\scparam$-strongly convex with respect to $\norm{\cdot}$, its conjugate
$\regconj$ is $(1/\scparam)$-smooth with respect to the dual norm $\dnorm{\cdot}$.
Thus, for any $\x,\x'\in\xsp$,
\begin{align*}
\dnorm*{\nabla G(\x)-\nabla G(\x')}
&=
\dnorm*{(\learn\,\tend\,\B^\top)^\top\!\Big(\nabla \regconj(f(\x))-\nabla \regconj(f(\x'))\Big)} \\
&\le
\learn\,\tend\,\opnorm{\B^\top}\,
\dnorm*{\nabla \regconj(f(\x))-\nabla \regconj(f(\x'))} \\
&\le
\learn\,\tend\,\opnorm{\B^\top}\,\frac{1}{\scparam}\,\dnorm*{f(\x)-f(\x')} \\
&=
\learn\,\tend\,\opnorm{\B^\top}\,\frac{1}{\scparam}\,\dnorm*{\learn\,\tend\,\B^\top(\x-\x')} \\
&\le
\frac{(\learn\,\tend\,\opnorm{\B^\top})^{2}}{\scparam}\,\norm{\x-\x'}.
\end{align*}
\end{proof}
\fwcomplexity*
\begin{proof}
    We apply the Frank-Wolfe algorithm to minimizing the objective function $G(\x)$.
    According to \citet[Theorem~3.8]{bubeck2015convexoptimizationalgorithmscomplexity}, applying Frank--Wolfe to a $\beta$-smooth convex function over a domain with diameter $R$ yields iterates $\{\x_k\}_{k\ge 0}$ satisfying the convergence rate
    \[
        G(\x_k)-G(\xopt)
        \;\le\;
        \frac{2R^2\,\beta}{k+1}
        \;=\;
        \frac{2R^2\,(\learn\,\tend)^2\,\opnorm{\mat}^2}{\scparam\,(k+1)},
    \]
    where $\xopt \in \argmin_{\x\in\xsp} G(\x)$.
    
    Recall from Lemma~\ref{lem:reward-potential-cont} (and Remark~\ref{rem:path_independence}) that the optimizer's continuous reward is given by
    \[
        \rcont{\x}{\HR(\tstart)}{\tend} = \frac{1}{\learn}\regconj(\learn\HR(\tstart)) - \frac{1}{\learn}G(\x).
    \]
    Therefore, the suboptimality gap in terms of reward is directly proportional to the suboptimality gap of the function $G$, scaled by $1/\learn$
    \[
        \optrcont{\HR(\tstart)}{\tend} - \rcont{\x_k}{\HR(\tstart)}{\tend}
        \;=\;
        \frac{G(\x_k) - G(\xopt)}{\learn}.
    \]
    To achieve a target accuracy of $\epsilon$ in the reward, we require
    \[
        \frac{G(\x_k) - G(\xopt)}{\learn} \le \epsilon \implies G(\x_k) - G(\xopt) \le \epsilon \learn.
    \]
    Substituting the convergence rate of Frank-Wolfe
    \[
        \frac{2R^2\,(\learn\,\tend)^2\,\opnorm{\mat}^2}{\scparam\,(k+1)} \le \epsilon \learn.
    \]
    Solving for the number of iterations $k$
    \[
        k+1 \;\ge\; \frac{2R^2\,\learn^2\,\tend^2\,\opnorm{\mat}^2}{\scparam\,\epsilon\,\learn}
        \;=\; \frac{2R^2\,\learn\,\tend^2\,\opnorm{\mat}^2}{\scparam\,\epsilon}.
    \]
    Thus, the number of iterations required is $k = \bigoh\!\left(\frac{R^{2}\,\learn\,\tend^2\,\opnorm{\mat}^{2}}{\scparam\,\epsilon}\right)$.
    Since each Frank-Wolfe iteration involves a Linear Minimization Oracle over the simplex (or strategy space), which takes $\bigoh(nm)$ time, the total runtime is
    \[
        \text{Total Time} = k \times \bigoh(nm) = \bigoh\!\left(nm\,\frac{R^{2}\,\learn\,\tend^2\,\opnorm{\mat}^{2}}{\scparam\,\epsilon}\right).
    \]
\end{proof}

\subsection{Optimizer Reward Sandwich}
\label{appendix:prop:optimal-continuous-reward-range}
Below we present a self-contained proof of the sandwich property for the reward of the optimizer. This recovers \eqref{eq:reward sandwich} from the Introduction.
Proposition~\ref{prop:optimal-continuous-reward-range} is the continuous-time analogue of the standard FTRL “regret $\le$ regularizer diameter$/\eta$” bound.
\begin{restatable}{proposition}{optimalcontinuousrewardrange}
\label{prop:optimal-continuous-reward-range}
    Let $\reg$ be a regularizer on $\ysp$. Fix a constant step size $\learn>0$ and horizon $\tend \in \N$. The optimizer's optimal continuous-time reward is bounded as
    \begin{align}\label{eq:cont-time-one-range}
    \tend \cdot \val
    \;\le\;
    \optrcont{\zeroreward}{\tend}
    \;\le\;
    \tend \cdot \val
    + \frac{\parens*{\regmax-\regmin}}{\learn}.
    \end{align}
\end{restatable}
\begin{proof}
    The optimal continuous-time reward is trivially lower bounded by $\tend \cdot \val$ since the optimizer can always play a max-min strategy throughout $[0, \tend]$.

For the upper bound, assume $\HR(0)=\zeroreward$. Then
\[
\regconj(0)
=\sup_{\y\in\ysp}\{-\reg(\y)\}
=-\inf_{\y\in\ysp}\reg(\y)
=-\regmin.
\]
Starting from the definition of $ \optrcont{\HR(\tstart)}{\tend}$:
\[        \optrcont{\HR(\tstart)}{\tend}
        \;=\; \max_{\x \in \xsp} \rcont{\x}{\HR(\tstart)}{\tend} =
        \Phi_\learn\bigl(\HR(\tstart)\bigr)
        - \min_{\x \in \xsp} \Phi_\learn\bigl(\HR(\tstart) + \tend \matalt^\top \x\bigr).
\]

\begin{align*}
\optrcont{\HR(0)}{\tend}
&\le
-\frac{\regmin}{\learn}
-\min_{\x\in\xsp}\frac{1}{\learn}\,
\regconj\!\parens*{\learn\parens*{\HR(0)+\tend \B^\top \x}}.
\end{align*}
By definition of $\regconj$ and $\A = -\B$,
\begin{align*}
-\min_{\x\in\xsp}\frac{1}{\learn}\regconj\!\parens*{\learn\parens*{\HR(0)+\tend \B^\top \x}}
&=
-\frac{1}{\learn}\min_{\x\in\xsp}\sup_{\y\in\ysp}
\braces*{\learn\,\inner{\HR(0)+\tend \B^\top \x,\y}-\reg(\y)}\\
&\le
-\frac{1}{\learn}\sup_{\y\in\ysp}\min_{\x\in\xsp}
\braces*{\learn\,\inner{\HR(0)+\tend \B^\top \x,\y}-\reg(\y)}\\
&=
\min_{\y\in\ysp}\braces*{-\inner{\HR(0),\y}
+\tend\max_{\x\in\xsp}\inner{\A^\top \x,\y}
+\frac{1}{\learn}\reg(\y)}\\
&\le
-\min_{\y\in\ysp}\inner{\HR(0),\y}
+\tend\cdot\val
+\frac{\regmax}{\learn}.
\end{align*}
Plugging back yields
\[
\optrcont{\HR(0)}{\tend}
\le
-\frac{\regmin}{\learn}
-\min_{\y\in\ysp}\inner{\HR(0),\y}
+\tend\cdot\val
+\frac{\regmax}{\learn},
\]
and $\HR(0)=\zeroreward$ concludes the proof.
\end{proof}

  \subsection{Details Omitted from Section \ref{section:exploitation} }
  \label{appendix:app_exploitation}
\subsubsection{Learner’s Viewpoint}
\label{appendix:learner-viewpoint}
We begin with a basic fact about FTRL when the optimizer plays a fixed constant strategy. For simplicity, assume the initial reward is $0$.
\begin{restatable}{observation}{obsftrfixedx} \it
If the optimizer plays a constant strategy $\xfix$, the learner's FTRL responses can be written, for all $t > 0$, as
\begin{align} \label{eqn:y-dynamics-when-x-fixed}
    \y(t) \;:=\; \choicemap_{\reg}\!\bigl(\learn\, t\, \B^\top \xfix\bigr) = \argmax_{\y\in\ysp} \Big\{\langle \learn\, t\, \B^\top \xfix, \y\rangle - \reg(\y)\Big\} = \argmin_{\y\in\ysp} \Big\{\langle \xfix, \A \y\rangle + \frac{1}{\learn\, t}\reg(\y)\Big\}.
\end{align}
From the last expression, we see that the regularization term vanishes as $t\to\infty$. Therefore, we expect to see the learner's FTRL responses converge to a best response to $\xfix$.
\end{restatable}
This observation will be a major tool of our analysis. We state this formally below.

\begin{restatable}[Lemma \ref{lemmain:ftrl-fixed-x}]
{lemma}{ftrlfixedx}
    \label{lem:ftrl-fixed-x}
Fix a constant strategy $\xfix\in\xsp$. Then the learner's FTRL response $\y(t)$ converges to $\ybrep$ as $t\to\infty$, where the limit is the unique point
\begin{align*}
    \ybrep
\;=\;
\argmin_{y \in \conv\{e_i:i\in\brep(\xfix)\}}
\reg(y).
\end{align*}
\end{restatable}

\begin{proof}
We start from Equation \eqref{eqn:y-dynamics-when-x-fixed}. For $t>0$,
\[
\y(t)=\argmin_{\y\in\ysp}\Big\{\inner{\xfix,\A\y}+\frac{1}{\learn t}\reg(\y)\Big\}.
\]
Let $\lambda_t = \frac{1}{\learn t}$ and define
\[
\psi_t(\y):=\inner{\xfix,\A\y}+\lambda_t \reg(\y),
\qquad
\psi(\y):=\inner{\xfix,\A\y},
\qquad \y\in\ysp.
\]
Let
\[
F=\argmin_{\y\in\ysp}\psi(\y)=\argmin_{\y\in\ysp}\inner{\xfix,\A\y}= \argmin_{y \in \conv \{\,e_i : i \in \brep(\xfix)\,\}} \inner{\xfix,\A\y}
\]
since $\psi$ is linear over the simplex.
Define $\ybrep:=\argmin_{\y\in F}\reg(\y)$, which is unique because $\reg$ is strictly convex and $F$ is convex.

Fix any sequence $t_k\to\infty$. By compactness of $\ysp$, the sequence $\{\y(t_k)\}$ has a convergent subsequence. Assume $\y(t_k)\to \y_\infty$. Because $\ysp$ is compact and $\reg$ is continuous, we have
\begin{align*}
    \sup_{\y\in\ysp}\,|\psi_{t_k}(\y)-\psi(\y)|
=\sup_{\y\in\ysp}\,\lambda_{t_k}|\reg(\y)|
\rightarrow 0,
\end{align*}
so $\psi_{t_k}\to \psi$ uniformly on $\ysp$, hence $\psi_{t_k}$ is epi-convergence to $\psi$.
Since $\y(t_k)\in\argmin_{\y\in\ysp}\psi_{t_k}(\y)$ for each $k$,
\citet[Theorem~7.33]{rockafellar1998variational} implies that every cluster point of $\{\y(t_k)\}$ belongs to
$\argmin_{\y\in\ysp}\psi(\y)=F$. In particular, $\y_\infty\in F$.

Since $\y(t_k)$ minimizes $\psi_{t_k}$, taking $\y=\ybrep$ gives
\[
\psi_{t_k}\!\bigl(\y(t_k)\bigr)\le \psi_{t_k}(\ybrep)
\quad\Rightarrow\quad
\psi\!\bigl(\y(t_k)\bigr)+\lambda_{t_k}\reg\!\bigl(\y(t_k)\bigr)
\le \psi(\ybrep)+\lambda_{t_k}\reg(\ybrep).
\]
Because $\ybrep\in F$, we have $\psi(\ybrep)\le \psi(\y(t_k))$, hence
\[
0\le \psi\!\bigl(\y(t_k)\bigr)-\psi(\ybrep)
\le \lambda_{t_k}\bigl(\reg(\ybrep)-\reg\!\bigl(\y(t_k)\bigr)\bigr),
\]
which implies $\reg(\y(t_k))\le \reg(\ybrep)$ for all $k$.
By continuity of $\reg$,
\[
\reg(\y_\infty)=\lim_{k\to\infty}\reg\!\bigl(\y(t_k)\bigr)\le \reg(\ybrep).
\]
But $\y_\infty\in F$ and $\ybrep$ minimizes $\reg$ over $F$, so $\reg(\ybrep)\le \reg(\y_\infty)$, and thus
$\reg(\y_\infty)=\reg(\ybrep)$. By uniqueness of the minimizer of $\reg$ on $F$, we conclude $\y_\infty=\ybrep$.

We have shown that every sequence $t_k\to\infty$ admits a subsequence along which $\y(t_k)\to \ybrep$; hence $\y(t)\to \ybrep$ as $t\to\infty$.
\end{proof}
The following remark completes the learner-viewpoint analysis for Lemma~\ref{lemmain:ftrl-fixed-x}.
\begin{remark}
    If the optimizer plays a max-min strategy $\xne$, then by Lemma~\ref{lem:ftrl-fixed-x}, the learner converges to $\yopt \in \argmin_{\y\in \conv\{e_i:i\in\brep(\xne)\}} \reg(\y),$ and moreover $\inner{\xne,  \A \yopt}$ equals the value of the game $\val$. Since we do not necessarily have that $\xne$ is a best response to $\yopt$, $\yopt$ is not guaranteed to be a min-max strategy for the learner.
\end{remark}

\subsubsection{Optimizer's Viewpoint}
\label{appendix:optimizer-viewpoint}
The optimizer plays a fixed constant strategy $\xfix$.
By Lemma~\ref{lem:ftrl-fixed-x}, the learner's response $\y(t)$ converges to the unique regularizer-selected best response $\ybrep$ with $\payoff(\xfix, \ybrep) = \valfix$. The instantaneous surplus above the best-response value $\valfix$ at time $t$ is $\payoff(\xfix,\y(t))-\valfix$. We restate and prove the gap decomposition below.
\surplusgapdecomp*
\begin{proof}
    By linearity, $\payoff(\xfix,\y(t)) = \inner{\vv, \y(t)} = \sum_{i=1}^{\ynum} (\valfix + \gap_i)\y_i(t) = \valfix \sum \y_i(t) + \sum \gap_i \y_i(t)$. Since $\y(t) \in \ysp$ and $\gap_i=0$ for $i \in \optset$, the result follows.
\end{proof}

\subsubsection{Completed Proof of Lemma \ref{lem:bound-for-cont-ftrl}}
\label{appendix:lem:bound-for-cont-ftrl}
\lemboundforcontftrl*
\begin{proof}
From Equation~\eqref{eqn:y-dynamics-when-x-fixed}, $y(t) = \choicemap_{\reg}\!\bigl(\learn\, t\, \B^\top \xfix\bigr) = \choicemap_{\reg}\!\bigl(-\learn\, t\, \vv \bigr)$. The KKT condition in \eqref{eq:kkt_structure} gives $y_i(t)=\phi(\kkt(t)-\learn t\, \vv_i)$ with $\sum_{i=1}^{\ynum} y_i(t)=1$ 
for some scalar $\kkt(t)$.

(i) Since $v_i=\valfix$ for all $i\in \optset$, then $y_i(t) = \phi(\kkt(t)-\learn t\, \valfix)$ for all $i\in \optset$, which is independent of $i$. Define $a_t := \phi(\kkt(t)-\learn t\, \valfix)\ge 0$. Then, $y_i(t) = a_t$ for all $i\in \optset$. 

Then we prove the bounds on $a_t$. Write the normalization as
\begin{align*}
    1=\sum_{i\in \optset} a_t+\sum_{j\notin \optset} y_j(t)=\optnum a_t+\sum_{j\notin \optset} y_j(t)\ \ge\ \optnum a_t,
\end{align*}
so $a_t\le 1/\optnum$. Also, for $j\notin \optset$, by monotonicity of $\phi$ we have
\begin{align*}
    y_j(t)=\phi(\kkt(t)-\learn t\, \vv_j) = \phi(\kkt(t)-\learn t(\valfix+\gap_j)) 
    \le \phi(\kkt(t)-\learn t\, \valfix)=a_t.
\end{align*}
Therefore, $1 = \sum_{i=1}^{\ynum} y_i(t) \le \ynum a_t$, so $a_t\ge 1/\ynum$.

(ii) For $j\notin \optset$,
\[
y_j(t)=\phi\big(\kkt(t)-\learn t(\valfix+\gap_j)\big)
=\phi\big(\theta'(a_t)-\learn t\,\gap_j\big).
\]
Since $a_t\in[1/\ynum,1/\optnum]$ and $\theta'$ is increasing, $\theta'(a_t)\in[\theta'(1/\ynum),\theta'(1/\optnum)]$.
Applying the increasing map $\phi$ yields the claimed bounds.
\end{proof}

\section{Details Omitted from Section \ref{subsec:inverse-rate} (The Inverse-Rate Law)}
\label{app subsec:inverse-rate}
We are now in a position to state our first main result. Assuming a fixed optimizer strategy, we establish that {surplus} is governed by the interplay between:
$a)$ the \emph{strategic landscape} (specifically, the mass of learner's actions that are \emph{not} best responses to optimizer) and $b)$ the \emph{geometry} of the regularizer's conjugate $\theta^*$.
The following meta-theorem encapsulates the essence of our findings: \begin{center}\it Optimizer's {surplus} is fundamentally the ratio of strategic suboptimality to the learning speed. \end{center}
\subsection{The Cost of Learning}
\label{subsec:cost-of-learning}
\noindent Our result establishes a fundamental sweeping condition: surplus against a fixed strategy arises \emph{solely} from the presence of suboptimal actions (non-best responses). If the learner is already fully aligned with the optimizer's strategy, the cumulative lag vanishes. Conversely, any misalignment incurs an inescapable surplus cost scaling with the inverse step size.
\begin{tcolorbox}[
  enhanced,
  breakable,
  colback=white,
  colframe=black!75, 
  boxrule=0.8pt,
  arc=3mm,
  left=2mm,right=2mm,top=2mm,bottom=2mm,
  drop fuzzy shadow 
]
\textbf{Meta-Theorem.}
Consider a learner with step size $\eta$ playing against a fixed optimizer strategy. Let $N_{sub}$ be the number of suboptimal actions (non-best responses).
\begin{itemize}[leftmargin=*,itemsep=0pt,topsep=-2pt]
    \item \textbf{\small Perfect Alignment:} If all learner actions are best responses ($N_{sub} = 0$), there is \textbf{no surplus}. 
  \[
  \lag = 0
  \]
 \item \textbf{Cost of Misalignment:} If there exist suboptimal actions ($N_{sub} > 0$), the cumulative surplus is proportional to the number of bad choices divided by the step size, corrected by the boundary geometry:
	\[
       \lag \approx \Theta\left( \frac{N_{sub}}{\eta(T)} \right) - o_{\eta T,\reg}(1) \cdot \mathbb{I}(\text{steep}).
    \]
    \footnotesize{\textit{(Note: While non-steep regularizers saturate quickly, steep regularizers retain a slowly decaying residual potential, strictly reducing the realized surplus at finite $T$.)}}
\end{itemize}
\end{tcolorbox}

\subsection{Surplus in Continuous-Time FTRL Dynamics}
\label{section:exploitation bounds}
\paragraph{Regularizer Geometry via a Dual Potential}
To turn the preceding meta-theorem into precise bounds, we need a scalar measure of how ``costly'' it is for the learner to maintain probability mass on a coordinate. We capture this effect through a \emph{dual potential} associated with the regularizer.
\begin{definition}
\label{def:dual-potential} \it
Let $\theta:[0,1]\to\mathbb{R}$ be the kernel of a separable regularizer and $\theta^\ast$ its convex conjugate. We define the \emph{dual potential} $V:[0,1]\to\mathbb{R}$ by $V(u)\;\coloneqq\;\theta^\ast(\theta'(u))\;=\;u\,\theta'(u)-\theta(u)$.
\end{definition}
\begin{remark}
    For standard kernels, $V$ admits simple closed forms:
\emph{Euclidean} regularizer yields $V(u)=\tfrac12 u^2$ (up to constants), \emph{Tsallis} regularizer with parameter $q\neq 1$ yields $V(u)=u^q$, and the \emph{entropic} case arises in the limit $q\to 1$, giving $V(u)=u$ (again up to irrelevant constants).
\end{remark}

We now provide the formal statement of Theorem \ref{thm:exploit-finite}. While the main text provided a simplified version for intuition, the statement below specifies the exact bounds in terms of the potential energy change $\Delta V$ and the discretization factor $N_{\mathrm{sub}}$.
\begin{theorem}[Formal Version of Theorem~\ref{thm:exploit-finite}]
\label{app thm:exploit-finite}
Fix a constant optimizer strategy $\xfix$. Let $\gapmin, \gapmax$ be the minimum and maximum payoff gaps of the suboptimal actions. For any horizon $T \ge 0$, $\lag$ is bounded by
\begin{align}
    \frac{\ynum-\optnum}{\learn} \cdot\underline{\Delta V}(T) \;\le\; \lag \;\le\; \frac{\ynum-\optnum}{\learn} \cdot \overline{\Delta V}(T), \tag{Surplus Bound}
\end{align}
where $\underline{\Delta V}(T),\overline{\Delta V}(T)$ represent the potential drops of form $\Delta V(T) \sim V_{start} - V_{end}(T):$
\begin{align}
        \underline{\Delta V}(T) &:= \frac{\gapmin}{\gapmax} \left[ V(1/\ynum) - \theta^*\Big( \max\{ \theta'(1/\ynum) - \learn \gapmax T, \; \theta'(0^+) \} \Big) \right].\label{eq:c_h_low}\\
     \overline{\Delta V}(T) &:= \frac{\gapmax}{\gapmin} \left[ V(1/\optnum) - \theta^*\Big( \max\{ \theta'(1/\optnum) - \learn \gapmin T, \; \theta'(0^+) \} \Big) \right].\label{eq:c_h_high}
\end{align}
\end{theorem}
\begin{proof}
    Start from \eqref{eqn:surplus_decomposition}
\begin{align*}
    \expl(T)
=
\vag
\;+\;
\int_0^T \sum_{i\notin\optset} \gap_i\,\y_i(t)\,dt.
\end{align*}
For $i\notin \optset$, $\gapmin \le \gap_i \le \gapmax$, hence
\[
\gapmin \sum_{i\notin \optset}\y_i(t)\ \le\ \sum_{i\notin \optset}\gap_i\,\y_i(t)\ \le\ \gapmax \sum_{i\notin \optset}\y_i(t).
\]
From Lemma \ref{lem:bound-for-cont-ftrl} with monotonicity of $\phi$,
\[
 \phi\big(\theta'(1/\ynum)-\learn t\,\gapmax\big) \le \y_i(t) \le \phi\big(\theta'(1/\optnum)-\learn t\,\gapmin\big).
\]
Therefore
\begin{align} 
    \sum_{i\notin \optset}\gap_i\,\y_i(t)\ \ge\ \gapmin(\ynum-\optnum)\,\phi\big(\theta'(1/\ynum)-\learn t\,\gapmax\big) \eqdef \underline{g}(t), \label{eqn:underline-g}\\ 
    \sum_{i\notin \optset}\gap_i\,\y_i(t)\ \le\ \gapmax(\ynum-\optnum)\,\phi\big(\theta'(1/\optnum)-\learn t\,\gapmin\big) \eqdef \overline{g}(t). \label{eqn:overline-g}
\end{align}

For the lower bound, let $u=\theta'(1/\ynum)-\learn\gapmax t$. Then $dt=-du/(\learn\gapmax)$ and
\[
\int_0^T \phi\big(\theta'(1/\ynum)-\learn\gapmax t\big)\,dt
=
\frac{1}{\learn\gapmax}
\int_{\theta'(1/\ynum)-\learn\gapmax T}^{\theta'(1/\ynum)} \phi(u)\,du.
\]
From definition of $\phi$, we have that $\phi(u)=0$ for $u\le \theta'(0^+)$, this equals
\begin{align} \label{eqn:integral-phi-lower}
    \int_0^T \underline {g}(t)\,dt
    &= \frac{\ynum-\optnum}{\learn}\frac{\gapmin}{\gapmax}
\int_{\max\{\theta'(1/\ynum)-\learn\gapmax T,\ \theta'(0^+)\}}^{\theta'(1/\ynum)} \phi(u)\,du \notag \\
    &=\frac{\ynum-\optnum}{\learn}\frac{\gapmin}{\gapmax}
\Big(\theta^*(\theta'(1/\ynum))-\theta^*(\max\{\theta'(1/\ynum)-\learn\gapmax T,\ \theta'(0^+)\})\Big),
\end{align}
which yields the stated $\Chl$.

For the upper bound, do the same with $u=\theta'(1/\optnum)-\learn\gapmin t$ to obtain $\Chu$ by integrating $\int_0^T \overline {g}(t)\,dt$ such that
\begin{align} \label{eqn:integral-phi-upper}
    \int_0^T \overline {g}(t)\,dt
    &= \frac{\ynum-\optnum}{\learn}\frac{\gapmax}{\gapmin}
\int_{\max\{\theta'(1/\optnum)-\learn\gapmin T,\ \theta'(0^+)\}}^{\theta'(1/\optnum)} \phi(u)\,du  \notag \\
    &=\frac{\ynum-\optnum}{\learn}\frac{\gapmax}{\gapmin}
\Big(\theta^*(\theta'(1/\optnum))-\theta^*(\max\{\theta'(1/\optnum)-\learn\gapmin T,\ \theta'(0^+)\})\Big),
\end{align}
which yields the upper bound.
\end{proof}

\paragraph{Regime Analysis: Steep vs.\ Non-Steep}
The asymptotic behavior of surplus is dictated by the regularizer's boundary behavior— leading to two distinct regimes:
 if $|\theta'(0^+)|<\infty$ (non-steep), suboptimal actions are eliminated in finite time and surplus saturates and if $|\theta'(0^+)|=\infty$ (steep), the learner approaches the boundary only asymptotically, leaving a persistent tail.
We provide the formal statement of Theorem \ref{thm:exploit-infinite} below, which was described informally in the main text.
\begin{theorem}[Formal Version of Theorem \ref{thm:exploit-infinite}]
\label{app thm:exploit-infinite}
Fix a separable regularizer with kernel $\theta$, step size $\learn>0$, and a fixed optimizer strategy $\xfix$ inducing gaps $\{\gap_i\}$, with $\gapmin\le \gap_i\le \gapmax$ for all $i\notin\optset$.
Let $V_{\mathrm{bdry}}\;\coloneqq\;\theta^\ast(\theta'(0^+))\in[0,+\infty]$ be the boundary energy.
Then, for large horizons $\learn T\to\infty$, $\lag$ satisfies the following dichotomy.

\begin{enumerate}[label=\textbf{(\Alph*)},leftmargin=*,itemsep=-3pt]
\item \textbf{Non-steep regularizers (finite boundary slope).}
If $\theta'(0^+)>-\infty$ (equivalently, $V_{\mathrm{bdry}}<\infty$), then suboptimal actions are eliminated in finite time $T^\ast=t_{\theta}^*<\infty$ and surplus \emph{saturates}. In particular, for all $T\ge T^\ast$,
            \[\lag \in  \frac{\ynum-\optnum}{\learn} \left[ \frac{\gapmin}{\gapmax}\Big( V(1/\ynum)-V_{\mathrm{bdry}}\Big), \; \frac{\gapmax}{\gapmin}\Big( V(1/\optnum)-V_{\mathrm{bdry}}\Big)\right].\]
\item \textbf{Steep regularizers (infinite boundary slope).}
If $\theta'(0^+)=-\infty$ (so $V_{\mathrm{bdry}}=0$ under the normalization $\theta^\ast(-\infty)=0$), then the learner approaches the boundary only asymptotically and surplus has an \emph{infinite tail}. As $\learn T\to\infty$,
            \[\lag \in  \frac{\ynum-\optnum}{\learn} \left[ \frac{\gapmin}{\gapmax}V\bigl(1/\ynum\bigr), \; \frac{\gapmax}{\gapmin}V\bigl(1/\optnum\bigr) \right] \pm o(1).\]
\end{enumerate}
\end{theorem}
\begin{proof}
In part (A), assume $\theta'(0^+)>-\infty$. 
For $\underline{\Delta V}(T)$ in \eqref{eq:c_h_low}, the inner maximum equals $\theta'(0^+)$ whenever $\theta'(1/\ynum)-\learn\,\gapmax\,T \;\le\; \theta'(0^+),$
i.e., for all
\begin{align*}
T \;\ge\; T_1 \;:=\; \frac{\theta'(1/\ynum)-\theta'(0^+)}{\learn\,\gapmax}.
\end{align*}
Thus, for all $T\ge T_1$,
\begin{align*}
\theta^*\!\Bigl(\max\{\theta'(1/\ynum)-\learn\gapmax T,\ \theta'(0^+)\}\Bigr)
\;=\;
\theta^*\!\bigl(\theta'(0^+)\bigr).
\end{align*}
Similarly, for $\overline{\Delta V}(T)$ in \eqref{eq:c_h_high}, the maximum equals $\theta'(0^+)$ once
\begin{align*}
T \;\ge\; T_2 \;:=\; \frac{\theta'(1/\optnum)-\theta'(0^+)}{\learn\,\gapmin},
\end{align*}
so for all $T\ge T_2$,
\begin{align*}
\theta^*\!\Bigl(\max\{\theta'(1/\optnum)-\learn\gapmin T,\ \theta'(0^+)\}\Bigr)
\;=\;
\theta^*\!\bigl(\theta'(0^+)\bigr).
\end{align*}
Since $T_2 \ge T_1$, for all $T\ge T_2$ both maxima equal $\theta'(0^+)$, yielding the bounds in part (A).

In part (B), assume $\theta'(0^+)=-\infty$ and $\theta^*(-\infty)=0$.
By Theorem~\ref{app thm:exploit-finite},
\begin{align*}
\Chl
&=\frac{\gapmin}{\gapmax}\Big(
\theta^*\big(\theta'(1/\ynum)\big)
-\theta^*\!\big(\theta'(1/\ynum)-\learn\gapmax\tend\big)
\Big), \\
\Chu
&=\frac{\gapmax}{\gapmin}\Big(
\theta^*\big(\theta'(1/\optnum)\big)
-\theta^*\!\big(\theta'(1/\optnum)-\learn\gapmin\tend\big)
\Big).
\end{align*}
As $\learn T\to\infty$, we have $\theta'(1/\ynum)-\learn\gapmax T$ and
$\theta'(1/\optnum)-\learn\gapmin T$ go to $-\infty$. Hence,
\begin{align*}
\theta^*\!\bigl(\theta'(1/\ynum)-\learn\gapmax T\bigr)=o_{\learn T}(1),
\qquad
\theta^*\!\bigl(\theta'(1/\optnum)-\learn\gapmin T\bigr)=o_{\learn T}(1),
\end{align*}
which yields the bounds in part (B).
\end{proof}

\subsection{Examples of Potential Drops, $\vag$, $\lag$}
\label{appendix:app:examples-gaps}
Theorems~\ref{app thm:exploit-finite} and \ref{app thm:exploit-infinite} yield the following takeaways.
Define the lower-bound potential endpoint
\begin{align*}
\underline{V_{\textnormal{end}}}(T)
\;\defeq\;
\theta^*\!\Bigg(
\max\Bigg\{
\theta'\!\Big(\frac{1}{\ynum}\Big)-\learn\,\gapmax\,T,\;
\theta'(0^+)
\Bigg\}
\Bigg).
\end{align*}
If $\theta'(0^+)=-\infty$ (steep), then the maximum is attained by the first term for all $T\ge 0$.
If $\theta'(0^+)>-\infty$ (non-steep), by Lemma \ref{lem:bound-for-cont-ftrl}, there exists $T^* = t^*_\theta<\infty$ such that the maximum is attained by the first term for $T\le T^*$ and by the second term for $T\ge T^*$.
The same steep/non-steep dichotomy applies to the corresponding upper-bound terms.
Using the kernel choice maps in Example~\ref{appendix:kernel-examples}, we now make the bounds in
Theorems~\ref{app thm:exploit-finite} and \ref{app thm:exploit-infinite} explicit.

\begin{enumerate}[label=(\alph*),leftmargin=*]
\item \emph{Negative entropy.}
By \eqref{eq:appendix:entropy-conj} and $\theta'(y)=1+\log y$,
\begin{align*}
\underline{V_{\textnormal{start}}}
&=
\theta^*\!\Big(\theta'\!\Big(\frac{1}{\ynum}\Big)\Big)
=
\frac{1}{\ynum},
&
\underline{V_{\textnormal{end}}}(T)
&=
\theta^*\!\Big(\theta'\!\Big(\frac{1}{\ynum}\Big)-\learn\gapmax T\Big)
=
\frac{1}{\ynum}e^{-\learn\gapmax T},\\
\overline{V_{\textnormal{start}}}
&=
\theta^*\!\Big(\theta'\!\Big(\frac{1}{\optnum}\Big)\Big)
=
\frac{1}{\optnum},
&
\overline{V_{\textnormal{end}}}(T)
&=
\theta^*\!\Big(\theta'\!\Big(\frac{1}{\optnum}\Big)-\learn\gapmin T\Big)
=
\frac{1}{\optnum}e^{-\learn\gapmin T}.
\end{align*}
Thus,
\begin{align*}
\underline{\Delta V}(T)
&=
\underline{V_{\textnormal{start}}}-\underline{V_{\textnormal{end}}}(T)
=
\frac{1}{\ynum}\bigl(1-e^{-\learn\gapmax T}\bigr),
&
\overline{\Delta V}(T)
&=
\overline{V_{\textnormal{start}}}-\overline{V_{\textnormal{end}}}(T)
=
\frac{1}{\optnum}\bigl(1-e^{-\learn\gapmin T}\bigr).
\end{align*}
By Theorem~\ref{app thm:exploit-finite}, bringing back to \eqref{eqn:surplus_decomposition}, for all $T\ge 0$,
\begin{align*}
\expl(T)
&\;\ge\;
\vag
\;+\;
\frac{\ynum-\optnum}{\learn}\cdot\frac{\gapmin}{\gapmax}\cdot
\frac{1}{\ynum}\bigl(1-e^{-\learn\gapmax T}\bigr),\\
\expl(T)
&\;\le\;
\vag
\;+\;
\frac{\ynum-\optnum}{\learn}\cdot\frac{\gapmax}{\gapmin}\cdot
\frac{1}{\optnum}\bigl(1-e^{-\learn\gapmin T}\bigr).
\end{align*}
Moreover, as $\learn T\to\infty$,
\begin{align*}
\expl(T)
&\;\ge\;
\vag
\;+\;
\frac{\ynum-\optnum}{\learn}\cdot\frac{\gapmin}{\gapmax}\cdot\frac{1}{\ynum}
\;+\; o_{\learn T}(1),\\
\expl(T)
&\;\le\;
\vag
\;+\;
\frac{\ynum-\optnum}{\learn}\cdot\frac{\gapmax}{\gapmin}\cdot\frac{1}{\optnum}
\;+\; o_{\learn T}(1).
\end{align*}

\item \emph{Negative Tsallis ($q\in(0,1)$).}
By \eqref{eq:appendix:tsallis-conj} and $\theta'(y)=\frac{q\,y^{q-1}-1}{q-1}$,
\begin{align*}
\underline{V_{\textnormal{start}}}
&=
\theta^*\!\Big(\theta'\!\Big(\frac{1}{\ynum}\Big)\Big)
=
\ynum^{-q},
&
\underline{V_{\textnormal{end}}}(T)
&=
\theta^*\!\Big(\theta'\!\Big(\frac{1}{\ynum}\Big)-\learn\gapmax T\Big)
=
\Big[\ynum^{1-q}+\frac{1-q}{q}\,\learn T\,\gapmax\Big]^{-\frac{q}{1-q}},\\
\overline{V_{\textnormal{start}}}
&=
\theta^*\!\Big(\theta'\!\Big(\frac{1}{\optnum}\Big)\Big)
=
\optnum^{-q},
&
\overline{V_{\textnormal{end}}}(T)
&=
\theta^*\!\Big(\theta'\!\Big(\frac{1}{\optnum}\Big)-\learn\gapmin T\Big)
=
\Big[\optnum^{1-q}+\frac{1-q}{q}\,\learn T\,\gapmin\Big]^{-\frac{q}{1-q}}.
\end{align*}
Therefore,
\begin{align*}
\underline{\Delta V}(T)
&=
\ynum^{-q}-\Big[\ynum^{1-q}+\frac{1-q}{q}\,\learn T\,\gapmax\Big]^{-\frac{q}{1-q}},\\
\overline{\Delta V}(T)
&=
\optnum^{-q}-\Big[\optnum^{1-q}+\frac{1-q}{q}\,\learn T\,\gapmin\Big]^{-\frac{q}{1-q}}.
\end{align*}
By Theorem~\ref{app thm:exploit-finite}, substituting into \eqref{eqn:surplus_decomposition}, for all $T\ge 0$,
\begin{align*}
\expl(T)
&\;\ge\;
\vag
\;+\;
\frac{\ynum-\optnum}{\learn}\cdot\frac{\gapmin}{\gapmax}\cdot
\Bigg(
\ynum^{-q}-\Big[\ynum^{1-q}+\frac{1-q}{q}\,\learn T\,\gapmax\Big]^{-\frac{q}{1-q}}
\Bigg),\\
\expl(T)
&\;\le\;
\vag
\;+\;
\frac{\ynum-\optnum}{\learn}\cdot\frac{\gapmax}{\gapmin}\cdot
\Bigg(
\optnum^{-q}-\Big[\optnum^{1-q}+\frac{1-q}{q}\,\learn T\,\gapmin\Big]^{-\frac{q}{1-q}}
\Bigg).
\end{align*}
As $\learn T\to\infty$,
\begin{align*}
\expl(T)
&\;\ge\;
\vag
\;+\;
\frac{\ynum-\optnum}{\learn}\cdot\frac{\gapmin}{\gapmax}\cdot \ynum^{-q}
\;+\; o_{\learn T}(1),\\
\expl(T)
&\;\le\;
\vag
\;+\;
\frac{\ynum-\optnum}{\learn}\cdot\frac{\gapmax}{\gapmin}\cdot \optnum^{-q}
\;+\; o_{\learn T}(1).
\end{align*}

\item \emph{Euclidean.}
$\theta'(0^+)=0$ and $\theta^*(z_i)=\frac12 z_i^2$ by \eqref{eq:appendix:euclid-conj}.
By Lemma~\ref{lem:bound-for-cont-ftrl}, suboptimal actions are eliminated after $t_\theta^* \;=\; \frac{1}{\optnum\,\learn\,\gapmin}.$
Accordingly, the endpoints take the clipped forms
\begin{align*}
\underline{V_{\textnormal{end}}}(T)
&=
\theta^*\!\Bigg(
\max\Bigg\{\theta'\!\Big(\frac{1}{\ynum}\Big)-\learn\gapmax T,\;\theta'(0^+)\Bigg\}
\Bigg)
=
\frac12\Big(\max\Big\{\frac{1}{\ynum}-\learn\gapmax T,\;0\Big\}\Big)^2,\\
\overline{V_{\textnormal{end}}}(T)
&=
\theta^*\!\Bigg(
\max\Bigg\{\theta'\!\Big(\frac{1}{\optnum}\Big)-\learn\gapmin T,\;\theta'(0^+)\Bigg\}
\Bigg)
=
\frac12\Big(\max\Big\{\frac{1}{\optnum}-\learn\gapmin T,\;0\Big\}\Big)^2.
\end{align*}
Moreover,
\begin{align*}
\underline{V_{\textnormal{start}}}
&=
\theta^*\!\Big(\theta'\!\Big(\frac{1}{\ynum}\Big)\Big)
=
\frac{1}{2\ynum^2},
&
\overline{V_{\textnormal{start}}}
&=
\theta^*\!\Big(\theta'\!\Big(\frac{1}{\optnum}\Big)\Big)
=
\frac{1}{2\optnum^2}.
\end{align*}

\emph{Case 1: $0\le T\le t_\theta^*$.}
In this regime the maximum is attained by the first term, so
\begin{align*}
\underline{V_{\textnormal{end}}}(T)
&=
\frac12\Big(\frac{1}{\ynum}-\learn\gapmax T\Big)^2,
&
\overline{V_{\textnormal{end}}}(T)
&=
\frac12\Big(\frac{1}{\optnum}-\learn\gapmin T\Big)^2,
\end{align*}
and hence
\begin{align*}
\underline{\Delta V}(T)
&=
\frac{1}{2\ynum^2}-\frac12\Big(\frac{1}{\ynum}-\learn\gapmax T\Big)^2,
&
\overline{\Delta V}(T)
&=
\frac{1}{2\optnum^2}-\frac12\Big(\frac{1}{\optnum}-\learn\gapmin T\Big)^2.
\end{align*}
By Theorem~\ref{app thm:exploit-finite} and substituting into \eqref{eqn:surplus_decomposition}, for all $0\le T\le t_\theta^*$,
\begin{align*}
\expl(T)
&\;\ge\;
\vag
\;+\;
\frac{\ynum-\optnum}{\learn}\cdot\frac{\gapmin}{\gapmax}\cdot
\Bigg(\frac{1}{2\ynum^2}-\frac12\Big(\frac{1}{\ynum}-\learn\gapmax T\Big)^2\Bigg),\\
\expl(T)
&\;\le\;
\vag
\;+\;
\frac{\ynum-\optnum}{\learn}\cdot\frac{\gapmax}{\gapmin}\cdot
\Bigg(\frac{1}{2\optnum^2}-\frac12\Big(\frac{1}{\optnum}-\learn\gapmin T\Big)^2\Bigg).
\end{align*}

\emph{Case 2: $T> t_\theta^*$.}
In this regime both endpoints clip to the boundary:
\begin{align*}
\underline{V_{\textnormal{bdry}}}=\overline{V_{\textnormal{bdry}}}=\theta^*(0)=0.
\end{align*}
By Theorem~\ref{app thm:exploit-infinite} and \eqref{eqn:surplus_decomposition}, for all $T> t_\theta^*$,
\begin{align*}
\expl(T)
&\;\ge\;
\vag
\;+\;
\frac{\ynum-\optnum}{\learn}\cdot\frac{\gapmin}{\gapmax}\cdot \frac{1}{2\ynum^2},\\
\expl(T)
&\;\le\;
\vag
\;+\;
\frac{\ynum-\optnum}{\learn}\cdot\frac{\gapmax}{\gapmin}\cdot \frac{1}{2\optnum^2}.
\end{align*}
\end{enumerate}

\subsection{Surplus in Discrete-Time FTRL Algorithm}
\label{app sec:discrete-extension}
While continuous-time dynamics offer clean geometric insights, practical algorithms like FTRL operate in discrete steps. A natural question arises: \emph{Does the inverse-rate law survive after discretization?} We show that the answer is affirmative: the discrete-time surplus is simply the continuous-time surplus plus a bounded error term. Consider the standard discrete-time FTRL update with step size $\learn$. The optimizer's cumulative surplus over $T$ rounds is defined as:
\begin{align}
    \expld(T)
    \;\defeq\;
    \sum_{t=0}^{T-1}\bigl(\payoff(\xfix,y(t))-\val\bigr)
    \;=\;
    \vag
    \;+\;
    \underbrace{\sum_{t=0}^{T-1}\sum_{i\notin \optset}\gap_i\,y_i(t)}_{\lagd},
\end{align}
where $\lagd$ is the discrete version of the cumulative lag. The following corollary establishes that $\lagd$ tracks its continuous counterpart $\lag$ closely, preserving the $\bigoh((m-k)/\learn)$ scaling.
\begin{restatable}[Discrete-Time Robustness]{corollary}{cordiscretetimeexploitation}
\label{cor:discrete-exploit-finite}
Fix a constant optimizer strategy $\xfix$ and step size $\learn>0$. The discrete-time cumulative lag $\lagd$ satisfies the continuous-time bounds of Theorem~\ref{app thm:exploit-finite} up to an additive constant:
\begin{align}
    \frac{\ynum-\optnum}{\learn} \cdot \underline{\Delta V}(T)
    \;\le\;
    \lagd
    \;\le\;
    \frac{\ynum-\optnum}{\learn} \cdot \overline{\Delta V}(T) \;+\; \frac{\ynum-\optnum}{\optnum}\gapmax.
\end{align}
\end{restatable}

\begin{remark}[Robustness of Geometry]\it
    This result confirms that the ``Inverse-Rate Law" is the dominant force in surplus of FTRL methodology. The discretization error is merely $\bigoh(1)$, whereas the geometric cost scales with $1/\learn$. Therefore, the classification of regularizers into \textbf{Steep} (infinite tail) and \textbf{Non-Steep} (finite saturation) derived in Theorem~\ref{app thm:exploit-infinite} applies essentially unchanged to the discrete FTRL algorithm.
\end{remark}

\begin{proof}
    From Lemma~\ref{lem:bound-for-cont-ftrl}, for every $j \notin \optset$ and every integer $t\ge 0$,
    \begin{align*}
            \phi\big(\theta'(1/\ynum)-\learn t\,\gap_j\big)
\ \le\ y_j(t)\ \le\
\phi\big(\theta'(1/\optnum)-\learn t\,\gap_j\big).
    \end{align*}
    Following the proof of Theorem~\ref{app thm:exploit-finite}. Recall
    $\underline g(t)$ and $\overline g(t)$ be as in \eqref{eqn:underline-g} and \eqref{eqn:overline-g}. Since $t\mapsto \theta'(1/\ynum)-\eta t\,\gap_j$ and
$t\mapsto \theta'(1/\optnum)-\eta t\,\gap_j$ are decreasing and $\phi$ is increasing, both $\underline g$ and $\overline g$ are non-increasing. Hence, for all $T\ge 1$,
\[
\sum_{t=0}^{T-1}\underline g(t)\ \ge\ \int_0^T \underline g(t)\,dt,
\qquad
\sum_{t=0}^{T-1}\overline g(t)\ \le\ \overline g(0)+\int_0^T \overline g(t)\,dt \leq \frac{\ynum-\optnum}{\optnum}\,\gapmax+\int_0^T \overline g(t)\,dt.
\]
Therefore, we have 
\begin{align*}
    \vag + \int_0^T \underline g(t)\,dt \leq \expld(T) \leq \vag + \frac{\ynum-\optnum}{\optnum}\,\gapmax + \int_0^T \overline g(t)\,dt.
\end{align*}
Then the rest of the proof follows the same steps as in Theorem~\ref{app thm:exploit-finite} to evaluate the integrals and cancel $\vag$.
\end{proof}

\section{Details omitted from Section \ref{subsec:alternating} (The Alternating Trap)}
\label{appendix:app_alternating}

Rather than committing to a fixed strategy $\xfix$, can the optimizer choose an $\xd(t)$ that achieves significantly larger reward against FTRL over $T$ rounds than any fixed $\xfix$? In this appendix we answer this question for random zero-sum games and any strongly convex separable regularizer. Specifically, we consider random games with payoff matrix $\A\in\R^{\xnum\times\ynum}$ whose entries are i.i.d.\ with $\A_{ij}\sim\mathrm{Unif}[-1,1]$.
We show that, with high probability over $\A$, there exists an optimizer strategy $\xd(t)$ that guarantees a per-round surplus $\Omega\,\!\bigl(\learn/\text{poly}(nm)\bigr)$ uniformly over all horizons $T$, against FTRL with any separable regularizer on $\ysp$ for some sufficiently small step size $\learn>0$.

We restate Theorem~\ref{thm:avg-main} in the main text.
\begin{tcolorbox}[
  enhanced,
  breakable,
  colback=white,
  colframe=black!75, 
  boxrule=0.8pt,
  arc=3mm,
  left=2mm,right=2mm,top=1mm,bottom=1mm,
  drop fuzzy shadow 
]
\begin{restatable}{theorem}{appthmavgmain}
\label{app thm:avg-main}
Let $\A\in\R^{\xnum\times\ynum}$ have i.i.d.\
$\mathrm{Unif}[-1,1]$ entries. Fix any $\delta \in (0 , 1/ \ynum)$,
$\scparam$-strongly convex separable regularizer $\reg$
with curvature bound
$M=\sup_{a\in[\delta,1-\delta]}\theta''(a)$, and
step size
\begin{align*}
    \learn \;\leq\; \frac{\scparam}{\cnorm \sup_{u\in[-1,1]^{\ynum}}\dnorm{u}}
\Bigl(\frac{1}{\ynum}-\delta\Bigr).
\end{align*}
Then with probability at least
\begin{align*}
    1-\frac{\xnum!\ynum!}{(\xnum+\ynum-1)!}-\frac{1}{\xnum\ynum},
\end{align*}
there exists an \emph{alternating trap strategy}
$\xd(t)$ such that for every horizon $T$,
\[
\rdisc{\xd(t)}{0}{T}
\;\ge\;
T\,\val \;+\; \learn T\cdot C_{\A}^{\reg},
\qquad
C_{\A}^{\reg}
\;\defeq\;
\frac{1}{2M(\xnum\ynum)^6}.
\]
\end{restatable}
{\footnotesize\emph{Remark.} This stepsize scaling is standard for a gradient method: $\eta=O(\mu/L)$ (strong convexity over smoothness).}
\end{tcolorbox}
We prove Theorem~\ref{app thm:avg-main} by following the same sequence of steps as the proof sketch in Section~\ref{subsec:alternating}.

\paragraph{Step 1: Two high-probability events on $\A$.} We condition on two events that hold with high probability under the i.i.d.\ $\mathrm{Unif}[-1, 1]$ model. The first is the non-identification-on-support event $\eventtwo$ from Assumption \ref{assump:non-identify-on-support}, whose probability is lower bounded in Lemma \ref{lem:assump:non-identify-on-support-holds-w.h.p.}. The second is the entry-gap anti-degeneracy event $\mathcal{E}_{\text{gap}}$, proved in Lemma \ref{lem:random-A-union-bound}.

\subparagraph{Non-identification-on-support event}
We rely on the following structural event, adapted from \citet[Assumption 1]{assos2024maximizing}. It ensures that, at some max-min optimizer strategy, the learner has at least two best responses that can be distinguished on the optimizer's support.
\begin{assumption}\label{assump:non-identify-on-support}
    There exists a max-min strategy $\xne$ for the optimizer and two learner best responses
    $i, j\in\brep(\xne)$ such that they do not identify on $\supp(\xne)$.
    Equivalently, there exists a support $\ell \in\supp(\xne)$ such that
    $e_\ell^\top \A e_i =\A_{\ell i}\neq\A_{\ell j} =e_\ell^\top \A e_j.$
\end{assumption}
Define the event $\eventtwo := \{\A:\ \text{Assumption~\ref{assump:non-identify-on-support} holds}\}$.
The next lemma shows that this event holds with high probability for random matrix games.
\begin{lemma} \label{lem:assump:non-identify-on-support-holds-w.h.p.}
Let $\A\in\R^{\xnum\times \ynum}$ have i.i.d.\ continuous entries $\A_{ij}\sim \mathrm{Unif}[-1,1]$.
Then
\[
\prob\!\big[\text{Assumption~\ref{assump:non-identify-on-support} holds}\big]
\;\ge\;
1-\frac{\xnum!\,\ynum!}{(\xnum+\ynum-1)!}.
\]
\end{lemma}
The assumption holds with high probability as $\xnum,\ynum$ grow.
\begin{proof}
Define the events
\begin{align*}
\Pi_{\xnum\ynum} &:= \{\A:\ \text{the game has no pure Nash equilibrium}\},\\
\mathcal D_{\xnum\ynum} &:= \{\A:\ \exists (i,j)\neq(i',j')\text{ with }\A_{ij}=\A_{i'j'}\}.
\end{align*}
\begin{claim}\label{claim:d-subset-measure-zero}
    $\prob(\mathcal D_{\xnum\ynum})=0$.
\end{claim}
\begin{proof}
    Since the entries are i.i.d.\ continuous, for any fixed distinct index pairs
    $(i,j)\neq(i',j')$ we have $\prob(\A_{ij}=\A_{i'j'})=0$.
    Taking a finite union over all such pairs yields $\prob(\mathcal D_{\xnum\ynum})=0$.
\end{proof}
\begin{claim} \label{claim:no-pure-intersect-d-subset-assump}
    $\Pi_{\xnum\ynum}\cap \mathcal D_{\xnum\ynum}^c \subseteq \eventtwo$.
\end{claim}
\begin{proof}
    Fix $A\in \Pi_{\xnum\ynum}\cap \mathcal D_{\xnum\ynum}^c$ and let $(x,y)$ be any Nash equilibrium.

We claim $y$ cannot be pure. Suppose for contradiction that $y=e_j$.
Because $A\in \mathcal D_{\xnum\ynum}^c$, the entries in column $j$ are all distinct, hence the maximizing row $i^\star=\argmax_i A_{ij}$
is unique. Therefore any best response of the row player to $y=e_j$ must be $x=e_{i^\star}$.
Since $(x,y)$ is a Nash equilibrium, $y=e_j$ must also be a best response to $x=e_{i^\star}$.
Thus $(e_{i^\star},e_j)$ is a pure Nash equilibrium, contradicting $A\in\Pi_{\xnum\ynum}$.
Hence $y$ is mixed and $|\supp(y)|\ge 2$.

Pick two distinct indices $i_1\neq i_2$ in $\supp(y)$.
Then $e_{i_1},e_{i_2}\in \brep(x)$.
Choose any $k\in \supp(x)$.
Because $A\in \mathcal D_{\xnum\ynum}^c$, all entries of $A$ are distinct, so in particular $A_{k i_1}\neq A_{k i_2}.$
Thus the two best responses $e_{i_1}$ and $e_{i_2}$ do not identify on $\supp(x)$, which is exactly
Assumption~\ref{assump:non-identify-on-support}. Hence $A\in\eventtwo$, proving
$\Pi_{\xnum\ynum}\cap \mathcal D_{\xnum\ynum}^c \subseteq \eventtwo$.
\end{proof}
A classical formula for i.i.d.\ continuous payoff matrices (see \citet{karlin1959mathematical}, p.~79)
gives
\[
\prob(\text{there exists a pure Nash equilibrium})
=\frac{\xnum!\,\ynum!}{(\xnum+\ynum-1)!},
\qquad\text{so }
\prob(\Pi_{\xnum\ynum})=1-\frac{\xnum!\,\ynum!}{(\xnum+\ynum-1)!}.
\]
Using Claim \ref{claim:no-pure-intersect-d-subset-assump} $\Pi_{\xnum\ynum}\cap \mathcal D_{\xnum\ynum}^c \subseteq \eventtwo$ and Claim \ref{claim:d-subset-measure-zero} $\prob(\mathcal D_{\xnum\ynum})=0$,
\[
\prob(\eventtwo)
\;\ge\;
\prob(\Pi_{\xnum\ynum}\cap \mathcal D_{\xnum\ynum}^c)
\;=\;
\prob(\Pi_{\xnum\ynum})-\prob(\Pi_{\xnum\ynum}\cap \mathcal D_{\xnum\ynum})
\;\ge\;
\prob(\Pi_{\xnum\ynum})-\prob(\mathcal D_{\xnum\ynum})
\;=\;
\prob(\Pi_{\xnum\ynum}),
\]
which yields the stated lower bound.
\end{proof}

\subparagraph{Anti-degeneracy gap event for entries}
We isolate an anti-degeneracy event ensuring that no two entries in the same row are ``too close''. This lets us lower bound the distinguishing gap $\A_{\ell i} - \A_{\ell j}$ uniformly whenever $\A_{\ell i} > \A_{\ell j}$.

\begin{lemma} \label{lem:random-A-union-bound}
Let $\A \in \R^{\xnum \times \ynum}$ have i.i.d.\ entries $\A_{\ell i} \sim \mathrm{Unif}[-1, 1]$. Define
\[
\gamma := \frac{2}{\xnum^2 \ynum^3}, \qquad
\eventgap := \left\{\A: \min_{\ell \in [\xnum]} \min_{\substack{i,j \in [\ynum] \\ i \ne j}} \abs{\A_{\ell i} - \A_{\ell j}} \ge \gamma \right\}.
\]
Then $\probof{\eventgap} \ge 1 - \frac{1}{\xnum\ynum}$.
\end{lemma}
\begin{proof}
Fix $\ell \in [\xnum]$ and $i< j$ and set $X := \A_{\ell i},\, Y := \A_{\ell j}$. Then $X,\, Y$ are independent $\mathrm{Unif}[-1, 1]$.
Consider the event $\mathcal{E}_{\ell ij} := \{\abs{X - Y} < \gamma\}$. Conditioning on $Y = y$,
\[
\probof{\mathcal{E}_{\ell ij} \mid Y = y}
= \prob(y-\gamma < X < y + \gamma \mid Y = y)
\le \gamma,
\]
since $X$ is uniform on an interval of length $2$ and $(y-\gamma, y+\gamma)\cap[-1,1]$ has length at most $2\gamma$. Taking expectation over $Y$ gives $\prob(\mathcal{E}_{\ell ij}) \le \gamma$.

There are at most $\xnum\binom{\ynum}{2}\le \xnum\cdot \frac{\ynum^2}{2}$ unordered triples $(\ell,i,j)$ with $i<j$. By a union bound,
\[
\probof{\exists(\ell,i<j): \abs{\A_{\ell i} - \A_{\ell j}} < \gamma}
\le \xnum\binom{\ynum}{2}\cdot \gamma
\le \xnum\cdot \frac{\ynum^2}{2}\cdot \gamma
= \frac{1}{\xnum\ynum}.
\]

Therefore, with probability at least $1 - \frac{1}{\xnum\ynum}$, for all $\ell\in[\xnum]$ and all $i\ne j$,
$\abs{\A_{\ell i} - \A_{\ell j}}\ge \gamma$, i.e., $\eventgap$ holds.
\end{proof}

Define the good event $\mathcal{E} \defeq  \eventtwo \cap \eventgap$. By Lemmas~\ref{lem:assump:non-identify-on-support-holds-w.h.p.} and \ref{lem:random-A-union-bound}, we have the union bound
\begin{align*}
    \probof{\mathcal{E}} = 1 - \probof{\mathcal{E}^c}
    \ge 1 - \probof{\mathcal{E}^c}
    \ge 1 - \probof{\eventtwo^c} - \probof{\eventgap^c} \ge 1 - \frac{\xnum!\,\ynum!}{(\xnum+\ynum-1)!} - \frac{1}{\xnum\ynum},
\end{align*}
which is the stated success probability in Theorem~\ref{app thm:avg-main}.

\paragraph{Step 2: An alternating perturbation construction.}
We assume $\A\in\mathcal{E}$. On the non-identification event $\eventtwo$, fix an optimizer max-min strategy $\xne$ and two distinct learner best responses $i, j\in\brep(\xne)$. By definition of $\eventtwo$, there exists a support $\ell\in\supp(\xne)$ such that $\A_{\ell i }\neq \A_{\ell j}.$
On the gap event $\eventgap$, this separation is uniformly bounded away from $0$ on the support: for every $\ell\in\supp(\xne)$, $\abs{\A_{\ell i}-\A_{\ell j}}\ge \gamma$.
Let's pick $\ell \in \supp(\xne)$ be a maximizer of $\xne$ over its support, i.e., $\xne_\ell \;=\; \max_{p\in\supp(\xne)} \xne_p.$
Define two perturbed strategies
\begin{equation}\label{eq:push-to-ek}
x' := (1-\xne_{\ell})\xne + \xne_{\ell} e_\ell,
\qquad
x'' := (1+\xne_{\ell})\xne - \xne_{\ell} e_\ell.
\end{equation}
Since \(x'\) is a convex combination of \(\xne\) and \(e_\ell\), we have
\(x'\in\Delta^n\). For \(x''\), note that
\(x''_\ell=(\xne_{\ell})^2\ge 0\), \(x''_i=(1+\xne_{\ell})\xne_i \ge 0\)
for \(i\neq \ell\), and \(\sum_i x''_i=1\). Hence \(x''\in\xsp\).
The next lemma shows some useful properties of $x'$ and $x''$.
\begin{lemma}\label{lem:push-to-ek-properties}
    The construction \eqref{eq:push-to-ek} has the following properties:
    \begin{enumerate}
    \item $(x'+x'')/2=\xne.$
    
    \item ${x'}^\top \A e_i > {x'}^\top \A e_j$ and ${x''}^\top \A e_i < {x''}^\top \A e_j$.
    \end{enumerate}
\end{lemma}

\begin{proof}
\noindent\textit{1. } Expanding the definition of $x'$ and $x''$,
\begin{align*}
    \frac{x'+x''}{2}
=
\frac{\bigl((1-\xne_\ell)\xne+\xne_\ell e_\ell\bigr)+\bigl((1+\xne_\ell)\xne-\xne_\ell e_\ell\bigr)}{2}
=\xne.
\end{align*}

\noindent\textit{2. } By Assumption~\ref{assump:non-identify-on-support}, $\A_{\ell i}\neq \A_{\ell j}$. WLOG, swap $i$ and $j$ so that $\A_{\ell i}>\A_{\ell j}$. Since $i, j \in \brep(\xne)$, we have $\xne^{\top} \A e_i=\xne^{\top} \A e_j$. Expanding,
\[
{x'}^\top \A e_i
=
\bigl((1-\xne_\ell)\xne+\xne_\ell e_\ell\bigr)^\top \A e_i
=
(1-\xne_\ell)\,\xne^\top \A e_i+\xne_\ell \A_{\ell i},
\]
and similarly ${x''}^\top \A e_j = (1+\xne_\ell)\,\xne^\top \A e_j-\xne_\ell \A_{\ell j}$. Subtracting and using $\xne^\top \A e_i=\xne^\top \A e_j$,
\[
{x'}^\top \A e_i-{x'}^\top \A e_j
=
\xne_\ell(\A_{\ell i}-\A_{\ell j})
>0.
\]
The same expansion yields
\[
{x''}^\top \A e_i-{x''}^\top \A e_j
=
-\xne_\ell(\A_{\ell i}-\A_{\ell j})
<0,
\]
so ${x'}^\top \A e_i>{x'}^\top \A e_j$ and ${x''}^\top \A e_i<{x''}^\top \A e_j$.
\end{proof}
We use $x'$ and $x''$ to define an alternating optimizer strategy whose time-average equals the max-min strategy $\xne$ while flipping the relative advantage between $i$ and $j$. Define
\begin{align}
    \xd(t):=
\begin{cases}
x', & t \text{ even},\\
x'', & t \text{ odd}.
\end{cases}
    \label{eq:discrete-time-xd-def-alternating}
\end{align}
Recall that learner plays FTRL with \eqref{eq:choice-map}
\begin{align*}
    y_t
=
\choicemap_{\reg}\!\Bigl(\eta \sum_{\tau=0}^{t-1} \B^\top \xd(\tau)\Bigr)
=
\argmax_{y\in\ysp}\Bigl\{\Bigl\langle \eta \sum_{\tau=0}^{t-1} \B^\top \xd(\tau),\, y \Bigr\rangle - \reg(y)\Bigr\}.
\end{align*}
Define $\vv'=\A^\top x'$ and $\vv=\A^\top \xne$. Since $(x'+x'')/2=\xne$, we have $\A^\top x'' = 2v - v'.$
For $t=2s$, using $\A=-\B$,
\[
y_{2s}   = \choicemap_{\reg}\!\bigl(-2\eta s\,v\bigr),
\qquad
y_{2s+1} = \choicemap_{\reg}\!\bigl(-2\eta s\,v-\eta v'\bigr).
\]
Over the pair of rounds $(2s,2s+1)$, the optimizer's payoff is
\[
{x'}^\top \A y_{2s}+{x''}^\top \A y_{2s+1}
=
\langle v',y_{2s}\rangle+\langle 2v-v',y_{2s+1}\rangle.
\]
Since $\xne$ is max-min, we have $\langle v,y\rangle = \xne^\top \A y \ge \val$ for all $y\in\ysp$, and in particular
$\langle 2v, y_{2s+1}\rangle \ge 2\val$. Therefore,
\begin{align}
    {x'}^\top \A y_{2s}+{x''}^\top \A y_{2s+1}
\ge
2\val+\inner*{\vv',y_{2s}} - \inner*{\vv',y_{2s+1}}.
    \label{eq:two-round-payoff-lower-bound-decomposition-discrete}
\end{align}

\paragraph{Step 3. Interpolation and a variance identity.}

Fix an integer $s\ge 0$ and define the interpolation
\begin{align}\label{eq:learner-interpolation-y(r)-def} \tag{Interpolation}
    y(r) := \choicemap_{\reg}\!\bigl(-2\learn s\,\vv - r\,\learn\, \vv'\bigr),\qquad r\in[0,1].
\end{align}
\begin{claim}\label{claim:path-is-Lipschitz}
    $y(\cdot)$ is $(\learn\|\vv'\|_*/\scparam)$-Lipschitz on $[0,1]$.
\end{claim}
\begin{proof}
    By Proposition \ref{prop: alpha strongly 1/alpha smooth}, since $\reg$ is $\scparam$-strongly convex with respect to $\norm{\cdot}$, the map $\choicemap_{\reg}$ is $1/\scparam$-Lipschitz. Hence, for any $r_1,r_2\in[0,1]$,
\[
    \norm{y(r_1)-y(r_2)}
    \le \frac{1}{\scparam}\norm{(r_1-r_2)\learn \vv'}
    = \frac{\learn\|\vv'\|_*}{\scparam}\,\abs{r_1-r_2}.
\]
Thus $y(\cdot)$ is $(\learn\|\vv'\|_*/\scparam)$-Lipschitz on $[0,1]$.
\end{proof}
Therefore, $y(\cdot)$ is absolutely continuous and differentiable for almost every $r\in[0,1]$ by Rademacher's theorem. Consequently, $\inner{\vv',y(\cdot)}$ is absolutely continuous, and the fundamental theorem of calculus yields
\begin{align}
    \inner{\vv',y_{2s}} - \inner{\vv',y_{2s+1}}
    = \inner{\vv',y(0)} - \inner{\vv',y(1)} 
    = -\int_0^1 \frac{d}{dr}\inner{\vv',y(r)}\,dr.
\end{align}

\begin{lemma}\label{lem:variance-identity}
Fix $s\ge 0$ and define $y(r)$ as in \eqref{eq:learner-interpolation-y(r)-def}. For each $r\in[0,1]$, let
\begin{align}
w_r(i)
\;\defeq\;
\begin{cases}
\displaystyle \frac{1}{\theta''(y_i(r))}, & i \in \supp(y(r)),\\
0, & \text{otherwise},
\end{cases}
\quad
W_r \;\defeq\; \sum_{i=1}^m w_r(i),
\quad
\pi_r(i) \;\defeq\; \frac{w_r(i)}{W_r}.
\end{align}
Let $I_r\sim \pi_r$ and define $Z_r\defeq v'_{I_r}$. Then, for almost every $r\in[0,1]$,
\begin{align} \label{eqn:deriv-inner-prod-variance}
\frac{d}{dr}\,\langle v', y(r)\rangle
\;=\;
-\eta\, W_r\,\Var(Z_r).
\end{align}
Moreover, for any $i, j \in \supp(y(r))$ and $i\neq j$, we have
\begin{align}\label{eq:W0Var-lb}
W_r\,\Var(Z_r)
\;\ge\;
\frac{(v'_i-v'_j)^2}{\theta''(y_i(r))+\theta''(y_j(r))}.
\end{align}
\end{lemma}

\begin{proof}
Fix $r \in [0, 1]$ at which $y(\cdot)$ is differentiable. The KKT conditions \eqref{eq:kkt_structure} for $\choicemap_{\reg}$ imply that
there exists a scalar $\kkt(r)$ such that for every $i \in \supp(y(r))$,
\begin{align*}
    \theta'(y_i(r)) = \kkt(r) - 2\learn s\, \vv_i - r\learn\, \vv_i'.
\end{align*}
Differentiating with respect to $r$ yields, for every $i \in \supp(y(r))$,
\begin{equation}\label{eqn:y(r)-derivative}
    \theta''(y_i(r))\,\frac{d}{dr}y_i(r)
    = \frac{d}{dr}\kkt(r) - \learn \vv_i',
    \qquad\text{equivalently }
    \frac{d}{dr}y_i(r)=\frac{\frac{d}{dr}\kkt(r)-\learn \vv_i'}{\theta''(y_i(r))}.
\end{equation}
For $i \notin \supp(y(r))$, we have $y_i(r)=0$, and our definition set $w_r(i)=0$. Therefore, the identity
\begin{align*}
    \frac{d}{dr}y_i(r) = w_r(i)\bigl(\frac{d}{dr}\kkt(r)-\learn \vv_i'\bigr)
\end{align*}
holds for all $i \in [\ynum]$ at this $r$.

Since $y(r)\in\ysp$ for all $r$, we have $\sum_i y_i(r)=1$ and thus $\sum_i \frac{d}{dr} y_i(r)=0$.
Using \eqref{eqn:y(r)-derivative} and $w_r(i)=1/\theta''(y_i(r))$,
\[
0=\sum_{i=1}^{\ynum} \frac{d}{dr} y_i(r)=\sum_{i=1}^{\ynum} w_r(i)\bigl(\kkt'(r)-\learn\, \vv_i'\bigr),
\]
so
\[
\frac{d}{dr} \kkt(r)
=\learn\,\frac{\sum_{i=1}^{\ynum} w_r(i)\,\vv_i'}{\sum_{i=1}^{\ynum} w_r(i)}
=\learn\sum_{i=1}^{\ynum} \pi_r(i)\,\vv_i' = \learn\,\exof{Z_r}.
\]
Substituting back gives
\begin{align*}
    \frac{d}{dr} y_i(r) = \learn\, w_r(i)\bigl(\exof{Z_r} - \vv_i'\bigr).
\end{align*}
Therefore,
\[
\frac{d}{dr}\langle \vv',y(r)\rangle
=
\sum_{i=1}^{\ynum} \vv_i' \frac{d}{dr} y_i(r)
=
\learn\sum_{i=1}^{\ynum} \vv_i' w_r(i)\bigl(\exof{Z_r}-\vv_i'\bigr)
=
-\learn\sum_{i=1}^{\ynum} w_r(i)\bigl(\vv_i'-\exof{Z_r}\bigr)^2.
\]
Since
\begin{align*}
    \sum_i w_r(i)\bigl(\vv_i'-\exof{Z_r}\bigr)^2
= W_r\sum_i \pi_r(i)\bigl(\vv_i'-\exof{Z_r}\bigr)^2
= W_r\Var(Z_r),
\end{align*}
we obtain \eqref{eqn:deriv-inner-prod-variance}.

Fix $i\neq j$ and $i, j \in \supp(y(r))$, and let 
$E_{ij}=\{I_r\in\{i,j\}\}$ with indicator $\mathbf{1}_{E_{ij}}$. By the law of total variance,
\begin{align*}\label{eq:two-point-var-lb}
\Var(Z_r) \;\ge\;
\exof{\Var(Z_r \mid \mathbf{1}_{E_{ij}})} 
\;\ge\;
\prob(E_{ij})\,\Var(Z_r\mid E_{ij}).
\end{align*}
Moreover, $\prob(E_{ij})=\pi_r(i)+\pi_r(j)$ and, conditional on $E_{ij}$, $Z_r$ takes values
$\vv_i'$ and $\vv_j'$ with probabilities $\pi_r(i)/(\pi_r(i)+\pi_r(j))$ and
$\pi_r(j)/(\pi_r(i)+\pi_r(j))$. Hence
\begin{align*}
    \Var(Z_r \mid E_{ij})
=
\frac{\pi_r(i)\pi_r(j)}{(\pi_r(i)+\pi_r(j))^2}\,(\vv_i'-\vv_j')^2,
\end{align*}
and therefore
\[
\Var(Z_r)
\;\ge\;
\frac{\pi_r(i)\,\pi_r(j)}{\pi_r(i)+\pi_r(j)}\,(\vv_i'-\vv_j')^2.
\]
Multiplying by $W_r$ and
using $\pi_r(i)=w_r(i)/W_r$ gives
\[
W_r\Var(Z_r)
\ge
W_r\cdot \frac{\pi_r(i)\pi_r(j)}{\pi_r(i)+\pi_r(j)}\,(\vv_i'-\vv_j')^2
=
\frac{w_r(i)w_r(j)}{w_r(i)+w_r(j)}\,(\vv_i'-\vv_j')^2.
\]
Substituting $w_r(i)=1/\theta''(y_i(r))$ yields \eqref{eq:W0Var-lb}. 
\end{proof}

\paragraph{Step 4. Uniform bound of curvature.}
We next show that along the interpolation path $y(r)$, the coordinates corresponding to best-response indices remain uniformly lower bounded away from $0$ for all $r\in[0,1]$. This yields a uniform upper bound on the curvature $\theta''(y_p(r))$ for all $p\in \brep(\xne)$ along the path. 
In particular, applying Lemma \ref{lem:variance-identity} to the two best-response indices $i,j\in \brep(\xne)$ from Assumption~\ref{assump:non-identify-on-support}, integrating \eqref{eqn:deriv-inner-prod-variance} over $r\in[0,1]$ and using \eqref{eq:W0Var-lb} gives
\begin{align} \label{eq:one-step-lower-have-integral}
    \inner{\vv',\y_{2s}}-\inner{\vv',\y_{2s+1}}
\ge
\learn\,(\vv'_i-\vv'_j)^2 \int_0^1 \frac{dr}{\theta''(\y_i(r))+\theta''(\y_j(r))}.
\end{align}
It therefore remains to show that the best-response coordinates $\y_i(r)$ and $\y_j(r)$ remain uniformly lower bounded away from $0$ for all $r\in[0,1]$. The next lemma establishes this bound for all $p\in \brep(\xne)$.
\begin{lemma}\label{lem:path-uniform-lb-yi-yj}
Fix an integer $s\ge 0$ and a max-min strategy $\xne$.
Assume that $\reg$ is $\scparam$-strongly convex on $\ysp$ with respect to $\norm{\cdot}$.
Define $\y(r)$ as in \eqref{eq:learner-interpolation-y(r)-def}.
Then, for all $r\in[0,1]$ and all $p\in \brep(\xne)$,
\begin{align}
y_p(r)
\;\ge\;
\max\Bigl\{0,\; \frac{1}{\ynum}-\frac{\learn}{\scparam}\,\dnorm{\vv'}\Bigr\}.
\end{align}
In particular, for any $\delta\in(0,1/\ynum)$, if
\begin{align}
0
\;<\;
\learn
\;\leq\;
\frac{\scparam}{\dnorm{\vv'}}\Bigl(\frac{1}{\ynum}-\delta\Bigr),
\end{align}
then $y_p(r)\ge \delta$ for all $r\in[0,1]$ and all $p\in \brep(\xne)$.
\end{lemma}
\begin{proof}
    Recall \eqref{eq:learner-interpolation-y(r)-def} for any fixed $s \geq 0$
\begin{align*}
    y(r) = \choicemap_{\reg}\!\big(-2\learn s\,\vv - r\,\learn \vv'\big),\qquad r\in[0,1].
\end{align*}
At even rounds $2s$, $y(0) = \choicemap_{\reg}\!\big(-2\learn s\,\vv\big)$, which is the FTRL response to a constant max-min strategy. By Lemma~\ref{lem:bound-for-cont-ftrl}, we have the uniform lower bound $y_p(0) \geq 1/\ynum$ for all $p\in \brep(\xne)$.

We have shown that $y(\cdot)$ is Lipschitz with constant $\learn\,\dnorm{\vv'}/\scparam$ from Claim~\ref{claim:path-is-Lipschitz}. Therefore, 
\begin{align*}
    \norm{y(r) - y(0)}
    = \norm{\choicemap_{\reg}\!\big(-2\learn s\, \vv\big) - \choicemap_{\reg}\!\big(-2\learn s\, \vv - r \learn \vv'\big)}
    \leq \frac{r \learn}{\scparam} \dnorm{\vv'}.
\end{align*}
Let $\cnorm
\defeq
\sup_{u\neq 0} (\norm{u}_\infty / \norm{u})$,
so that $\|u\|_\infty \le c_{\norm{\cdot}}\,\|u\|$ for all $u\in\R^m$.
For any coordinate $i \in [\ynum]$,
\[
y_i(r) \geq \posi*{y_i(0) - \norm{y(r) - y(0)}_\infty}
\geq \posi*{y_i(0) - \cnorm \cdot \frac{\learn}{\scparam}\,\dnorm{\vv'}},
\]
and therefore for all $p \in \brep(\xne)$
\[
y_p(r) \geq \frac{1}{\ynum} - \cnorm \cdot \frac{\learn}{\scparam}\,\dnorm{\vv'}.
\]
To guarantee $y_p(r)\ge \delta$ for all $r\in[0,1]$ and all $p \in \brep(\xne)$, it suffices that
\begin{align*}
    \frac{1}{\ynum}-\cnorm \cdot \frac{\learn}{\scparam}\,\dnorm{\vv'} \ge \delta
\qquad\Longleftrightarrow\qquad
\learn \le \frac{\scparam}{\cnorm \dnorm{\vv'}}\Bigl(\frac{1}{\ynum}-\delta\Bigr).
\end{align*}
\end{proof}
Moreover, since $\A\in[-1,1]^{\xnum\times \ynum}$ and $x'\in\xsp$, each coordinate of $\vv'=\A^\top x'$ lies in $[-1,1]$. Hence $\vv'\in[-1,1]^{\ynum}$ and
\begin{align*}
    \dnorm{\vv'} \le \sup_{u\in[-1,1]^{\ynum}} \dnorm{u}.
\end{align*}
Therefore, it suffices that
\begin{align*}
    \learn \le \frac{\scparam}{\cnorm \sup_{u\in[-1,1]^{\ynum}} \dnorm{u}}
\Bigl(\frac{1}{\ynum}-\delta\Bigr),
\end{align*}
which is the step size condition in Theorem \ref{thm:avg-main}.

With Lemma~\ref{lem:path-uniform-lb-yi-yj}, we can now uniformly upper bound the curvature terms $\theta''(y_i(r))$ and $\theta''(y_j(r))$ in the integral.
\begin{corollary}\label{cor:uniform-curvature-upper-bound}
    Under the same assumptions as in Lemma~\ref{lem:path-uniform-lb-yi-yj},
fix any $s>0$ and $\delta\in(0,1/\ynum)$, and choose
\begin{align*}
\learn \in \Bigl(0,\; \frac{\scparam}{\cnorm \dnorm{\vv'}}\Bigl(\frac{1}{\ynum}-\delta\Bigr)\Bigr].
\end{align*}
Let $\y(r)$ be defined in \eqref{eq:learner-interpolation-y(r)-def}. Then, for all $r\in[0,1]$ and all $p \in \brep(\xne)$,
\begin{align}\label{eq:uniform-curvature-upper-bound}
    0 \;<\;
\theta''(y_p(r))
\;\le\;
\sup_{t\in[\delta,\,1-\delta]} \theta''(t)
\;\defeq\;
M
\;<\;
\infty.
\end{align}
\end{corollary}
\begin{proof}
By Lemma~\ref{lem:path-uniform-lb-yi-yj}, for all $r\in[0,1]$ and all $p \in \brep(\xne)$, we have $y_p(r)\ge \delta$.
By Assumption~\ref{assump:non-identify-on-support}, there exists $q\in \brep(\xne)$ with $q\neq p$.
Since $\y(r)\in \ysp$, we have $y_p(r)+y_q(r)\le 1$, hence $y_p(r)\le 1-\delta$ and $y_q(r)\le 1-\delta$.
Therefore $y_p(r)\in[\delta,1-\delta]$ for all $r\in[0,1]$ and all $p \in \brep(\xne)$.
Because $\theta''$ is continuous and strongly convex on $(0,1]$, it is bounded on the compact interval $[\delta,1-\delta]$, so \eqref{eq:uniform-curvature-upper-bound} holds with
\[
0 < M = \sup_{t\in[\delta,1-\delta]} \theta''(t) < \infty.
\]
\end{proof}
Fix any $\delta\in(0,1/m)$ and choose any $\learn \in \Bigl(0,\; \frac{\alpha}{\|v'\|_*}\Bigl(\frac{1}{m}-\delta\Bigr)\Bigr]$. By Corollary~\ref{cor:uniform-curvature-upper-bound}, for all $r\in[0,1]$ and any $i,j\in \brep(\xne)$,
\[
\theta''(y_i(r))+\theta''(y_j(r)) \le 2M \qquad \text{for all } r\in[0,1].
\]
Hence,
\[
\int_0^1 \frac{dr}{\theta''(y_i(r))+\theta''(y_j(r))}
\;\ge\;
\int_0^1 \frac{dr}{2M}
\;=\;
\frac{1}{2M}.
\]
Plugging this into \eqref{eq:one-step-lower-have-integral} yields the one-step bound
\[
\inner{\vv',y_{2s}}-\inner{\vv',y_{2s+1}}
\;\ge\;
\learn\,(\vv'_i-\vv'_j)^2
\int_0^1 \frac{dr}{\theta''(y_i(r))+\theta''(y_j(r))}
\;\ge\;
\frac{\eta}{2M}\,(\vv'_i-\vv'_j)^2.
\]

\paragraph{Step 5. Sum the uniform one-step lower bound}
With our construction of $\xd(t)$ in \eqref{eq:push-to-ek}, we know that $\xne_\ell = \max_{a_\ell\in\supp(\xne)} \xne_\ell$. Since $\xne$ is a probability distribution over $\xnum$ actions, we have $\xne_\ell \ge 1 / \abs{\supp(\xne)} \geq 1/\xnum$. Since $\A \in \eventgap$, we have $\abs{\A_{\ell i} - \A_{\ell j}} \geq \gamma = \frac{2}{\xnum^2 \ynum^3}$. Since for all $i \in [\ynum]$, $\vv'_i = (\A^\top x')_i = (1 - \xne_\ell) (\A^\top \xne)_i + \xne_\ell \A_{\ell i}$ and $i, j \in \brep(\xne)$, we have
\begin{align*}
    \vv'_i - \vv'_j
    = (1 - \xne_\ell)\parens*{(\A^\top \xne)_i - (\A^\top \xne)_j} + \xne_\ell (\A_{\ell i} - \A_{\ell j})
    & = \xne_\ell (\A_{\ell i} - \A_{\ell j}) \\
    & \geq \frac{1}{\xnum} \cdot \frac{2}{\xnum^2 \ynum^3}
    = \frac{2}{\xnum^3 \ynum^3}.
\end{align*}
Combining all the pieces, for every $s \geq 0$, we have each two-round payoff bounded as
\begin{align*}
    {x'}^\top \A y_{2s}+{x''}^\top \A y_{2s+1}
    & \geq 2\val + \inner{\vv',y_{2s}} - \inner{\vv',y_{2s+1}} \\
    & \geq 2\val + \frac{\eta}{2M} (\vv'_i - \vv'_j)^2 \\
    & \geq 2\val + \frac{\eta}{2M} \cdot \parens*{\frac{2}{\xnum^3 \ynum^3}}^2 \\
    & = 2\val + \frac{2\eta}{M \xnum^6 \ynum^6}.
\end{align*}
For even $T$, the optimizer's total reward satisfies
\begin{align}
\sum_{t=0}^{T-1} \xd(t)^\top \A y_t
&=
\sum_{s=0}^{T/2-1}\Bigl({x'}^\top \A y_{2s}+{x''}^\top \A y_{2s+1}\Bigr) \nonumber\\
&\ge
\frac{T}{2}\Bigl(2\val+\frac{2\eta}{M \xnum^6 \ynum^6}\Bigr) \nonumber\\
&=
T\val+\frac{\eta T}{M \xnum^6 \ynum^6}.
\label{eq:avg-main-even-T}
\end{align}
For odd $T$, set the last-round optimizer strategy to $\xne$, which guarantees payoff at least $\val$. Then
\begin{align}
\sum_{t=0}^{T-1} \xd(t)^\top \A y_t
&=
\sum_{s=0}^{(T-1)/2-1}\Bigl({x'}^\top \A y_{2s}+{x''}^\top \A y_{2s+1}\Bigr)
+ \xne^\top \A y_{T-1} \nonumber\\
&\ge
\frac{T-1}{2}\Bigl(2\val+\frac{2\eta}{M \xnum^6 \ynum^6}\Bigr)+\val \nonumber\\
&=
T\val+\frac{\eta (T-1)}{M \xnum^6 \ynum^6}.
\label{eq:avg-main-odd-T}
\end{align}
This completes the proof of Theorem~\ref{app thm:avg-main}.

\section{Details Omitted from Section \ref{subsec:pbr} (The Price of $\varepsilon$-Best-Response)}
\label{appendix:app_pbr}

We conclude our work with a natural question about no-regret algorithmic design:
\begin{center}\emph{"How much surplus must the learner
inevitably surrender simply to identify her approximate $\varepsilon$-best response?"}\end{center}
We formalize the minimum such fee from learner’s
perspective as the Price of Best-Response.

\begin{tcolorbox}[
  enhanced,
  breakable,
  colback=white,
  colframe=black!35,
  boxrule=0.6pt,
  arc=2.5mm,
  left=1mm,right=1mm,top=0.5mm,bottom=0.5mm
]
\begin{restatable}{definition}{priceprecision}
Let $\acc(\learn, t) = \norm{y(t) - \yopt_{\xopt}}_1$ be the learner's distance from the optimal strategy, and let $\expl_{\reg}^{(\learn, t)}$ be the cumulative surplus up to time $t$. We define the \emph{cost curve} $\mathcal{C}(\varepsilon)$ as:
\begin{equation} \label{eq:price-precision}
    \mathcal{C}(\varepsilon) \;\defeq\; \inf_{\learn, t} \sup_{\xd(\cdot), \A} \Big\{ \expl_{\reg}^{(\learn, t)} \;\Big|\; \acc(\learn, t) \le \varepsilon \Big\}.\tag{Price of Best-Response (PBR)}
\end{equation}
{This quantity represents the immutable lower bound on cumulative surplus that any FTRL learner must incur to guarantee $\varepsilon$-convergence against the worst-case stabilizing environment.}
\end{restatable}
\end{tcolorbox}
\noindent While a full classification of all regularizers and optimizer choices is beyond the scope of this work, restricting our attention to the separable case yields Theorem~\ref{thm:learning-cost}.
\learningCost*
Before proving Theorem~\ref{thm:learning-cost}, we state a useful corollary of
Lemma~\ref{lem:bound-for-cont-ftrl} that characterizes the $\ell_1$-distance between the learner's response
and the best-response under a fixed optimizer strategy.
\begin{restatable}{corollary}{loneTwoSided}\label{cor:l1-two-sided}
Fix a constant optimizer strategy $\xfix$. Let $\reg$ be a separable regularizer on $\ysp$, and fix a constant step size $\learn>0$. Under FTRL against $\xfix$, the learner's response at each $t\ge 0$ satisfies
\begin{align}\label{eq:l1-exact}
\norm{y(t)-\ybrep}_1
\;=\;
2\sum_{j\notin \optset} y_j(t), \qquad \ybrep = \unif(\optset).
\end{align}
Moreover, we have the two-sided bound
\begin{align}\label{eq:l1-two-sided-coarse}
2(\ynum - \optnum)\,\phi\!\big(\theta'(1/\ynum)-\eta t\,\gapmax\big)
\;\le\;
\norm{y(t)-\ybrep}_1
\;\le\;
2(\ynum - \optnum)\,\phi\!\big(\theta'(1/\optnum)-\eta t\,\gapmin\big).
\end{align}
\end{restatable}
\begin{proof}
By Lemma \ref{lem:bound-for-cont-ftrl}(i), there exists $a_t\ge 0$ such that $y_i(t)=a_t$ for all $i\in \optset$.
Since $\sum_{i=1}^\ynum y_i(t)=1$, we have $\sum_{j\notin \optset} y_j(t)=1-\optnum a_t$, and $a_t\le 1/\optnum$.
Hence
\begin{align*}
\norm{y(t)-\ybrep}_1
&=
\sum_{i\in \optset}\bigl|a_t-\tfrac1{\optnum}\bigr|+\sum_{j\notin \optset}y_j(t)
=
\optnum\Big(\tfrac1{\optnum}-a_t\Big)+\sum_{j\notin \optset}y_j(t)
=
2\sum_{j\notin \optset}y_j(t),
\end{align*}
which proves \eqref{eq:l1-exact}. The two-sided bound \eqref{eq:l1-two-sided-coarse} follows by summing the coordinate-wise bounds in Lemma~\ref{lem:bound-for-cont-ftrl}(ii) over $j\notin\optset$, using monotonicity of $\phi$, and applying the definitions of $\gapmin$ and $\gapmax$.
\end{proof}
Then, we are ready to prove Theorem~\ref{thm:learning-cost} under the same assumption of Corollary~\ref{cor:l1-two-sided}.
\begin{proof}[Proof of Theorem~\ref{thm:learning-cost}]
Fix a matrix $\A$ in the high-probability event of Theorem~\ref{app thm:avg-main}. Let $\xne$ be a max-min strategy, and let $x',x''$ be the two optimizer strategies from \eqref{eq:push-to-ek}. Define $\xd(\cdot)$ to alternate between $x'$ and $x''$ as in \eqref{eq:discrete-time-xd-def-alternating}.

Fix $\learn>0$ and a horizon $t\ge 1$. If $t$ is even, play $\xd(0),\dots,\xd(t-1)$; if $t$ is odd, play $\xd(0),\dots,\xd(t-2)$ and then $\xne$ at time $t-1$. In both cases,
\begin{align*}
\frac{1}{t}\sum_{\tau=0}^{t-1}\xd(\tau) \;=\; \xne,
\qquad
\sum_{\tau=0}^{t-1}\A^\top \xd(\tau) \;=\; t\,\A^\top \xne.
\end{align*}
Let $\yopt\defeq \unif(\brep(\xne))$. By \eqref{eq:choice-map},
\begin{align*}
y(t)
\;=\;
\choicemap_{\reg}\!\bigl(-\learn t\,\A^\top \xne\bigr),
\qquad
\acc(\learn,t)
\;=\;
\norm{y(t)-\yopt}_1.
\end{align*}

By Theorem~\ref{app thm:avg-main}, the cumulative surplus under this $\xd(\cdot)$ satisfies
\begin{align*}
\expl_{\reg}^{(\learn,t)}
\;\ge\;
\learn(t-1)\,C_{\A}^{\reg}(\xnum,\ynum).
\end{align*}
Applying Corollary~\ref{cor:l1-two-sided} with $\xfix=\xne$ yields
\begin{align*}
\acc(\learn,t)
\;=\;
2\sum_{j\notin\brep(\xne)} y_j(t),
\qquad
\acc(\learn,t)
\;\ge\;
2(m-k)\,\phi\!\bigl(\theta'(a_t)-\learn t\,\gapmax\bigr),
\end{align*}
where $k=|\brep(\xne)|$, and $a_t$ is the common mass on $\brep(\xne)$.

Assume $\acc(\learn,t)\le \varepsilon<2(m-k)$. Then $\varepsilon/(2(m-k))\in(0,1]$, and using $\phi^{-1}(u)=\theta'(u)$ on $(0,1]$,
\begin{align*}
\learn t
\;\ge\;
\frac{\theta'(a_t)-\theta'\!\bigl(\varepsilon/(2(m-k))\bigr)}{\gapmin}.
\end{align*}
Moreover, $\acc(\learn,t)\le \varepsilon$ implies $\sum_{j\notin\brep(\xne)}y_j(t)\le \varepsilon/2$, hence
$\sum_{i\in\brep(\xne)}y_i(t)\ge 1-\varepsilon/2$, i.e., $k a_t\ge 1-\varepsilon/2$ and
\begin{align*}
a_t \;\ge\; \frac{1-\varepsilon/2}{k}.
\end{align*}
Since $\theta'$ is increasing, we obtain the necessary condition
\begin{align*}
\learn t
\;\ge\;
\frac{\theta'\!\bigl((1-\varepsilon/2)/k\bigr)-\theta'\!\bigl(\varepsilon/(2(m-k))\bigr)}{\gapmax}.
\end{align*}

Using the definition
\begin{align*}
\mathcal{C}(\varepsilon)
\;\defeq\;
\inf_{\learn,t}\ \sup_{\xd(\cdot),\,\A}
\setdef{\expl_{\reg}^{(\learn,t)}}{\acc(\learn,t)\le \varepsilon},
\end{align*}
we lower bound the inner supremum by evaluating it at the specific pair $(\xd,\A)$ constructed above:
\begin{align*}
\sup_{\xd(\cdot),\,\A}
\setdef{\expl_{\reg}^{(\learn,t)}}{\eps(\learn,t)\le\varepsilon}
\;\ge\;
\expl_{\reg}^{(\learn,t)}(\xd,\A)
\;\ge\;
\learn(t-1)\,C_{\A}^{\reg}(\xnum,\ynum).
\end{align*}
Taking $\inf_{\learn,t}$ and using $t-1\ge t/2$ for $t\ge 2$ gives
\begin{align*}
\mathcal{C}(\varepsilon)
\;\ge\;
\frac{C_{\A}^{\reg}(\xnum,\ynum)}{2\gapmax}\,
\Biggl(
\theta'\!\Bigl(\frac{1-\varepsilon/2}{k}\Bigr)
-
\theta'\!\Bigl(\frac{\varepsilon}{2(m-k)}\Bigr)
\Biggr),
\qquad
\forall \varepsilon\in(0,2(m-k)).
\end{align*}
Since $\theta'$ is continuous near $1/k$, we have $\theta'((1-\varepsilon/2)/k)=\theta'(1/k)+O(\varepsilon)$; thus, for all sufficiently small $\varepsilon$,
\begin{align*}
\mathcal{C}(\varepsilon)
\;=\;
\Omega \Bigl(\theta'(1/k)-\theta'\!\bigl(\varepsilon/(2(m-k))\bigr)\Bigr),
\end{align*}
as claimed.
\end{proof}

\section{From Full-Feedback to Bandit-Feedback}
\label{appendix:azuma}
This appendix only gives a realized-payoff concentration comparison. It treats the mixed strategies \(\{\x(t),\y(t)\}_{t=1}^T\) as arbitrary predictable processes and controls the gap between sampled payoffs and their mixed-strategy counterparts. Let $\A\in[-1,1]^{\xnum\times\ynum}$. As before, the optimizer maximizes payoffs and the learner minimizes them.

\subsection{Full-Feedback Learner Regret}

At each round $t=1,\dots,T$, the players choose mixed strategies $\x(t)\in\xsp$ and $\y(t)\in\ysp$. The learner's full-feedback regret is
\begin{align} \label{eq:regret_y} 
\gregret_{y}(T)
&\defeq
\sum_{t=1}^T \x(t)^\top \A\,\y(t)
-
\min_{\y\in\ysp}\sum_{t=1}^T \x(t)^\top \A\,\y \nonumber \notag \\
&=
\sum_{t=1}^T \x(t)^\top \A\,\y(t)
-
\min_{j\in[\ynum]}\sum_{t=1}^T \x(t)^\top \A\,e_j . \tag{Learner's Regret}
\end{align}
This compares the learner's mixed play to the best fixed learner action in hindsight, evaluated against the opponent optimizer's sequence $\{\x(t)\}_{t=1}^T$.

\subsection{Bandit-Feedback Learner Regret}
In the bandit setting, at each round $t=1,\dots,T$, the players choose mixed strategies $\x(t)\in\xsp$ and $\y(t)\in\ysp$, and then sample actions $i_t \sim \x(t)$ and $j_t \sim \y(t)$
independently of each other given the past.
The realized payoff is $\A_{i_t,j_t}$.
Define the learner's bandit-feeback/realized regret by
\begin{align}\label{eq:realized-regret}
\widehat \gregret_y(T)
&\defeq
\sum_{t=1}^T \A_{i_t,j_t}-\min_{j\in[\ynum]}\sum_{t=1}^T \A_{i_t,j}
=
\sum_{t=1}^T \A_{i_t,j_t}-\min_{j\in[\ynum]}\sum_{t=1}^T \A_{i_t,j}. \tag{Learner's Realized Regret}
\end{align}
Related work on adversarial bandits and repeated games is closely connected to our setting. \citet{odonoghue2021matrixgamesbanditfeedback} analyze repeated zero-sum matrix games with bandit-feedback and observed opponent actions, \citet{abernethy2015fighting} study regularization-based bandit methods, and \citet{lee2020biasmorehighprobabilitydatadependent} give high-probability guarantees against adaptive adversaries.
\begin{proposition}
\label{prop:bandit-to-fullinfo-regret}
Let $\A\in[-1,1]^{\xnum\times\ynum}$. At each round $t=1,\dots,T$, the optimizer and learner choose mixed strategies
$\x(t)\in\xsp$ and $\y(t)\in\ysp$ using arbitrary (possibly adaptive) mechanisms.
Define the pre-sampling history
\begin{align*}
\mathcal F_{t-1}
\defeq
\sigma\!\Big(\{\x(\tau),\y(\tau),i_\tau,j_\tau\}_{\tau=1}^{t-1},\,\x(t),\y(t)\Big),
\end{align*}
so that $\x(t)$ and $\y(t)$ are $\mathcal F_{t-1}$-measurable. At round $t$, sample $i_t \sim \x(t)$ and $j_t \sim \y(t)$ conditionally independently given $\mathcal F_{t-1}$.
Then for any $\delta\in(0,1)$, with probability at least $1-\delta$,
\begin{align}
\widehat \gregret_y(T)
\leq 
\gregret_y(T)
+\bigoh\!\left(\sqrt{T\log\frac{\ynum}{\delta}}\right).
\end{align}
\end{proposition}

\begin{proof}
    Decompose \eqref{eq:realized-regret} as
\begin{align*}
\widehat \gregret_y(T)
&=
\gregret_y(T)
+
\sum_{t=1}^T\Big(\A_{i_t,j_t}-\x(t)^\top \A\,\y(t)\Big)
+\Bigg(\min_{j\in[\ynum]}\sum_{t=1}^T \x(t)^\top \A\,e_j-\min_{j\in[\ynum]}\sum_{t=1}^T \A_{i_t,j}\Bigg) \notag\\
&\le
\gregret_y(T)
+
\underbrace{\sum_{t=1}^T\Big(\A_{i_t,j_t}-\x(t)^\top \A\,\y(t)\Big)}_{(I)}
+
\underbrace{\max_{j\in[\ynum]}\sum_{t=1}^T\Big(\x(t)^\top \A\,e_j-\A_{i_t,j}\Big)}_{(II)} .
\end{align*}
Term (I) is a martingale difference sequence with increments in $[-2,2]$. By Azuma--Hoeffding, with probability at least $1-\delta/2$,
\begin{align} \label{eq:azuma-one}
(I)\le 2\sqrt{2T\log\frac2\delta}.
\end{align}
For term (II), for each fixed $j\in[\ynum]$, the sum $\sum_{t=1}^T\big(\x(t)^\top \A e_j-\A_{i_t,j}\big)$ is a martingale difference sequence with increments in $[-2,2]$. Applying Azuma--Hoeffding and a union bound over $j\in[\ynum]$, with probability at least $1-\delta/2$,
\begin{align}
(II)\le 2\sqrt{2T\log\frac{2\ynum}{\delta}}.
\end{align}
Combining the two bounds, with probability at least $1-\delta$,
\begin{align*}
\widehat \gregret_y(T)
\le
\gregret_y(T)
+
2\sqrt{2T\log\frac2\delta}
+
2\sqrt{2T\log\frac{2\ynum}{\delta}}
=
\gregret_y(T)+\bigoh\!\left(\sqrt{T\log\frac{\ynum}{\delta}}\right).
\end{align*}
\end{proof}

\subsection{From Regret Bounds to Bandit Reward}
In the introduction, we have that in the full-feedback setting, 
\begin{align*}
    \sum_{t=1}^T \x(t)^\top \A\,\y(t) \leq T \cdot \val + \gregret_y(T).
\end{align*}
In the bandit setting, we have the following proposition relating the full-feedback reward to the realized learner regret.
\begin{proposition}\label{prop:bandit-realized-reward}
Under the same setting as Proposition \ref{prop:bandit-to-fullinfo-regret}, with probability at least $1-\delta$,
\begin{align*}
    \sum_{t=1}^T \A_{i_t,j_t}
    \;\le\;
    T\cdot \val
    \;+\;
    \gregret_y(T)
    \;+\;
    \bigoh\!\Bigl(\sqrt{T\log\!\Bigl(\frac{1}{\delta}\Bigr)}\Bigr).
\end{align*}
\end{proposition}
\begin{proof}
    Observe that
\begin{align*}
    \sum_{t=1}^T \A_{i_t,j_t}
    &=
    \sum_{t=1}^T \x(t)^\top \A\,\y(t)
    +
    \sum_{t=1}^T\Big(\A_{i_t,j_t}-\x(t)^\top \A\,\y(t)\Big) \\
    &\le
    T\cdot \val + \gregret_y(T)
    +
    \sum_{t=1}^T\Big(\A_{i_t,j_t}-\x(t)^\top \A\,\y(t)\Big) \\
    &\le
    T\cdot \val + \gregret_y(T)
    +
    2\sqrt{2T\log\frac1\delta},
\end{align*}
with probability at least $1-\delta$, where the last inequality follows from \eqref{eq:azuma-one}.
\end{proof}
\section{Experimental Evaluation}
\label{appendix:experiment}
{We design experiments to illustrate both passive and active utility extraction by the optimizer in the full-feedback setting. We further include bandit-feedback experiments to examine how these phenomena change when the learner observes only realized payoffs. Unlike the full-feedback setting, the bandit experiments do not exhibit the same clean qualitative patterns, likely due to the additional sampling noise and the variance introduced by importance-weighted feedback. Since our theoretical analysis is restricted to full feedback, we leave a rigorous understanding of the bandit-feedback setting as an open question for future work. All experiments were run on a standard laptop (MacBook M1 Pro). The full-feedback experiments took less than one minute to run, whereas the bandit-feedback experiments took around three hours for each game.}

\subsection{Setup.}
We evaluate the FTRL dynamics induced by three regularizers: negative entropy, Euclidean, and negative Tsallis with \(q=\tfrac12\). We consider both the full-feedback and bandit-feedback settings. In each setting, we study three zero-sum games and two optimizer strategies: the max-min strategy and the alternating strategy defined from Equation~\eqref{eq:push-to-ek}.

For all experiments, the learner plays
\[
\y(t)=\choicemap_{\reg}\!\bigl(\eta \HR(t)\bigr),
\qquad
\eta=T^{-\beta},
\qquad
\beta\in\{0,0.3,0.5,0.7,1.0\}.
\]

\paragraph{Full-Feedback Update.}
In the full-feedback setting, after the optimizer plays \(\x(t)\), the learner observes the full reward vector \(-\A^\top \x(t)\) and updates
\[
\HR(t+1)=\HR(t)-\A^\top \x(t).
\]

\paragraph{Bandit-Feedback Update.}
Appendix~\ref{appendix:azuma} studies realized-payoff concentration for sampled actions, treating the mixed strategies as given. The experiments below use a different, algorithmic bandit FTRL update: the learner must update its score from a one-coordinate importance-weighted estimate of the reward vector. Suppose first that the learner samples \(j_t\sim \y(t)\), while the optimizer samples \(i_t\sim \x(t)\). The learner observes only the realized reward \(-A_{i_t,j_t}\), and forms the importance-weighted estimator (see \citet[Section 11.2]{lattimore2020bandit} for background on importance-weighted estimators in the EXP3 algorithm for adversarial bandits):
\begin{align}
\label{eq:bandit-estimator-exp}
\hat r(t)=\frac{-A_{i_t,j_t}}{\y(t)_{j_t}}\,e_{j_t}.
\end{align}
This estimator is unbiased conditionally on the past
\(
\mathbb E[\hat r(t)\mid \mathcal F_{t-1}]=-\A^\top \x(t).
\)

However, without explicit exploration, the estimator is only well-defined and
unbiased on the support of \(\y(t)\). If a non-steep FTRL, like Euclidean regularizer, assigns zero probability
to an action, then that action is never sampled, and the bandit feedback no
longer provides an unbiased estimate of the full reward vector. To avoid this
issue, we also consider the explicitly explored sampling distribution
\[
    q(t)=(1-\gamma)\y(t)+\gamma u,
    \qquad
    u=\frac{1}{m}\mathbf 1,
\]
where \(\gamma\in[0,1]\) is the exploration parameter. In this case, the learner
samples \(j_t\sim q(t)\) and uses
\begin{align}
\label{eq:bandit-estimator-exp-exploration}
\hat r_\gamma(t)
=
\frac{-A_{i_t,j_t}}{q(t)_{j_t}}\,e_{j_t}.
\end{align}
Since \(q(t)_j\ge \gamma/m\) for every action \(j\), the denominator is bounded
away from zero whenever \(\gamma>0\). Moreover,
\[
    \mathbb E[\hat r_\gamma(t)\mid \mathcal F_{t-1}]
    =
    -\A^\top \x(t),
\]
so the estimator remains unbiased for the full reward vector (see \citet{auer2003nonstochastic} for details on explicit exploration in EXP3). The score vector is then updated by
\(
\HR(t+1)=\HR(t)+\hat r_\gamma(t)
\)
in the experiment.

We consider the following three zero-sum games.

\begin{description}[leftmargin=1.5cm,style=nextline]
\item[\textbf{Game 1.}]
The optimizer payoff matrix is
\[
A =
\begin{pmatrix}
1 & -1 & 1 & 2\\
-1 & 1 & 1 & 0\\
0 & 0 & -1 & 0
\end{pmatrix},
\]
and the learner payoff matrix is \(-A\).
A max-min strategy is \(x^\star=(\tfrac13,\tfrac13,\tfrac13)\), with
\(\brep(x^\star)=\{1,2\}\), game value \(0\), and
\(\nsub=4-2=2\).
For the alternating strategy, we use Equation~\eqref{eq:push-to-ek} with \(\ell=1\), yielding
\[
x'=\Bigl(\tfrac59,\tfrac29,\tfrac29\Bigr),
\qquad
x''=\Bigl(\tfrac19,\tfrac49,\tfrac49\Bigr).
\]

\item[\textbf{Game 2.}]
The Rock--Paper--Scissors payoff matrix is
\[
A_{\mathrm{RPS}}=
\begin{pmatrix}
0 & -1 & 1\\
1 & 0 & -1\\
-1 & 1 & 0
\end{pmatrix},
\]
and the learner payoff matrix is \(-A_{\mathrm{RPS}}\).
The max-min strategy is \(x^\star=(\tfrac13,\tfrac13,\tfrac13)\), with
\(\brep(x^\star)=\{1,2,3\}\), game value \(0\), and
\(\nsub=3-3=0\).
We use the same alternating construction as in Game~1.

\item[\textbf{Game 3.}]
The third game is a random zero-sum game
\(A_{\mathrm{rand}}\in[-1,1]^{30\times 30}\),
whose entries are drawn independently and uniformly from \([-1,1]\).
We compute a max-min strategy \(x^\star\) and the game value
\(\Valop(A_{\mathrm{rand}})\) via linear programming.
We then choose
\(\ell\in\supp(x^\star)\) maximizing \(x^\star_\ell\) over the support, and construct the alternating strategy using Equation~\eqref{eq:push-to-ek}.
\end{description}

\subsection{Full-Feedback Experiments}
In the full-feedback experiments, for each triple \((\beta,\reg,\A)\) and each
optimizer strategy, we run one trajectory from the initialization
\(\HR(0)=0\). The horizon \(T\) is chosen separately for each game and is
specified in the corresponding figure. We report the cumulative optimizer
surplus
\[
    \sum_{t=1}^T \x(t)^\top \A\,\y(t)-T\,\val .
\]

\paragraph{Rock--Paper--Scissors: isolating the alternating effect}
Figure~\ref{fig:full-rps} illustrates the contrast between passive and active utility extraction in Rock--Paper--Scissors. When the optimizer repeatedly plays the max-min strategy \(\xne=(1/3,1/3,1/3)\), the cumulative surplus is identically zero for all three regularizers. This is the boundary case of the Inverse-Rate Law: every learner action is a best response to \(\xne\), so \(\nsub=0\), and there is no passive surplus for the optimizer to collect.

The right column shows that the disappearance of fixed-strategy surplus does not make the game unexploitable. Under the alternating trap strategy, the optimizer alternates between two perturbations whose average is \(\xne\). Thus the value benchmark is preserved over each pair of rounds, while the perturbations reverse the payoff ordering between two learner best responses. Because the FTRL iterate cannot track this reversal instantaneously, its regularized inertia creates a positive payoff gap across each pair of rounds. Summing these gains yields the active-extraction rate predicted by Theorem~\ref{thm:avg-main}, namely \(\Omega(\eta T)\).

Since we set \(\eta=T^{-\beta}\), the active surplus scales as \(T^{1-\beta}\), so smaller \(\beta\) produces steeper growth. The small sawtooth pattern is the finite-time signature of the alternating trap. Across rows, negative entropy, negative Tsallis with \(q=1/2\), and Euclidean regularizer all exhibit positive surplus under the alternating optimizer. This supports the main message that active extraction is not specific to MWU, but is a structural phenomenon for FTRL.

\begin{figure}[H]
\centering
\includegraphics[width=0.95\textwidth]{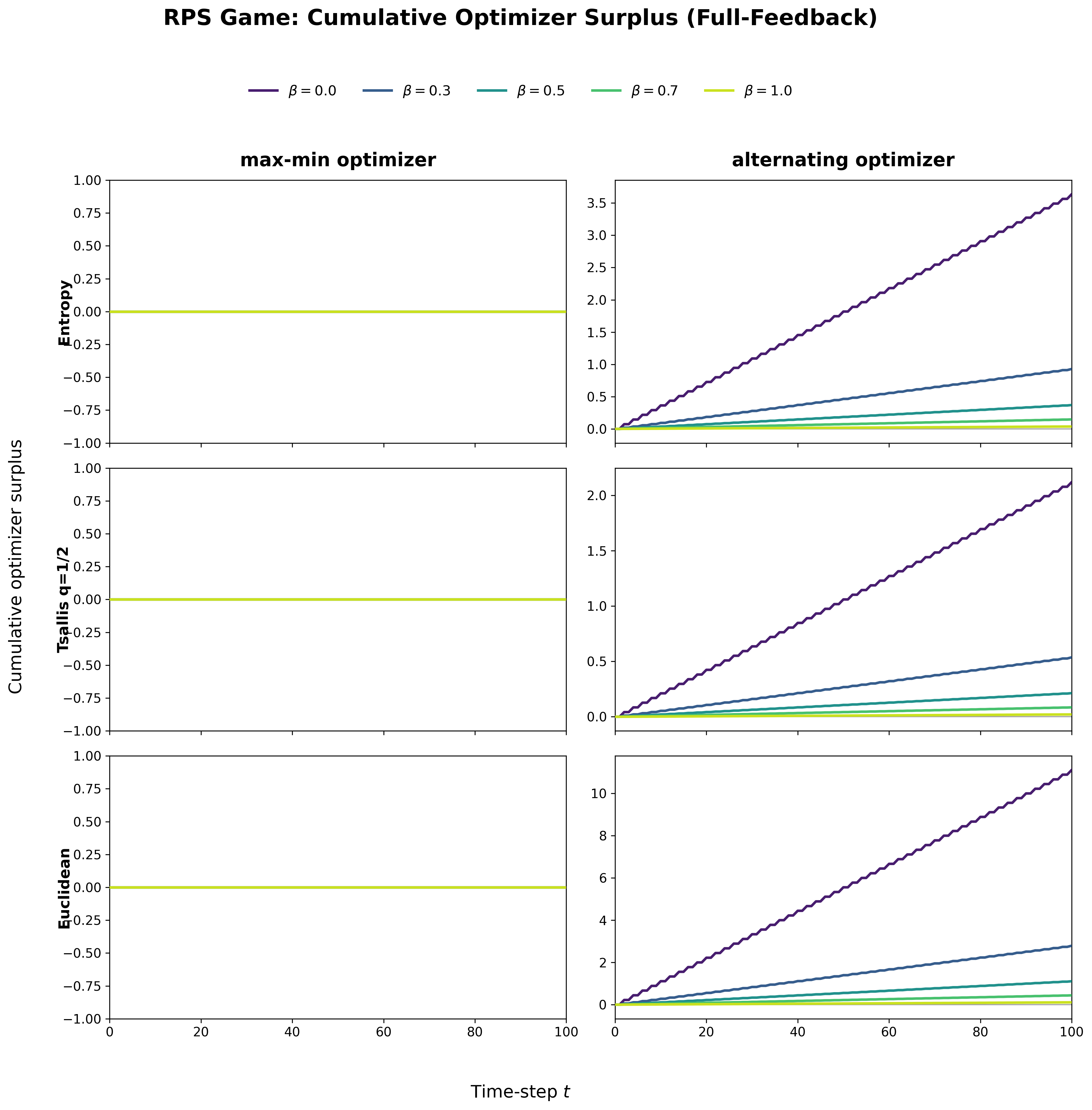}
\caption{Full-feedback results for RPS. We plot cumulative optimizer surplus under the fixed max-min optimizer (left) and alternating trap optimizer (right) for Negative Entropy, Negative Tsallis \(q=1/2\), and Euclidean regularizers. Curves use \(\eta=T^{-\beta}\), \(\beta\in\{0,0.3,0.5,0.7,1.0\}\), with \(T=100\). For the fixed strategy, all surplus curves are identically zero.
}
\label{fig:full-rps}
\end{figure}

The three regularizers exhibit the same predicted \(\Omega(\eta T)\) order, but Euclidean regularizer has a noticeably larger empirical constant at this horizon. This is consistent with the curvature dependence in the variance lower bound, Equation~\eqref{eq:W0Var-lb} of Lemma~\ref{lem:variance-identity}. Euclidean regularizer has constant curvature \(\theta''=1\), whereas negative entropy and negative Tsallis with \(q=1/2\) have curvatures \(\theta''(y_i)=1/y_i\) and \(\theta''(y_i)=\frac12 y_i^{-3/2}\), respectively. This explains why the Euclidean curve can have a larger empirical slope even though all three regularizers follow the same \(\Omega(\eta T)\) mechanism.

\paragraph{Game 1: coexistence of passive and active surplus}
\begin{figure}[H]
\centering
\includegraphics[width=0.95\textwidth]{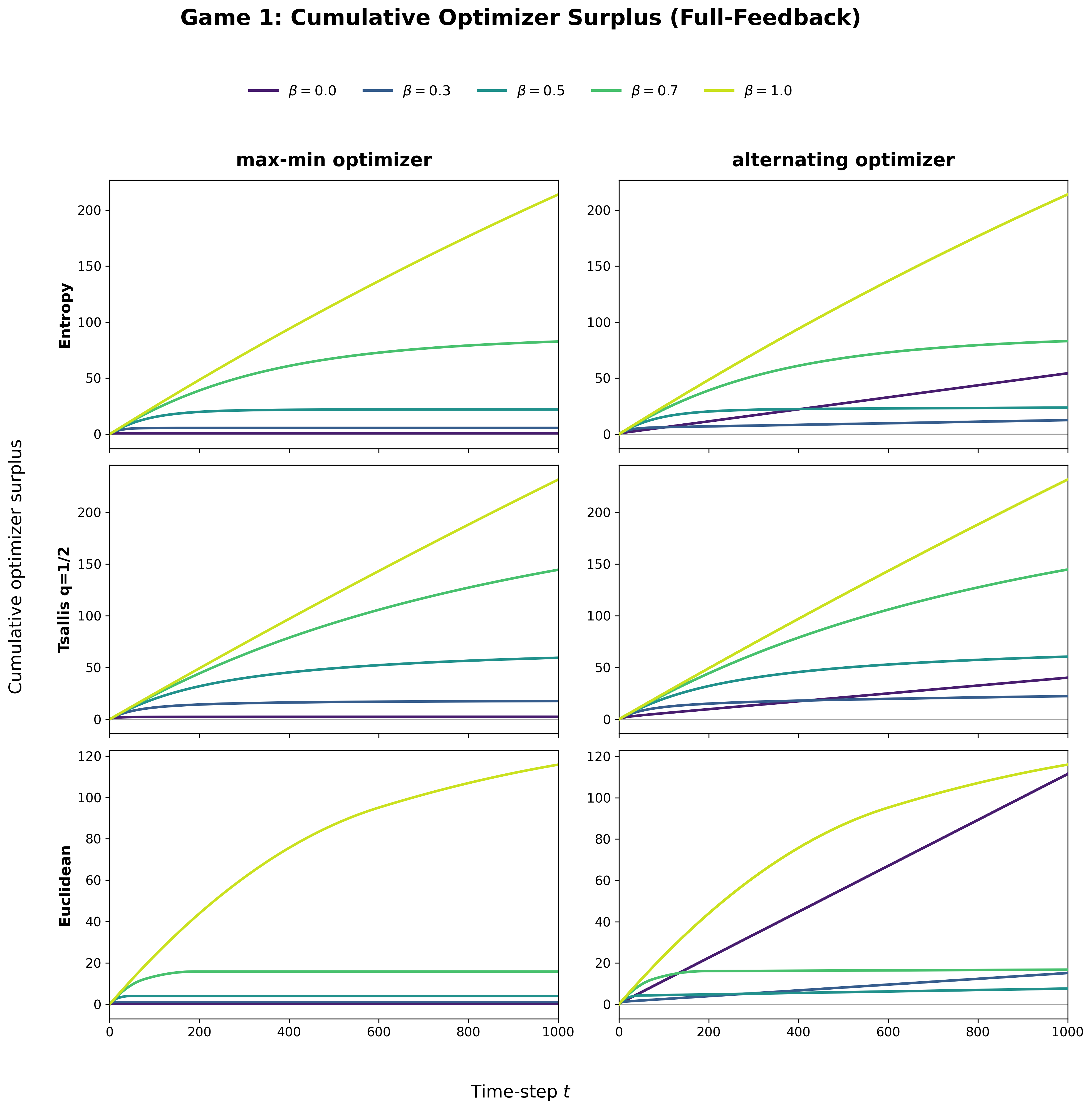}
\caption{Full-feedback results for Game~1. We plot cumulative optimizer surplus under the fixed max-min optimizer (left) and alternating trap optimizer (right) for Negative Entropy, Negative Tsallis \(q=1/2\), and Euclidean regularizers. Curves use \(\eta=T^{-\beta}\), \(\beta\in\{0,0.3,0.5,0.7,1.0\}\).}
\label{fig:full-game1}
\end{figure}
Figure~\ref{fig:full-game1} applies the same fixed-versus-alternating comparison to Game~1. Unlike RPS, the max-min optimizer strategy has \(\nsub=2\).Hence the fixed max-min optimizer already extracts
positive surplus. This is the passive mechanism in Theorem~\ref{thm:exploit-finite}: the learner initially assigns mass to
suboptimal actions, and the optimizer collects the resulting convergence lag at scale \(\Theta(\nsub/\eta)\). Since \(\eta=T^{-\beta}\), this passive contribution scales as \(T^\beta\), so larger \(\beta\) produces larger fixed-strategy surplus.

The right column uses the same alternating construction as in Figure~\ref{fig:full-rps}, so we do not repeat the mechanism here. In Game~1, this active term is added to the passive surplus already present under the fixed max-min strategy. Thus the total surplus reflects two competing scales: the passive \(\Theta(\nsub/\eta)=\Theta(T^\beta)\) term and the active \(\Omega(\eta T)=\Omega(T^{1-\beta})\) term. Small \(\beta\) emphasizes the active contribution, while large \(\beta\) emphasizes the passive inverse-rate contribution, making the fixed and alternating curves closer on the raw scale.

The regularizer-dependent constants follow the same curvature explanation discussed for RPS. Euclidean regularizer has a larger empirical constant under the alternating strategy when $\beta = 0$. Under the fixed max-min strategy, the optimizer obtains less surplus against the Euclidean regularizer, consistent with the discussion of steep and non-steep geometries and the price of best response in Section~\ref{subsec:pbr}.

Overall, Game~1 shows that the passive and active mechanisms are complementary. Suboptimal learner actions generate the passive \(\Theta(\nsub/\eta)\) surplus, while alternating play can further exploit FTRL inertia through the \(\Omega(\eta T)\) trap.

\paragraph{Random games: surplus extraction beyond toy examples}
Figure~\ref{fig:full-game3} repeats the fixed-versus-alternating comparison for a random \(30\times 30\) zero-sum game. As in Game~1, the fixed max-min
optimizer already obtains positive surplus, indicating that the random instance contains learner actions that are suboptimal against \(\xne\). The passive contribution therefore follows the Inverse-Rate Law, scaling as \(\Theta(\nsub/\eta)=\Theta(T^\beta)\).

The alternating part use the same active trap mechanism discussed above. In this random game, the alternating optimizer adds an \(\Omega(\eta T)=\Omega(T^{1-\beta})\) contribution on top of the passive surplus. Thus small \(\beta\) emphasizes the active term, while large \(\beta\) emphasizes the passive inverse-rate term.

The regularizer-dependent constants are especially visible in this higher-dimensional instance. Under the alternating strategy, Euclidean regularizer has a much larger empirical constant, particularly when \(\beta=0\). This is consistent with the curvature explanation: in a \(30\times 30\) game, the entropy and Tsallis iterates can place smaller mass on best-response actions, making \(\theta''(y_i)=1/y_i\) and \(\theta''(y_i)=\frac{1}{2}y_i^{-3/2}\) larger relative to the constant Euclidean curvature \(\theta''=1\) than in the previous games. Under the fixed max-min strategy, the optimizer obtains less surplus against Euclidean regularizer for all values of \(\beta\), consistent with the steep versus non-steep geometry discussion and the price of best response in Section~\ref{subsec:pbr}.

Overall, the random-game experiment shows that the passive and active mechanisms are not artifacts of the handcrafted examples. Passive surplus comes from suboptimal learner actions, active surplus comes from the alternating trap, and the regularizers mainly differ through geometry-dependent constants.

\begin{figure}[H]
\centering
\includegraphics[width=0.95\textwidth]{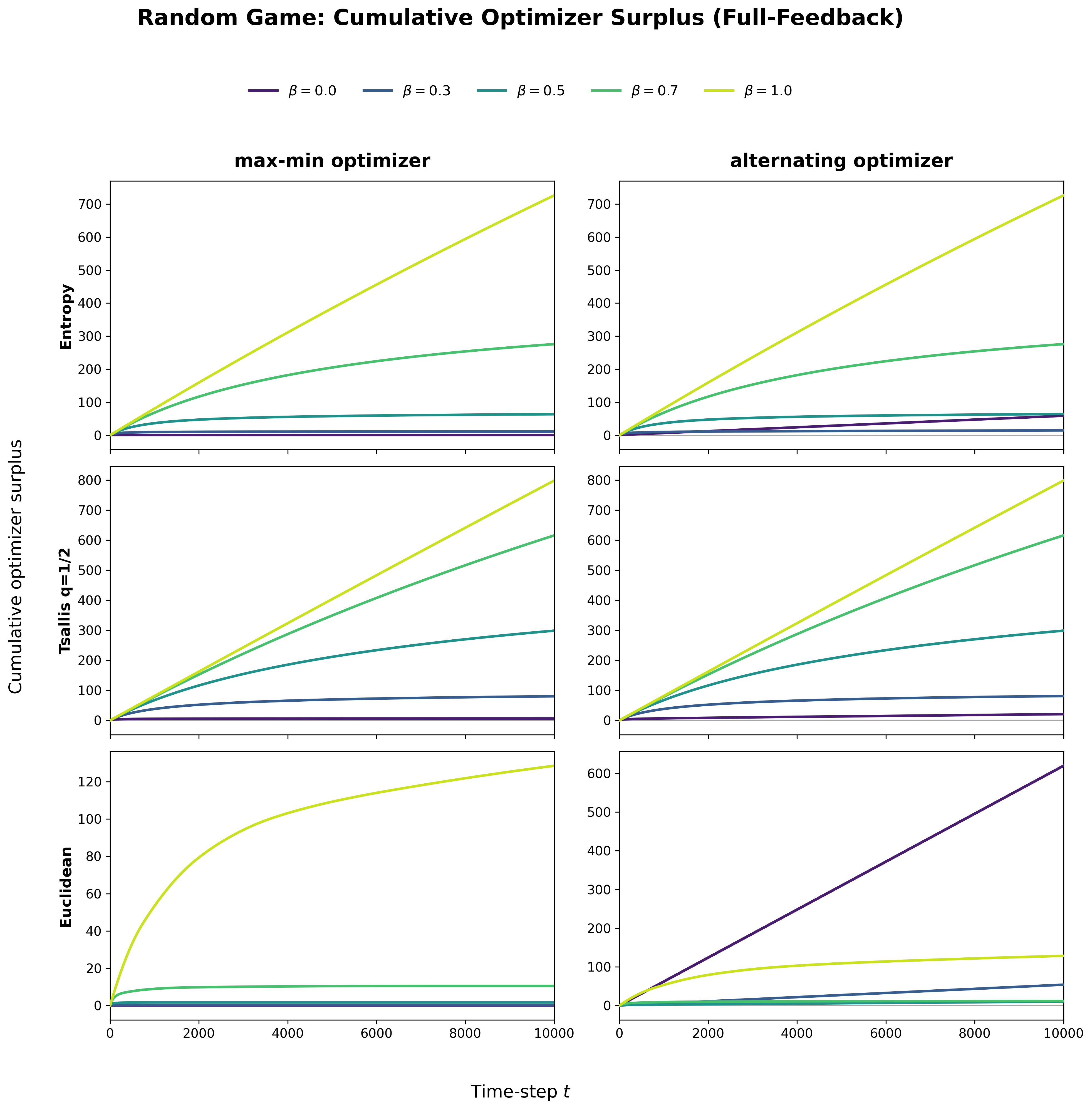}
\caption{Full-feedback results for random Game~3. We plot cumulative optimizer surplus under the fixed max-min optimizer (left) and alternating trap optimizer (right) for Negative Entropy, Negative Tsallis \(q=1/2\), and Euclidean regularizers. Curves use \(\eta=T^{-\beta}\), \(\beta\in\{0,0.3,0.5,0.7,1.0\}\), with \(T=10000\). 
}
\label{fig:full-game3}
\end{figure}

\subsection{Bandit-Feedback Experiments}
In the bandit experiments with explicit exploration, for each triple \((\beta,\reg,\A, \gamma)\) and each optimizer strategy, we run \(500\) independent trials with horizon \(T=5000\), again with initialization \(\HR(0)=0\). For each trial, we record the cumulative realized optimizer surplus
\[
\sum_{t=1}^T A_{i_t,j_t}-T\,\val.
\]
Across trials, we summarize the distribution at each time \(t\) by the median together with the \(10\)th and \(90\)th percentiles because of realized action draws and importance-weighted feedback. 
We set \(\gamma\approx 0.02, 0.025, 0.11\) for the three games below, respectively, following the standard EXP3 tuning in \citet[Corollary~3.2]{auer2003nonstochastic}
\[
    \gamma
    =
    \sqrt{\frac{m\ln m}{(e-1)T}},
\]
where \(m\) is the number of learner actions and \(T\) is the horizon. We use the same exploration rate for all regularizers as a standard variance-control choice for the importance-weighted bandit estimator.

\paragraph{Rock--Paper--Scissors under bandit-feedback}
Figure~\ref{fig:bandit-rps} shows that, under bandit feedback, the fixed max-min optimizer produces no visible surplus. This mirrors the full-feedback RPS behavior in the left column of Figure~\ref{fig:full-rps}, where the fixed-strategy surplus is identically zero. The wide percentile bands reflect the additional randomness from realized action draws and importance-weighted feedback, rather than passive utility extraction.

The alternating trap optimizer exhibits a less clean pattern under bandit-feedback than in the full-feedback experiment. In contrast to the right column of Figure~\ref{fig:full-rps}, the median realized surplus remains close to zero for all three regularizers and all tested step sizes, while the percentile bands are large. A rigorous analysis of whether utility extraction persists in expectation is left for future work.

\paragraph{Game 1 under bandit-feedback}
We next repeat the Game~1 experiment under bandit-feedback. The median realized surplus is clearly positive under both optimizer strategies, and its dependence on \(\beta\) mainly follows the inverse-rate law already visible in the full-feedback experiment: larger \(\beta\) produces larger median surplus. This is consistent with Game~1 having suboptimal learner actions against \(\xne\), which already generate passive utility extraction under the fixed max-min optimizer.

The \(10\)th--\(90\)th percentile bands are wider for smaller \(\beta\), indicating that sampling noise has a larger effect when the step size is larger. Under the alternating optimizer, the median curves remain broadly similar to the fixed max-min curves, rather than showing the clean additional separation seen in the right column of Figure~\ref{fig:full-game1}. Thus, Figure~\ref{fig:bandit-game1} still displays the passive utility extraction, but it does not provide a clean additional active extraction under alternating play.

\paragraph{Random game under bandit-feedback}
Figure~\ref{fig:bandit-random} shows the same qualitative behavior as Game~1 under bandit-feedback. The median realized surplus still follows the inverse-rate law. However, the alternating optimizer does not produce a clean additional extraction.

Overall, the full-feedback experiments match the theoretical predictions: fixed max-min play exhibits the inverse-rate passive surplus when suboptimal learner actions are present, while alternating play produces the additional active-extraction effect. In the bandit-feedback experiments, the median realized surplus still reflects the inverse-rate law in games with suboptimal actions, but the alternating effect is not cleanly visible at the median level. We also note that the bandit surplus curves include the cost of explicit exploration, which is absent from the full-feedback experiments and may further affect the observed trajectories. Whether active utility extraction persists for the expected realized reward under bandit-feedback remains an important direction for future work.

\begin{figure}[H]
\centering
\includegraphics[width=\textwidth]{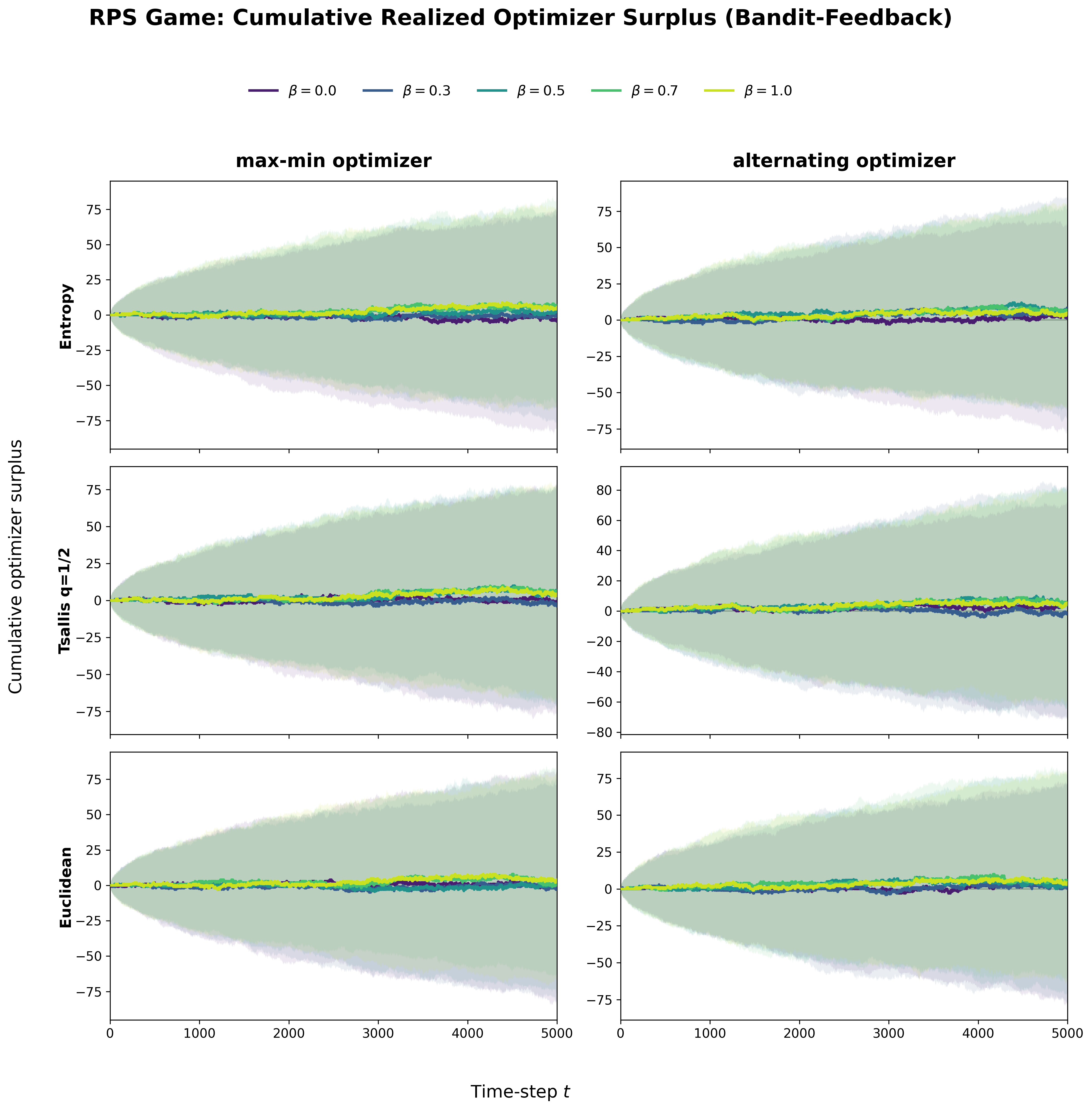}
\caption{
    Bandit-feedback results for RPS. We plot cumulative realized optimizer surplus under the fixed max-min optimizer (left) and the alternating trap optimizer (right) for Negative Entropy, Negative Tsallis with \(q=1/2\), and Euclidean regularizer (top to bottom). Curves correspond to step sizes \(\eta=T^{-\beta}\), \(\beta\in\{0,0.3,0.5,0.7,1.0\}\), with \(T=5000\) and \(\gamma=0.02\). For each setting, we run \(500\) independent trials and plot the median together with the \(10\)th and \(90\)th percentiles.
}
\label{fig:bandit-rps}
\end{figure}

\begin{figure}[H]
\centering
\includegraphics[width=\textwidth]{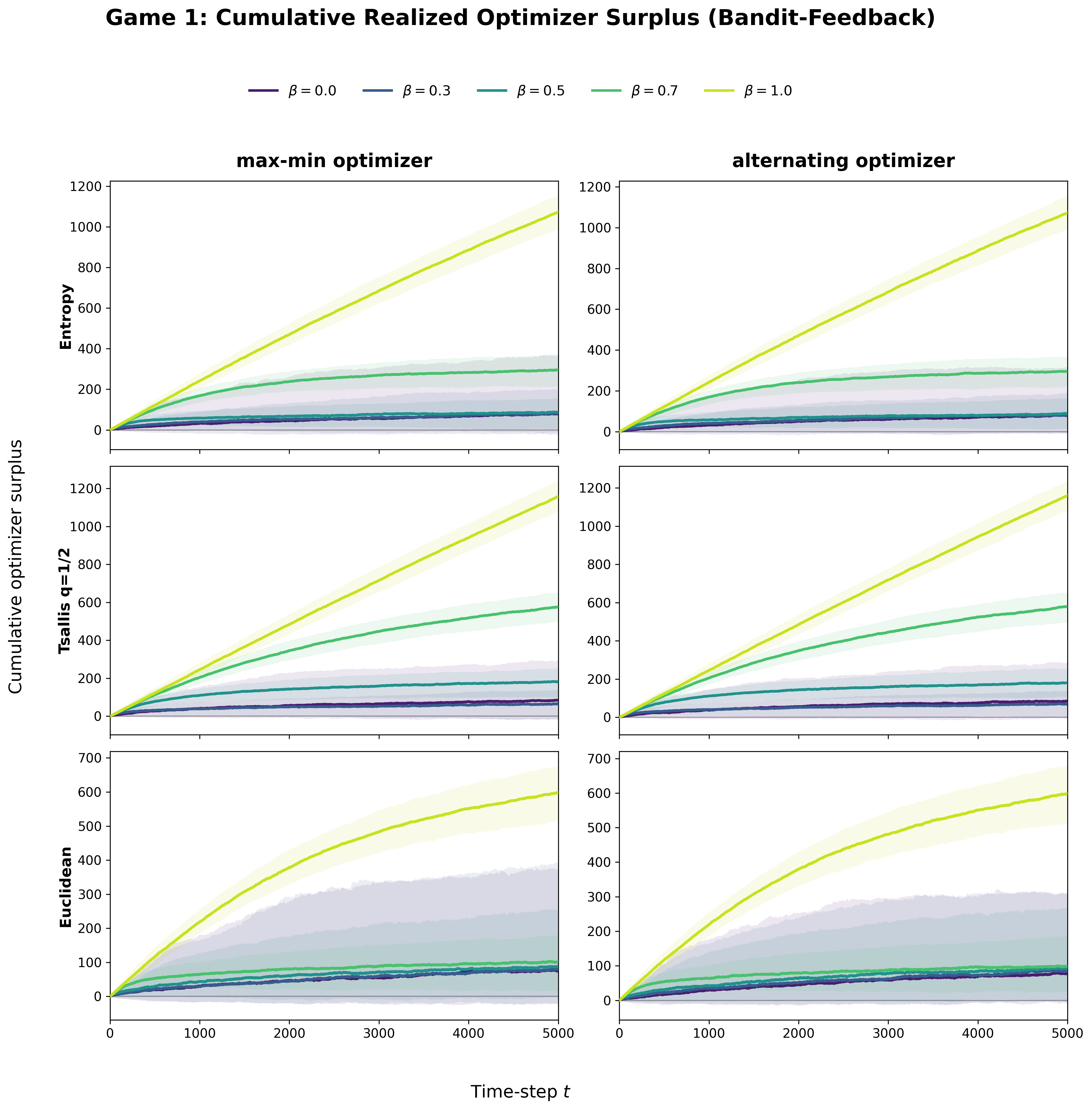}
\caption{Bandit-feedback results for Game~1. We plot the cumulative realized
optimizer surplus under the fixed max-min optimizer (left) and the alternating
trap optimizer (right) for negative entropy, negative Tsallis with \(q=1/2\),
and Euclidean regularizer (top to bottom). Curves correspond to step sizes
\(\eta=T^{-\beta}\), \(\beta\in\{0,0.3,0.5,0.7,1.0\}\), with \(T=5000\) and \(\gamma=0.025\). For
each parameter setting we run \(500\) independent trials and plot the median
together with the \(10\)th and \(90\)th percentiles. 
}
\label{fig:bandit-game1}
\end{figure}

\begin{figure}[H]
\centering
\includegraphics[width=\textwidth]{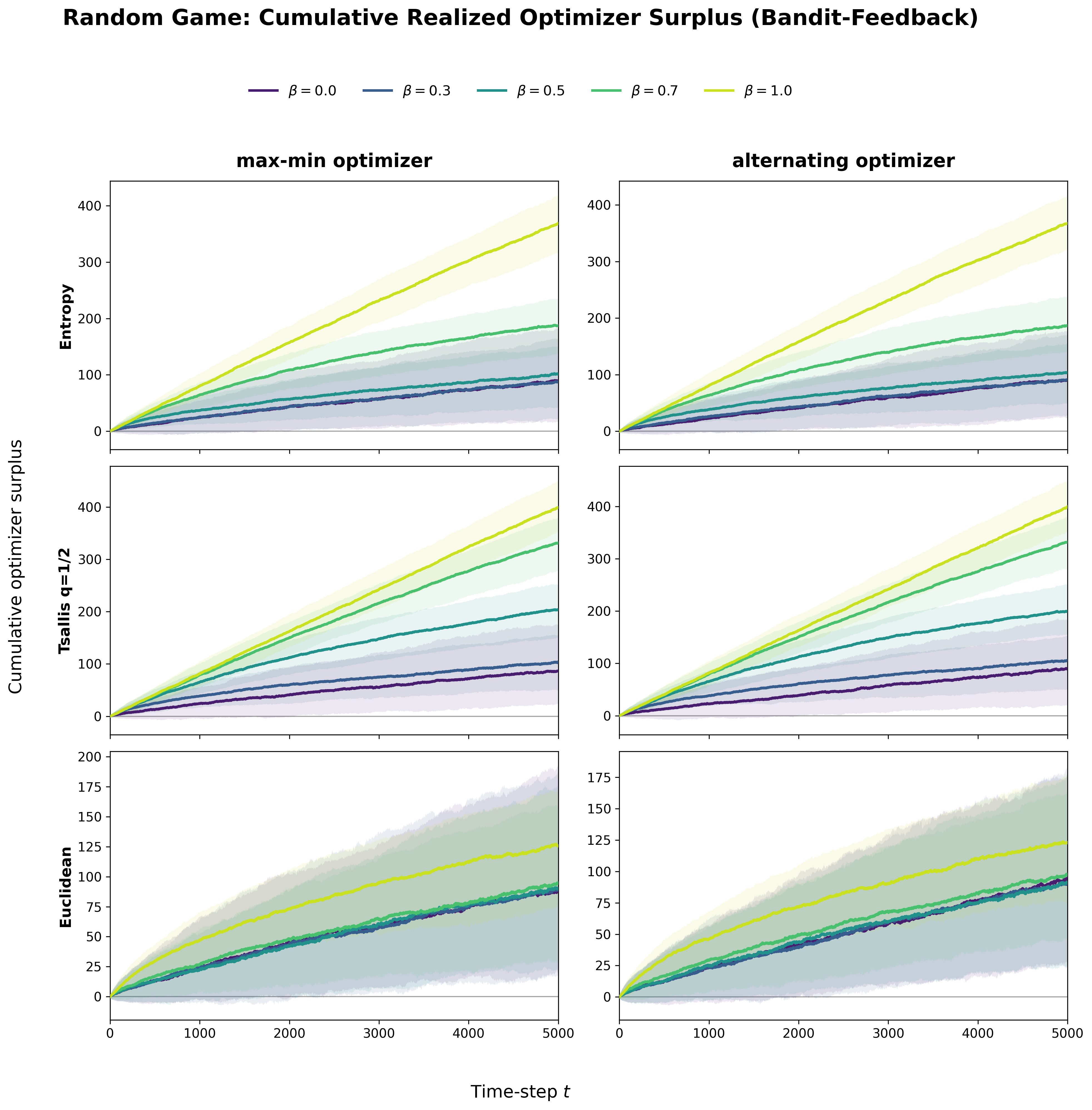}
\caption{Bandit-feedback results for random game. We plot cumulative realized optimizer surplus under the fixed max-min optimizer (left) and the alternating trap optimizer (right) for negative entropy, negative Tsallis with \(q=1/2\), and Euclidean regularizer (top to bottom). Curves correspond to step sizes \(\eta=T^{-\beta}\), \(\beta\in\{0,0.3,0.5,0.7,1.0\}\), with \(T=5000\) and \(\gamma=0.11\). For each parameter setting we run \(500\) independent trials and plot the median together with the \(10\)th and \(90\)th percentiles.}
\label{fig:bandit-random}
\end{figure}

\section{Future Directions: An Extended Discussion}
\label{appendix:future-work}

The framework developed in this paper — pricing the FTRL learner's regret as a transferable resource extracted by a clairvoyant optimizer, and reading the regularizer's geometry as the mediator of that transfer — opens several natural research directions. We organize them around the three themes flagged in the main-text epilogue: the learner's flexibility, partial feedback, and structure beyond full knowledge. For each, we describe the question, what we believe the right form of the answer should be, and which technical pieces of the present work do (and do not) carry over.

\subsection{The Learner's Flexibility}
\label{appendix:future-flexibility}

Our results assume two simplifying choices on the learner's side: a constant step size $\eta$, and an unconstrained ability to revise the mixed strategy from one round to the next. Relaxing either assumption changes the picture in a structured way, and the two relaxations interact in ways we find conceptually appealing.

\paragraph{Adaptive and decaying schedules.}
The inverse-rate law $\Theta(N_{\mathrm{sub}}/\eta)$ and the alternating-trap rate $\Omega(\eta T / \mathrm{poly}(n,m))$ balance precisely at $\eta = \Theta(T^{-1/2})$, recovering the classical $\sqrt{T}$ benchmark — but as a \emph{consequence} of the static–dynamic duality rather than as a free design choice. For an adaptive or vanishing schedule $\{\eta_t\}_{t=1}^T$, the natural conjecture is that $\eta$ should be replaced by an effective rate derived from a \emph{cumulative learning-rate budget}. The cleanest candidate is $H_T := \sum_{t=1}^T \eta_t$, yielding a static-style bound of the form $\Omega(N_{\mathrm{sub}} \cdot T / H_T)$ in terms of the time-averaged rate $\overline{\eta}_T = H_T/T$, and a trap-style bound of the form $\Omega(H_T / \mathrm{poly}(n,m))$. Under the AdaGrad-style schedule $\eta_t = \Theta(1/\sqrt{t})$, $H_T = \Theta(\sqrt{T})$, again recovering the benchmark. Establishing this rigorously requires extending our continuous-time decomposition (Theorem~\ref{thm:main-reduction}) to time-inhomogeneous mirror flows, where the standard energy arguments acquire an additional drift term proportional to $\dot\eta_t$. The technical question is whether this drift is benign — preserving the inverse-rate structure with $H_T$ replacing $\eta T$ — or whether it introduces a genuinely new failure mode for steep regularizers, whose mirror map is most sensitive precisely in the regimes where $\eta_t$ shrinks fastest.

\paragraph{Switching-constrained and lazy learners.}
A more realistic class of learners cannot revise their mixed strategy arbitrarily fast — either because of an explicit \emph{switching cost} per unit movement in $y$, or because the update rule itself enforces a small step (e.g., projected lazy mirror descent, or FTRL with proximal regularization). This constraint cuts both ways. On one hand, the alternating trap of Theorem~\ref{thm:avg-main} relies on the learner's response cycling between two best responses within a constant number of rounds; if switching is penalized, the trap is weakened or destroyed entirely, and the dynamic surplus drops below $\Omega(\eta T)$. On the other hand, the inverse-rate law \emph{sharpens}: a learner that cannot abandon suboptimal actions quickly accrues a longer transient on each one. We conjecture the static surplus grows from $\Theta(N_{\mathrm{sub}}/\eta)$ to $\Theta(N_{\mathrm{sub}}/(\eta \cdot \kappa))$, where $\kappa \in (0,1]$ is a switching-cost parameter measuring the maximum per-round movement allowed. The borderline case — where the learner's switching budget is comparable to the optimizer's alternation period — is the regime in which we expect the most interesting phase transition, and possibly a third extraction mechanism that interpolates between static and alternating play.

\paragraph{Joint adaptivity.}
Combining the two relaxations above is the most realistic — and most challenging — variant. A learner who tunes $\eta_t$ adaptively while paying a switching cost effectively chooses a trajectory in a two-dimensional design space, and the optimizer's optimal extraction strategy presumably depends on which axis dominates. We expect a clean trade-off curve here, parameterized by the ratio of the switching budget to the cumulative learning-rate budget $H_T$.

\subsection{Partial Feedback}
\label{appendix:future-bandit}

Appendix~\ref{appendix:azuma} extends our extraction laws to bandit feedback under \emph{multiplicatively suboptimal} step-size choices, via an Azuma-style concentration argument.
The behaviour at the optimal step size — say, EXP3 with $\eta_T = \Theta(\sqrt{\log n / (nT)})$ — is genuinely unsettled, and the experiments in Appendix~\ref{appendix:experiment} are themselves ambiguous on this point: under the alternating-trap optimizer, the median realized surplus collapses toward zero across all three regularizers and all tested step sizes, while the percentile bands remain wide. Two competing hypotheses are consistent with this evidence.

\paragraph{The variance hypothesis.}
The importance-weighting noise of $\hat r_t = r_t \cdot \mathbf{1}[i_t = i] / y_t(i)$ inflates the second moment of the learner's gradient estimate by a factor of $1/y_{\min}$, where $y_{\min}$ is the minimum probability the learner places on any action. Under this hypothesis, the same $o(T)$ leak persists in expectation, but is masked at the median level by variance that scales with $1/y_{\min}$. The leak should re-emerge over horizons $T \gtrsim 1/y_{\min}^2$, i.e., once the optimizer has had enough rounds for the law of large numbers to dominate the per-round noise. A clean test would be to run the bandit experiment at exponentially longer horizons and look for a crossover at which the median surplus separates from zero.

\paragraph{The cancellation hypothesis.}
Alternatively, the unbiasedness of $\hat g_t$ together with the learner's optimally-tuned regularization actually averages the leak away in expectation, and the true bandit surplus is $o(\sqrt{T})$ or smaller. This would mark a genuine separation between the full-feedback and bandit settings: the geometric vulnerability of FTRL would be a property of \emph{deterministic} gradient access, not of the regularizer alone. Distinguishing the two hypotheses requires either a tight Azuma-style upper bound matching our lower bound at the optimal rate, or a constructive optimizer strategy that survives unbiased noise — both of which appear to demand new technical tools beyond the concentration arguments developed here. Our prior is mildly toward the variance hypothesis, since the inverse-rate mechanism is fundamentally about transient suboptimality, and unbiasedness alone does not eliminate transients; but we view the question as open.

\paragraph{Bandits with side information.}
A softer variant, intermediate between full and bandit feedback, is the case where the learner observes the realized payoff plus a coarse signal about unplayed actions — for example, a noisy unbiased estimate of the full gradient with bounded variance. We expect our extraction rates to interpolate smoothly between the two extremes as a function of the signal's variance, and this regime may be technically more tractable than pure bandit feedback while still capturing the realistic case of an analyst who has access to historical aggregates but not full counterfactuals.

\subsection{Structure Beyond Full Knowledge}
\label{appendix:future-structure}

The third theme groups together two extensions that look superficially different — going beyond two-player zero-sum games, and weakening the optimizer's payoff knowledge — but which share a common technical core: in both, the optimizer must operate without one of the structural primitives our current proofs rely on.

\paragraph{Fine grained analysis beyond two-player zero-sum.}
Our analysis leans on three structural features of the zero-sum bilinear setting: the reward sandwich $T \cdot \gvalue
\;\le\;
\mathrm{Rew}_T
\;\le\;
T \cdot \gvalue + \gregret(T)$ that pins the surplus into the $o(T)$ regime, the bilinear form $x^\top A y$ that decouples the two players' contributions, and the simplex geometry of mixed strategies that supports our regularizer dichotomy. Each of these breaks differently in the natural extensions. \emph{General-sum games} preserve the simplex geometry but lose the sandwich, so the right benchmark is the Stackelberg value $V^{\mathrm{Stack}}$ rather than the minimax value $V^*$; the inverse-rate law must be reformulated as extraction \emph{above the Stackelberg payoff}, and the alternating trap may yield linear surplus when the two values separate. \emph{Network and polymatrix games} preserve bilinearity locally but introduce cross-coupling between learners, opening the possibility of \emph{indirect} extraction, where the optimizer steers one learner through another — a phenomenon with no analogue in the bilateral setting. \emph{Extensive-form games} replace the simplex with the sequence-form polytope, and the regularizer's geometry interacts with the tree structure in ways not captured by our separability assumption. We conjecture the inverse-rate law survives in network games (with $N_{\mathrm{sub}}$ summed across players, weighted by the network's spectral structure) and in extensive-form games under dilated regularizers, but with constants that depend on the depth of the tree and possibly degrade exponentially in pathological branching patterns.

\paragraph{The partially-informed optimizer.}
Perhaps the most operationally interesting variant: in realistic learning markets, the optimizer rarely has full knowledge of the payoff matrix $A$, but often does know — or can fingerprint from a few rounds of observation — the learner's \emph{algorithmic profile}: the regularizer family, the step-size schedule, and the feedback model. We conjecture that a substantial fraction of our surplus is recoverable from the profile alone, with $A$ entering only as a lower-order term. Formally, let $S^\star(A, \mathcal{L})$ denote the optimal extractable surplus under full payoff knowledge against a learner $\mathcal{L}$, and define the profile-only worst-case surplus
\[
S^\dagger(\mathcal{L}) \;:=\; \sup_{\sigma}\; \inf_{A}\; \mathbb{E}\bigl[\, S(\sigma; A, \mathcal{L}) \,\bigr],
\]
where $\sigma$ ranges over optimizer strategies that depend only on the learner's profile and the observed history of play. The conjecture is that $S^\dagger(\mathcal{L}) \ge (1 - o(1)) \cdot S^\star(A, \mathcal{L})$ for a $1-o(1)$ fraction of payoff matrices drawn from a natural prior, with the gap closing as the optimizer accumulates observations. A constructive version would proceed in two phases: a short \emph{profiling} phase in which the optimizer estimates the active best-response correspondence from the learner's trajectory, followed by an \emph{extraction} phase in which the alternating trap of Section~\ref{section: toolbox} is deployed once two non-dominated best responses have been identified. A clean such reduction — say, $\widetilde O(n^2 m)$ rounds of profile-only exploration sufficient to recover the trap rate up to log factors — would unify our static and dynamic results into a single anytime algorithm, and would represent, in our view, the natural endpoint of the program initiated here.

\paragraph{Profile uncertainty.}
A symmetric variant is the case where the optimizer knows $A$ but is uncertain about the learner's profile — for example, the regularizer family or the step size is drawn from a small set and revealed only through play. Here we expect a robust extraction rate that interpolates between the worst-case and best-case profiles, with the optimizer paying a logarithmic exploration overhead to identify the active profile before deploying the matching trap. This is technically closer to a partial-monitoring problem than to the full-knowledge analysis of the present work, and we expect the right tools to come from the literature on online learning against unknown opponents rather than from FTRL theory itself.

\subsection{A Closing Thought}
\label{appendix:future-closing}

Across all three themes, a common pattern emerges: the inverse-rate law and the alternating trap are robust under \emph{relaxations of the learner's side} (slower updates, bandit feedback, structurally constrained games), but require new ideas under \emph{relaxations of the optimizer's side} (unknown payoffs, unknown profiles). This asymmetry is, we think, the right one to expect — the optimizer's leverage in our framework comes from its informational advantage, and chipping away at that advantage is precisely what makes the questions hard. We hope the framework here proves a useful starting point for chipping back.

\end{document}